\documentclass{article}
\usepackage{graphicx} 

\usepackage{changepage}

\usepackage{authblk}

\usepackage{float}

\usepackage[dvipsnames]{xcolor}
\usepackage{xcolor,todonotes}
\usepackage{enumitem}
\usepackage{multicol}

\usepackage{tikz}
\tikzstyle{edge}=[shorten <=2pt, shorten >=2pt, >=stealth, line width=1.5pt]
\tikzstyle{arc}=[->, shorten <=2pt, shorten >=2pt, >=stealth, line width=1pt]
\tikzstyle{vertex}=[circle, fill=white, draw, minimum size=6pt, inner sep=0pt,
                    outer sep=0pt]

\usepackage{amsmath}
\usepackage{amssymb,hyperref}
\usepackage{bm} 
\usepackage{comment}

\usepackage{amsthm,geometry,cite}
\usepackage[linesnumberedhidden]{algorithm2e}

\usepackage{pdfcomment}
\usepackage{todonotes}

\newtheorem{theorem}{Theorem}

\newtheorem{lemma}[theorem]{Lemma}
\newtheorem{corollary}[theorem]{Corollary}
\newtheorem{proposition}[theorem]{Proposition}
\newtheorem{observation}[theorem]{Observation}

\newtheorem{problem}[theorem]{Problem}

\theoremstyle{remark}
\newtheorem{remark}{Remark}

\theoremstyle{definition}
\newtheorem{example}{Example}

\newcommand{\blue}[1]{\textcolor{blue}{#1}}

\DeclareMathOperator{\Cyc}{Cyc}
\DeclareMathOperator{\Sig}{Sig}

\newcommand{\bA}{\mathbb A}
\newcommand{\bB}{\mathbb B}
\newcommand{\bC}{\mathbb C}
\newcommand{\bD}{\mathbb D}
\newcommand{\bH}{\mathbb H}
\newcommand{\bG}{\mathbb G}
\newcommand{\bK}{\mathbb K}

\newcommand{\bM}{\mathbb M}
\newcommand{\bP}{\mathbb P}
\DeclareMathOperator{\CSP}{CSP}
\DeclareMathOperator{\NP}{NP}
\DeclareMathOperator{\PO}{P}
\DeclareMathOperator{\Pol}{Pol}

\DeclareMathOperator{\cP}{P}

\DeclareMathOperator{\Alt}{Alt}

\title{On the complexity of Sandwich Problems for $M$-partitions}

\author[1]{Alexey Barsukov\thanks{alexey.barsukov@matfyz.cuni.cz}}
\author[2]{Santiago Guzm\'an-Pro\thanks{santiago.guzman\_pro@tu-dresden.de\\
This project has been funded by the European Research Council (Project POCOCOP, ERC Synergy Grant 101071674). Views and opinions expressed are however those of the authors only and do not necessarily reflect those of the European Union or the European Research Council Executive Agency. Neither the European Union nor the granting authority can be held responsible for them.}}

\affil[1]{Faculty of Mathematics and Physics, Charles University, Prague}
\affil[2]{Institut f\"ur Algebra, TU Dresden}
\date{\today}

\begin{document}

\maketitle
\begin{abstract}
We present a structural classification of \emph{constraint
   satisfaction problems} (CSP) described by reflexive complete $2$-edge-coloured  graphs.
   In particular, this classification extends the structural dichotomy for graph homomorphism
   problems known as the  Hell--Ne\v{s}et\v{r}il theorem (1990). Our classification is also efficient:
   we can check in polynomial time whether the CSP of a reflexive complete $2$-edge-coloured
   graph is in P or NP-complete, whereas for arbitrary $2$-edge-coloured graphs, this task is NP-complete.
    We then apply our main result in the context of \emph{matrix partition problems} and \emph{sandwich problems}.
    Firstly, we obtain one of the few algorithmic solutions to general classes of matrix partition
    problems.  And secondly, we present a P vs.\ NP-complete classification of sandwich problems for matrix partitions.
\end{abstract}

\section{Introduction}

A \emph{split} graph is a graph $\bG$ whose vertex set can be partitioned into an
independent set and a clique, and more generally,  $\bG$ admits a $(k,\ell)$-partition
if it can be partitioned into at most $k$ independent sets and  at most $\ell$
cliques. \emph{Matrix partitions} are a broad generalisation of $(k,\ell)$-partitions
--- unless stated otherwise, all matrices in this paper are $(n\times n)$ symmetric
matrices $M$ with entries $m_{ij}$ in $\{0,1,\ast\}$.
An \emph{$M$-partition} of a graph $\bG = (V,E)$ is a function
$f\colon V\to [n]$  such that for  every pair of different vertices $u,v\in V$, if $uv\in E$,
then $m_{f(u)f(v)} \in\{1,\ast\}$,  and if $uv\not\in E$, then $m_{f(u)f(v)}\in \{0,\ast\}$.
So graphs that admit a $(2,0)$-partition correspond to bipartite graphs, and graphs that admit a
$(1,1)$-partition to split graphs; we depict the corresponding $\{0,1,\ast\}$-matrices below.
\[ M_B:=\begin{pmatrix}
                     0 & \ast   \\
                    \ast & 0   \\
    \end{pmatrix}~~~~
     M_S :=\begin{pmatrix}
                     0 & \ast   \\
                    \ast & 1   \\
                \end{pmatrix} 
    \]

The \emph{$M$-partition problem} takes as  input a graph $\bG$, and the task is to determine 
whether $\bG$ admits an $M$-partition. Similarly, the \emph{list} $M$-partition problem takes
as input a graph $\bG$ together  with a list of indices $L(v)\subseteq[n]$ for each vertex $v\in V$, and
the task is to decide if there is an $M$-partition $f:V\to [n]$ of $\bG$ respecting the lists,
i.e., $f(v)\in L(v)$ for each $v\in V$. Clearly, the list version is at least as
hard as the version without lists.

Matrix partition problems are a natural generalisation of graph homomorphism problems, and the
latter enjoy a nice P vs.\ NP-complete structural classification:
if $\bH$ is a bipartite graph or it contains a loop, then the $\bH$-homomorphism problem is
polynomial-time
solvable; otherwise, the $\bH$-homomorphism problem is NP-complete~\cite{HellNesetril}.
(The precise definition of homomorphism problems is given below.) Matrix  partition problems have
succeeded on attracting vast attention (see, e.g.,~\cite{cameronSODA2004,cookDM310,cyganSODA2011,%
federDAM159,federSIDMA16,federDM306,federTCS349,figueiredoJA37,HNBook}, and for a nice survey
see~\cite{hellEJC35}), however a complexity classification remains elusive.
In fact, most of the algorithmic solutions to matrix partitions concert specific small
matrices~\cite{cameronSODA2004,cookDM310,cameronSIDMA21,cyganSODA2011,federDAM154,figueiredoJA37},
while efficient algorithms for general families of matrices are scarce~\cite{federSODA2005,federSIDMA16}

\paragraph{The natural reduction to CSPs.}
Every matrix $M$ as above can be represented as a $2$-edge-coloured graph $\mathbb M$
by coding an entry $m_{ij} = 0$ or $m_{ij} = 1$ as a red or a blue edge  connecting
$i$ and $j$, respectively, and an entry $m_{ij} = \ast$ as a blue and a red edge
between $i$ and $j$. The following illustrate the representations of $M_B$ 
and $M_S$ (introduced above) as $2$-edge-coloured graphs.
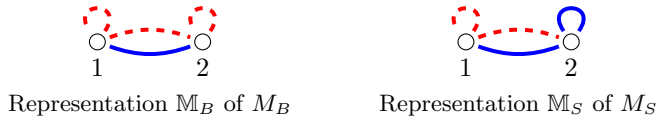
\begin{center}
    \begin{tikzpicture}
        \begin{scope}[scale=0.7]
           \node (L1) at (0,-1.2) {\small Representation $\mathbb M_B$ of $M_B$};
            \node (0) [vertex, label = below:{$1$}] at (-1,0){};
            \node (1) [vertex, label = below:{$2$}] at (1,0){};
            
            \draw [edge, red, dashed] (0) to [out=55,in=125, looseness=12] (0);
             \draw [edge, red, dashed] (1) to [out=55,in=125, looseness=12] (1);
           
            \draw [edge, blue] (0) to [bend right = 20] (1);
            \draw [edge, red, dashed] (0) to [bend left = 20] (1);

        \end{scope}
    
        \begin{scope}[scale=0.7, xshift=7cm]
           \node (L1) at (0,-1.2) {\small Representation $\mathbb M_S$ of $M_S$};
            \node (0) [vertex, label = below:{$1$}] at (-1,0){};
            \node (1) [vertex, label = below:{$2$}] at (1,0){};
            
            \draw [edge, red, dashed] (0) to [out=55,in=125, looseness=12] (0);
             \draw [edge, blue] (1) to [out=55,in=125, looseness=12] (1);
           
            \draw [edge, blue] (0) to [bend right = 20] (1);
            \draw [edge, red, dashed] (0) to [bend left = 20] (1);
        \end{scope}
\end{tikzpicture}
\end{center}
Similarly, we can represent a graph $\bG$ as a $2$-edge-coloured complete
graph $\nu(\bG)$ by coding edges in $\bG$ as blue edges, and pairs of non-edges $uv$ 
where $u\neq v$ as red edges. These representations yield a natural  and
well-known (see, e.g.,~\cite{federDAM154}) reduction 
from matrix partition problems to homomorphism problems of $2$-edge-coloured graphs:
a graph $\bG$ admits an $M$-partition if and only if $\nu(\bG)$  maps
homomorphically to $\mathbb M$ (see also Observation~\ref{obs:natural-reduction}).
We follow notation
from constraint satisfaction theory, and denote by $\CSP(\mathbb M)$
the homomorphism problem described by $\mathbb M$.

This  reduction is thus a natural tool to obtain efficient algorithmic
solutions to general families of matrix partition problems, whenever the corresponding
CSPs are tractable.
Unfortunately, in several cases there is a large complexity gap  between the
$M$-partition problem and $\CSP(\mathbb M)$.
For instance, if $M$ is a $\{0,1\}$-matrix, then the $M$-partition
problem is polynomial-time solvable (even in first-order logic~\cite{feder_hell_full_hom}),
but the homomorphism problem described by $\mathbb M$  might be $\NP$-complete, e.g., 
when $M$ is the adjacency matrix of $\mathbb K_3$ (the complete graph on $3$ vertices).
We believe that an instinctive step  to obtain general efficient solutions
to matrix partition problems, and thus, towards understanding their computational complexity, is
to first  understand when the natural reduction is algorithmically efficient.

\begin{problem}\label{prob:natural-reduction}
    Structurally classify the 
    matrices $M$ such that the
    homomorphism problem described by $\mathbb M$ is polynomial-time solvable
    (assuming $\PO \neq \NP$).
\end{problem}

\paragraph{Homomorphisms of edge-coloured graphs.}

Homomorphisms of edge-coloured graphs trace back at least to~\cite{brewster1993vertex}.
Since then, these have been studied from a combinatorial point of
view~\cite{alonJAC8,hellDM234,montejanoDAM158,ochemJGT85}, as well as from a computational
perspective~\cite{brewsterDAM49,brewsterENDM5,brewsterDM340,brewsterJGT110}.
Of course, the finite-domain CSP dichotomy~\cite{BulatovFVConjecture,ZhukFVConjecture}
presents a P vs.\ NP-complete classification of homomorphism problems of edge-coloured graphs.
Contrary to the Hell--Ne\v{s}et\v{r}il theorem, this characterisation
is of algebraic nature, and NP-complete to verify~\cite{MetaChenLarose}. In fact,
it is already NP-complete to verify if the homomorphism problem of a $2$-edge-coloured 
graph is tractable~\cite{brewsterDM340}. Hence, any efficient complexity classification
for these homomorphism problems must be restricted to certain subclasses of $2$-edge-coloured graphs.
For instance, in~\cite{brewsterENDM5} the authors pursue this task for edge-coloured cycles,
and in~\cite{bokMFCS2020} they consider caterpillar-like structures.
Towards settling Problem~\ref{prob:natural-reduction}, we consider the class of $2$-edge-coloured
graphs $\bM$ where for every pair of vertices there is at least one edge between them;
we call them \emph{reflexive complete} $2$-edge-coloured graphs. So we ask, 
\emph{is there an efficient classification of the complexity of homomorphism problems
    of reflexive complete $2$-edge-coloured graphs?}

Homomorphism problems of $2$-edge-coloured graphs have also been applied to 
model further problems in graph theory. The most studied direction concerns
homomorphisms of signed graphs~\cite{bokTCS1001,bokDM346,brewsterDAM49,brewsterDM340,kimKMJ63,naserasrJGT79}, and most
recently, homomorphism problems of $2$-edge-coloured graphs have resurfaced
in the context of Graph Sandwich Problems~\cite{bodirskySODA2026}.

\paragraph{Graph Sandwich Problems.}

Graph Sandwich Problems were introduced by Golumbic, Kaplan, and Shamir in~\cite{golumbicJA19}.
The \emph{Sandwich Problem} (SP) for a graph property $\Pi$ is the following computational problem. 
The input is a pair of graphs $(V,E_1)$ and $(V,E_2)$ where $E_1\subseteq E_2$, and the task
is to decide if there is an edge set $E$ where $E_1\subseteq E\subseteq E_2$, and the graph $(V,E)$
satisfies $\Pi$. In particular, the SP for $\Pi$ is at least as hard as the recognition problem
for $\Pi$: $(V,E)$ satisfies $\Pi$ if and only if the input $\bigl((V,E),(V,E)\bigr)$ is a yes-instance
of the Sandwich Problem for $\Pi$. 
Classifying the complexity of the Sandwich Problem for specific graph classes has been the
subject of a fair amount of literature~\cite{dantasAOR188,dantasDAM182,dantasDAM143,
dantasENTCS346,dantasICCGI07,figueiredoDAM121,golumbicJA19,alvaradoAOR280}.
In particular,  in~\cite{dantasDAM143,dantasDAM182,dantasICCGI07,figueiredoDAM121,golumbicJA19}
the authors studied the complexity of the SP for graph classes defined by certain
admissible partitions.
Here, we consider the \emph{Sandwich Problems for matrix partitions} together with its list variant. 
Even though these might be clear for context, we explicitly define them now.

\vspace{0.4cm}
\noindent\textbf{Sandwich Problem for $M$-partition}

 \textsc{Input:} A pair of graphs $(V,E_1)$ and $(V,E_2)$ where $E_1\subseteq E_2$.
 
 \textsc{Question:} Is there an edge set $E$ such that $E_1\subseteq E\subseteq E_2$, and $(V,E)$ admits an $M$-partition?

\vspace{0.4cm}
\noindent\textbf{List Sandwich Problem for $M$-partition}

\textsc{Input:} A pair of graphs $(V,E_1)$ and $(V,E_2)$ where $E_1\subseteq E_2$, and a list
$L(v)\subseteq [n]$ for each  $v\in V$.

\textsc{Question:} Is there an edge set $E$ such that $E_1\subseteq E\subseteq E_2$, and $(V,E)$ admits an $M$-partition
$f\colon V\to [n]$ \newline \indent such that $f(v)\in L(v)$ for each vertex $v\in V$?

\vspace{0.4cm}

In their seminal paper~\cite{golumbicJA19}, Golumbic, Kaplan, and Shamir showed that the SP
for split graphs is polynomial-time solvable, and in~\cite{dantasDAM143}, the authors
classify the complexity of the SP for $(k,\ell)$-partitions.
As mentioned before,  SP for split graphs, and more generally for $(k,\ell)$-partitions, correspond
to SP for certain $M$-partitions, and classifying the complexity of the latter was the subject
of~\cite{dantasDAM143}. Another famous matrix partition problem is the
\emph{stubborn partition}~\cite{cameronSODA2004,cyganSODA2011,dantasCATS2010},
i.e.,  partitioning the vertex set into at most four parts $A$, $B$, $C$, and $D$, such that
$A$ and $B$ are independent sets, $D$ is a clique, and no vertices from $A$ and $C$ are
adjacent. Equivalently,  a stubborn partition is an $M$-partition where $M$  is the matrix
below. 
The list stubborn sandwich problem is $\NP$-complete~\cite{dantasICCGI07}, whereas the original list version is in P~\cite{cyganSODA2011}.
\[  \begin{pmatrix}
                    0 & \ast & 0 & \ast  \\
                    \ast & 0 & \ast & \ast  \\
                    0 & \ast & \ast & \ast  \\
                    \ast & \ast & \ast & 1  \\
    \end{pmatrix} 
\]
Bodirsky and Guzm\'an-Pro~\cite{bodirskySODA2026} recently showed that for many natural SPs,
there is a polynomial-time equivalent homomorphism problem of $2$-edge-coloured graphs. In all their
cases, the target $2$-edge-coloured graph is necessarily infinite. We will see that in the case
of SP for matrix partitions, there is a polynomial-time equivalent  homomorphism
problem of a finite $2$-edge-coloured graph. Using this connection, we
will solve the following problem; a broad, but natural generalisation of the
classification of SPs for $(k,\ell)$-partitions~\cite{dantasDAM143}.

\begin{problem}\label{prob:SP-for-matrices}
    Classify the complexity of SPs for matrix partitions.
\end{problem}

\paragraph{Contributions.} In this paper, we present a \emph{structural}
classification of the complexity of homomorphism problems described by reflexive
complete $2$-edge-coloured  graphs (Theorem~\ref{thm:main})
that extends the classification of the complexity of graph homomorphism problems given
in~\cite{HellNesetril}. Our characterisation is also \emph{efficient} (Corollary~\ref{cor:dichotomy-P}): 
we can check in polynomial time whether the CSP of a reflexive complete $2$-edge-coloured
graph is in P or NP-complete, whereas for arbitrary $2$-edge-coloured graphs, this
task is NP-complete~\cite{brewsterDM340}. As an application of our main result, 
we solve Problems~\ref{prob:natural-reduction} and~\ref{prob:SP-for-matrices}, and
we present one of the few general classes of matrix partition problems
solved with a universal algorithm (Theorem~\ref{thm:bounded_width})--- ``most of the existing algorithms apply to concrete
small matrices''~\cite{federSODA2005}.

We now present a brief overview of the main ideas and components in the proof of our main result.
\begin{itemize}
    \item \emph{Homogeneous concatenations (Section~\ref{sec:homogeneous}).}
    These are specific kind of decompositions of a $2$-edge-coloured graph
    $\bH$, whose building blocks are \emph{homogeneous sets} in $2$-edge-coloured
    graphs (the name is inspired by the notion of homogeneous sets relevant
    in modular decompositions of uncoloured graphs~\cite{habibCSR4}). We see that
    if a $2$-edge-coloured reflexive graph $\bH$ can
    be constructed from a $2$-element structure by iteratively adding certain kind
    of homogeneous sets of size two (see Theorem~\ref{thm:bounded_width}),
    then the  $\CSP(\bH)$, and its list variant are solvable in polynomial time;
    and actually, in Datalog.

    \item \emph{Hereditarily pp-constructing $\bK_3$ (Section~\ref{sec:HH}).}
    Primitive positive constructions are certain specific ways of encoding one
    CSP into another one. In particular, if $\bH$ pp-constructs $\bK_3$, 
    then $\CSP(\bH)$ is $\NP$-complete. Motivated by the structural approach
    to graph homomorphism problems~\cite{HellNesetril}, and to CSPs of smooth
    digraphs~\cite{bangjensenDM138,BartoKozikNiven},
    we present a set  $\mathcal F$ of ``small'' \emph{hereditarily hard} $2$-edge-coloured
    graphs $\mathbb F$, that is, if $\bH$ contains $\mathbb F$ as a (not necessarily induced)
    substructure, then $\bH$  pp-constructs $\bK_3$. In turn, we achieve this by constructing the
    \emph{Siggers power} $\Sig(\bH)$,  and  establishing structural conditions
    which, if satisfied by $\Sig(\bH)$,
    imply that $\bH$ hereditarily pp-constructs $\bK_3$.
    \item \emph{Alternating components (Section~\ref{sec:alt-components}).}   
    We introduce a connectivity notion for $2$-edge-coloured graphs, where
    \emph{alternating components}
    correspond to connected components. We observe
    that alternating components yield canonical decompositions of
    $2$-edge-coloured graphs matching homogeneous concatenations
    (Observation~\ref{obs:topological-ordering}). Using the set
    $\mathcal F$ of ``small'' graphs that hereditarily pp-construct
    $\bK_3$, we see that if $\bH$ does not pp-construct $\bK_3$,
    then each alternating component has size at most four
    (Proposition~\ref{prop:intermediate-class}) ---
    almost matching the sufficient condition for tractability from
    Theorem~\ref{thm:bounded_width}.
    \item \emph{Algebra again (Section~\ref{sec:structura-clas}).} To close the
    gap between Proposition~\ref{prop:intermediate-class} and
    Theorem~\ref{thm:bounded_width}, we construct the \emph{$p$-cyclic power}
    $\Cyc_p(\bH)$ of a $2$-edge-coloured graph $\bH$. It follows from constraint
    satisfaction theory, that if $\bH$ does not pp-construct $\bK_3$, then
    there is a homomorphism $\Cyc_p(\bH)\to \bH$ for every prime $p>|H|$
    (Theorem~\ref{thm:CSP-dichotomy} and Lemma~\ref{lem:cyclic-power}). 
    Finally, we show that if a $2$-edge-coloured graph $\bH$
    admits a decomposition as in Proposition~\ref{prop:intermediate-class}
    but not as in Theorem~\ref{thm:bounded_width}, then there is
    no homomorphism $\Cyc_p(\bH)\to \bH$, and thus $\CSP(\bH)$
    is $\NP$-complete --- closing the gap, and settling our main
    result (Theorem~\ref{thm:main}).

\end{itemize}

It is worth pointing out that structural classifications of the complexity
of finite-domain CSPs are scarce.
Besides the Hell--Ne\v set\v ril theorem, we are only aware of structural
classifications for smooth digraphs~\cite{BartoKozikNiven}, semicomplete
digraphs~\cite{bangjensenSIDMA1}, signed graphs~\cite{BrewsterSiggersSigned},
and the list-homomorphism problems for graphs~\cite{ListStructuralAll}.
As mentioned above, one might not
expect such a beautiful and simple structural classification for finite-domain
CSPs because it is NP-complete to decide if the CSP of a finite structure $\bA$ is tractable~\cite{MetaChenLarose}
(assuming $\PO\neq \NP$). In fact, we do not expect such a classification 
even for digraph CSPs,  because this family already encodes all finite-domain
CSPs~\cite[Theorem 10]{FederVardi}. In turn,
CSPs of $2$-edge-coloured (undirected) graphs capture all digraph
CSPs~\cite{brewsterDM340}, leaving little hope for a
nice structural classification of this family of CSPs.

\paragraph{Outline of the paper.}
In Section~\ref{sec:prelim}, we provide the essential background on CSPs and revisit the connection between matrix partition problems 
with CSPs of reflexive complete $2$-edge-coloured  graphs.
The contents of Sections~\ref{sec:homogeneous}--\ref{sec:structura-clas} were previously introduced.
Finally, in Section~\ref{sec:CSP-SP} we apply our main result to 
SP of matrix problems, solving Problem~\ref{prob:SP-for-matrices}. 
We present conclusions and open problems in Section~\ref{sec:conclusion}.

\section{Preliminaries}
\label{sec:prelim}

We assume familiarity with elementary first-order logic
and first-order structures (for a reference see~\cite{Hodges}). In this
paper, we only consider finite relational signatures, and all structures
are finite (unless stated otherwise), and use the following conventions.
A \emph{relational signature} $\tau$ is a set of relation symbols,
and to each relation symbol $S\in \tau$ we associate a positive integer
$s$ called its \emph{arity}. We denote structures with symbols
$\bG, \bH, \dots$ and their vertex sets (domains) by $G,H,\dots$. 
For each relational symbol $S\in \tau$ of arity $s$, the
\emph{interpretation} of $S$ in a $\tau$-structure $\bH$ is a subset
of $H^s$; we denote this set by $S(\bH)$.

\subsection{Graphs and digraphs}

We consider digraphs as $\{E\}$-structures where $E$ is a binary
relation symbol, 
and (undirected) graphs as digraphs
$\bG$ where the interpretation $E(\bG)$ is symmetric.
Besides this convention, we follow standard notions from 
graph theory (see, e.g.,{\cite{bondy2008}}). In particular, 
we denote edges of an (undirected) graph by $uv$ instead of $(u,v)$. 
We write $\bK_n$ to denote the complete graph on $n$ vertices, 
$\bC_n$ denotes the (undirected) cycle on $n$ vertices, and
$\bP_n$ the (undirected) path on $n$ vertices. Similarly, 
$\vec{\bC}_n$ and $\vec{\bP}_n$, denote the directed
cycle and path on $n$ vertices. When depicting a graphs and digraphs,
we draw solid black line edge to represent undirected edges $uv$, and
an solid black arrow to depict directed edges $(u,v)$.

A graph $\bG$ is an \emph{induced subgraph} of $\bH$
if $G\subseteq H$, and $E(\bG) = E(\bH)\cap G^2$. Similarly, 
if the subgraph of $\bH$ \emph{induced} by a set of vertices
$U\subseteq H$ is the graph $(U, E(\bH)\cap U^2)$; we will sometimes
abuse notation and simplify notation by writing $(U, E(\bH))$.

The \emph{girth} of a graph $\bG$ is the length of the shortest cycle in $\bG$.
The girth of a digraph $\bD:=(D, E(\bD))$ is the girth of the graph whose
edge set is the symmetric closure of $E(\bD)$.

\subsection{Edge-coloured graphs}

In this paper, a \emph{$2$-edge-coloured graph} $\bH$ is a
$\{R,B\}$-structure where $R$ and $B$ are binary symbols, 
and the interpretations $R(\bH)$ and $B(\bH)$ are symmetric
binary relations. We follow standard notation in graph theory,
and  write $uv\in R(\bH)$ and $uv\in B(\bH)$, instead of
$(u,v)\in R(\bH)$ and $(u,v)\in B(\bH)$, respectively. 
We refer to elements $uv$ in $R(\bH)$ as
\emph{red edges}, to elements in $B(\bH)$ as \emph{blue edges}, 
and if $uv\in B(\bH)\cap R(\bH)$ we call it an \emph{$\ast$-edge}
(the name arises from the encoding of matrices $M$ as $2$-edge-coloured graphs).
When drawing $2$-edge-coloured graphs, we depict a red edge by a red
dashed edge, and a blue edge by a blue solid edge.

We say that a $2$-edge-coloured graph $\bH$ is \emph{complete}
if for every pair of non-equal vertices 
$u,v\in H$ are connected by some edge, in symbols, $uv\in R(\bH)\cup B(\bH)$.
Recall that in the introduction we defined the complete $2$-edge-coloured graph 
$\nu(\bG)$ constructed from a graph $\bG$ by colouring all edged with blue, 
and all non-edge $uv$ with $u\neq v$ with red. 

To every (uncoloured) graph $\bG$, we associate the reflexive complete
$2$-edge-coloured  graph $\bG^\ast$ where $B(\bG^\ast) = E(\bG)$, and
$R(\bG^\ast) = \{uv\colon u,v\in G, uv\not\in E(\bG)\}$.
Notice $\nu(\bG)$ corresponds to the structure obtained
from $\bG$ by removing all loops.

Given a reflexive complete $2$-edge-coloured  graph $\bH$, we denote by $\overline\bH$ the 
\emph{dual} of $\bH$, i.e., the reflexive complete $2$-edge-coloured  graph by 
changing blue edges for red edges and vice versa. See Figure~\ref{fig:K2}
for an illustration of these operations.

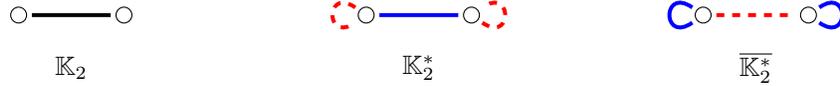
\begin{figure}[ht!]
\centering
    \begin{tikzpicture}
    
        \begin{scope}[scale=0.7]
            \node (L1) at (0,-1) {$\bK_2$};
            \node (1) [vertex] at (-1,0){};
            \node (2) [vertex] at (1,0){};
            
            \draw [edge] (1) to  (2);
            
        \end{scope}

        \begin{scope}[scale=0.7, xshift = 6.6cm]
           \node (L1) at (0,-1) {$\bK_2^\ast$};
             \node (1) [vertex] at (-1,0){};
            \node (2) [vertex] at (1,0){};
            
            \draw [edge, red, dashed] (1) to [out=55 + 90,in=125 + 90, looseness=12] (1);
            \draw [edge, red, dashed] (2) to [out=55 -90,in=125 + -90, looseness=12] (2);               \draw [edge, blue] (2) to (1);
        \end{scope}
        
        \begin{scope}[scale=0.7, xshift=13cm]
            \node (L1) at (0,-1) {$\overline{\bK_2^\ast}$};
             \node (1) [vertex] at (-1,0){};
            \node (2) [vertex] at (1,0){};
        
            \draw [edge, blue] (1) to [out=55 + 90,in=125 + 90, looseness=12] (1);
            \draw [edge, blue] (2) to [out=55 -90,in=125 + -90, looseness=12] (2);
           \draw [edge, red, dashed] (1) to (2);
                
        \end{scope}

\end{tikzpicture}
\caption{An illustration of the operations $(\cdot)^\ast$, and $\overline{(\cdot)}$
with the structures $\bK_2$, $\bK_2^\ast$, and $\overline{\bK_2^\ast}$.}
\label{fig:K2}
\end{figure}

We say that a vertex $v$ of a $2$-edge-coloured graph $\bH$ is a
\emph{blue vertex} if $vv\in B(\bH)$, and $v$ is a \emph{red vertex} if $vv\in R(\bH)$.
We denote by $H_B$ and by $H_R$ the subsets of blue and of red vertices of $\bH$, 
respectively. \emph{Red loops}, \emph{blue loops}, and \emph{$\ast$-loops}
are defined analogously to red edges, blue edges and $\ast$-edges.

\emph{Induced $2$-edge-coloured subgraphs} are defined analogously to 
induces subgraphs. As explained above, we will sometimes write $(U,R(\bH), B(\bH))$
to denote the $2$-edge-coloured subgraph of $\bH$ induces by $U\subseteq H$.

A \emph{red component} of a reflexive complete $2$-edge-coloured graph $\bH$ is a
connected component of $(H_B,R(\bH_B))$, i.e., a connected component of the graph with
red edges induced by the blue vertices.
A \emph{blue component} is defined similarly.

\subsection{Homomorphisms and CSPs}

Consider a pair of $\tau$-structures $\bG$ and $\bH$. 
A \emph{homomorphism} is a function $f\colon G\to H$ such that
for each relation symbol $S\in \tau$ of arity $s$, and each
$(u_1,\dots, u_s)\in S(\bG)$ the tuple $(f(u_1),\dots, f(u_s))$ belongs
to $S(\bH)$. In particular, a homomorphism between (edge-coloured) graphs
is a vertex mapping $f\colon G\to H$ that preserved adjacencies (and edge colours). 
When such a homomorphism exists we write $\bG\to \bH$, and otherwise,
we write $\bG\not\to \bH$. A pair of structures $\bG$ and $\bH$ are
\emph{homomorphically equivalent} if $\bG\to \bH$ and $\bH\to \bG$.

The \emph{constraint satisfaction problem} $\CSP(\bH)$ with template $\bH$, takes as
an input a structure $\bG$ (with the same signature as $\bH$), and the task is to decide
if $\bG\to \bH$. For instance, $\CSP(\bK_3)$ is the well-known 3-Colourability problem.
A well-known result from Erd\H{o}s~\cite{Erd:Gtp} about $k$-Colourability states that for
every pair of positive integers $\ell,k$ there is a graph $\bG$ with girth strictly greater
than $\ell$, and such that $\bG$ does not admit a proper $k$-colouring. This result generalises
to arbitrary relational structures, and it is known as the Sparse Incomparability
Lemma~\cite[Theorem 1.1]{Kun}. Below we state a well-known consequence of this result, and to 
stay in the context of this paper, we consider a version for $2$-edge-coloured graphs.
The girth of a $2$-edge-coloured graph $\bH$ is one if $\bH$
contains a loop, is $2$ if it contains no loop but it contains an $\ast$-edge, and
otherwise, the girth of $\bH$ is the girth of the (uncoloured) graph with edge set
$R(\bH)\cup B(\bH)$.

\begin{theorem}\label{thm:large-girth}
    For every finite $2$-edge-coloured graph $\bH$ and every positive integer $\ell$, 
    $\CSP(\bH)$ is polynomial-time equivalent to $\CSP(\bH)$ restricted to $2$-edge-coloured
    graphs with girth strictly larger than $\ell$.
\end{theorem} 

\subsection{Primitive positive constructions}

A \emph{primitive positive} formula (pp-formula) is an existential
formula whose quantifier-free part is a conjunction of positive atoms.
For instance, $\delta(x,y):=\exists z.\; E(x,z)\land E(y,z)$ is a pp-formula
stating that $x$ and $y$ are connected by a walk of length two.
An $r$-ary relation $S\subseteq H^s$ is \emph{primitive positively
definable} (pp-definable) if there is a pp-formula $\phi(x_1,\dots, x_s)$
such that $(u_1,\dots, u_s)\in S$ if and only if $\bH\models \phi(u_1,\dots, u_s)$.
In particular, we say that a set of vertices $S\subseteq H$ is pp-definable
if $S$ is pp-definable as a unary relation.

A binary \emph{pp-power}\footnote{This definition naturally generalises to
non-binary power (see, e.g.,~\cite{wonderland}) but they are not needed for this work.} 
of a $2$-edge-coloured graph (structure) $\bH$, is a
structure $\mathbb P$ with binary signature $S_1,\dots, S_k$ such that
\begin{itemize}
    \item the domain $P$ of $\mathbb P$ is $H^d$ for some positive integer $d$, and
    \item for each $i\in[k]$ there is a pp-formula $\delta_i(x_1,\dots, x_d,y_1,\dots, y_d)$
    such that $((u_1,\dots, u_d),(v_1,\dots, v_d))\in S_i(\bP)$ if and only if
    $\bH\models \delta_i(u_1,\dots, u_d,v_1,\dots, v_d)$.
\end{itemize}
We say that a structure $\bH$ \emph{primitively positively constructs} 
(\emph{pp-constructs}) a structure $\bB$
if $\bB$ is homomorphically equivalent to a pp-power of $\bH$.
Primitive positive constructions encode certain specific kind of gadget
reduction, which yield log-space reductions --- we state this in the
following theorem, and subsequently present a simple example.

\begin{theorem}[\hspace{-4pt}{\cite[Corollary~3.5]{wonderland}}]\label{thm:pp-constructions}
    If a structure $\bH$ pp-constructs a structure $\bH'$, then 
    there is a log-space reduction from $\CSP(\bH')$ to $\CSP(\bH)$.
\end{theorem}

\begin{example}\label{ex:pp-def}
    Consider the $2$-edge-coloured  graph $\bA$ depicted below, and the
    pp-formula $\delta(x,y):=\exists z.\; B(x,z)\land R(y,z)$.
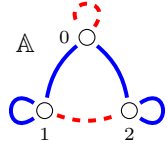
    \begin{center}
    \begin{tikzpicture}
    \begin{scope}[scale=0.5]
        \node at (145:2){$\bA$};
            \node (0) [vertex, label = left:{\scriptsize $0$}] at (90:1.3){};
            \node (1) [vertex, label = below:{\scriptsize $1$}] at (210:1.3){};
            \node (2) [vertex, label = below:{\scriptsize $2$}] at (330:1.3){};
            
            \draw [edge, red, dashed] (0) to [out=55,in=125, looseness=12] (0);
             \draw [edge, blue] (1) to [out=55 + 90,in=125 + 90, looseness=12] (1);
            \draw [edge,  blue] (2) to [out=55 -90,in=125 + -90, looseness=12] (2);
               
            \foreach \from/\to in {0/1, 2/0} 
                \draw [edge, blue] (\from) to [bend right = 20] (\to);
                \draw [edge, red, dashed] (1) to [bend right = 20] (2);
        \end{scope}
\end{tikzpicture}
\end{center}
    Notice that the graph with vertex set $A$, and edges
    $(u,v)$ where $\delta(u,v)$ is true in $\bA$ is the complete graph 
    on three vertices. Hence, $\bA$ pp-constructs $\bK_3$ (in this simple
    case, the pp-power has vertex set $A^1$, and $\bK_3$ is isomorphic to a pp-power
    of $\bH$). Therefore, $\CSP(\bA)$ is $\NP$-complete by Theorem~\ref{thm:pp-constructions}
    (and because 3-Colourability is $\NP$-complete).

\end{example}

\subsection{Cores}

An \emph{endomorphism} of a structure $\bH$ is a homomorphism
$f\colon \bH\to \bH$. An \emph{automorphism} is a bijective endomorphism.
A finite 2-edge-coloured graph $\bH$ is a \emph{core} if every endomorphism of
$\bH$ is an automorphism. 
It is known that every 2-edge-coloured graph $\bH$ contains an induced subgraph
$\bG$ which is a core and is homomorphically equivalent to $\bH$.
Such $\bG$ is unique up to isomorphism and is called \emph{the core} of $\bH$.

Cores are useful in the study of the complexity of CSPs. Firstly, 
if $\bH$ is the core of $\bH'$, then $\CSP(\bH)$ and $\CSP(\bH')$ are the same problem.
Moreover, ``adding constants''  does not change the complexity of
the CSP of a core. That is, if $H = \{h_1,\dots, h_n\}$, then $\CSP(\bH)$ and
$\CSP(\bH,\{h_1\},\dots, \{h_n\})$ are log-space equivalent~\cite{JBK}. 
On a semantic level, $\CSP(\bH,\{h_1\},\dots, \{h_n\})$ encodes the pre-coloured
version of the homomorphism problem described by $\bH$ (i.e., some vertices $v$
of the input $\bG$ might be restricted to map to some fixed element $h\in H$.)
To simplify notation, we will write $(\bH,h_1,\dots, h_n)$ instead
of $(\bH,\{h_1\},\dots, \{h_n\})$, and when useful, we will assume that
$H = \{h_1,\dots, h_n\}$ without explicitly writing. We will often
refer to $(\bH,h_1,\dots, h_n)$ as $\bH$ \emph{with constants}. The following statement
presents an explanation of the log-space equivalence between $\CSP(\bH)$ and
$\CSP(\bH, h_1,\dots, h_n)$ in terms of pp-constructions.

\begin{theorem}[\hspace{-0.8pt}{\cite[Lemma~3.9]{wonderland}}]\label{thm:core_constants}
    If $\bH$ is a core, then $\bH$ pp-constructs $(\bH,h_1,\dots, h_n)$.
\end{theorem}

\subsection{Polymorphisms}

The product $\bG\times \bH$ of a pair of $2$-edge-coloured graphs
is the structure with vertex set $G\times H$ where $(g,h)(g',h')$ is a
red edge (resp.\ blue edge) if and only if $gg'\in R(\bG)$ and $hh'\in R(\bH)$
(resp.\ $gg'\in B(\bG)$ and $hh'\in B(\bH)$). For a positive integer $m$, the
\emph{$m$-power} $\bH^m$ of a $2$-edge-coloured graph $\bH$ is recursively 
defined as $\bH^m:=\bH^{m-1}\times \bH$ where $\bH^1 = \bH$.

For a positive integer $m\ge 2$, an \emph{$m$-ary polymorphism}
of a 2-edge-coloured graph $\bH$ is a homomorphism $f\colon \bH^m\to \bH$.
Equivalently, it is a function $f\colon H^m\to H$ such that, for every choice
of blue edges $a_1b_1,\ldots,a_mb_m\in B(\bH)$, and red edges
$c_1d_1,\ldots,c_md_m\in R$, we have that $f(a_1,\ldots,a_m)f(b_1,\ldots,b_m)\in B(\bH)$,
and $f(c_1,\ldots,c_m)f(d_1,\ldots,d_m)\in R(\bH)$. We denote
by $\Pol(\bH)$ the set of all polymorphisms of $\bH$.
An $m$-ary polymorphism $f\colon \bH^m\to \bH$ is called
\begin{itemize}
    \item \emph{cyclic} if it satisfies $f(x_1,\dots, x_m) = f(x_2,\dots, x_m,x_1)$ 
    for all $x_1,\dots, x_m\in H$;
    \item a \emph{weak near-unanimity} if it satisfies
    $f(x,\dots, x,y) = f(x,\dots, x,y,x) = \dots = f(y,x,\dots, x)$, for all $x,y\in H$;
    \item \emph{Siggers} if $m = 4$ and for all $a,r,e\in H$ the equality
    $f(a,r,e,a) = f(r,a,r,e)$ holds. 
\end{itemize}
The following theorem is known as the finite-domain CSP
dichotomy~\cite{BulatovFVConjecture,ZhukFVConjecture} (for a journal
version see~\cite{Zhuk20}).

\begin{theorem}\label{thm:CSP-dichotomy}
    For every finite structure $\bH$ one of the following holds: either
    $\bH$ pp-constructs $\mathbb{K}_3$, and in this case $\CSP(\bH)$ is $\NP$-complete;
    or otherwise, the following equivalent statements hold and in 
    any such case $\CSP(\bH)$ is polynomial-time solvable:
    \begin{itemize}
        \item for every prime $p > |H|$ there is a $p$-ary cyclic polymorphism $f\colon \bH^p\to \bH$,
         \item there is a weak near-unanimity polymorphism $f\colon \bH^n\to \bH$ for some positive integer $n\ge 2$, and
        \item there is a Siggers polymorphism $f\colon  \bH^4\to \bH$.
    \end{itemize}
\end{theorem} 

For a set $A$ and $n\in\mathbb{N}$, a mapping $f\colon A^n\to A$ is called
\emph{idempotent} if, for all $a\in A$, we have $f(a,\ldots,a)=a$.
The following is an immediate consequence of Theorem~\ref{thm:core_constants}
and Theorem~\ref{thm:CSP-dichotomy}.

\begin{corollary}\label{cor:core_idempotent}
    If $\bH$ is a core, then either $\bH$ pp-constructs $\mathbb{K}_3$, and in these case $\CSP(\bH)$ is $\NP$-complete; or otherwise the following equivalent statements hold and in 
    any such case $\CSP(\bH)$ is polynomial-time solvable:
    \begin{itemize}
        \item for every prime $p > |H|$ there is a $p$-ary idempotent cyclic polymorphism $f\colon \bH^p\to \bH$, and
        \item there is an idempotent Siggers polymorphism $f\colon  \bH^4\to \bH$.
    \end{itemize}
\end{corollary}

\subsection{The natural reduction to CSPs}
\label{sec:natural-reduction}

As mentioned in the introduction, for every $(n\times n)$ 
symmetric matrix $M$ we construct a reflexive complete $2$-edge-coloured  graph
$\mathbb M$ such that the $M$-partition naturally reduces to $\CSP(\mathbb M)$. 
The vertex set of $\mathbb M$ is $[n]$, the set of blue edges is
$\{ij\colon m_{ij}\in\{1,\ast\}\}$, and the set of red edges is
$\{ij\colon m_{ij}\in \{0,\ast\}\}$.

\begin{observation}\label{obs:natural-reduction}
    For every matrix $M$, the \emph{natural reduction}, defined by transforming $\bG$ into
    the irreflexive complete $2$-edge-coloured $\nu(\bG)$
    is a log-space reduction from 
    the $M$-partition problem to $\CSP(\mathbb M)$. 
\end{observation}

A concrete example of the complexity gap between 
the $M$-partition problem and $\CSP(\mathbb M)$
is the matrix $M$ associated to the structure $\bA$ from
Example~\ref{ex:pp-def} (i.e., $m_{ij} = 0$ if $ij\in R(\bA)$
and $m_{ij} = 1$ if $ij\in B(\bA)$). Then, $\mathbb M = \bA$,
and so  $\CSP(\mathbb M)$ is NP-complete by
Example~\ref{ex:pp-def}.  However, for every $\{0,1\}$-matrix $M'$,
the $M'$-partition problem is polynomial-time solvable and in first-order
logic~\cite{feder_hell_full_hom}.

Clearly, the mapping $M\mapsto \mathbb M$ is a one-to-one
correspondence between symmetric matrices over $\{0,1,\ast\}$
and reflexive complete $2$-edge-coloured  graphs. Hence, 
by Observation~\ref{obs:natural-reduction},
and by the finite-domain dichotomy (Theorem~\ref{thm:CSP-dichotomy}),
Problem~\ref{prob:natural-reduction} is subsumed by the following one.

\begin{problem}\label{prob:2-ec-rc}
    Present a structural characterisation of reflexive complete $2$-edge-coloured
    graphs $\bH$ that do not pp-construct $\bK_3$, and thus
    $\CSP(\bH)$ is polynomial-time solvable.
\end{problem}

The known algorithms used to solve CSPs of structures that do not pp-construct 
$\bK_3$~\cite{BulatovFVConjecture,ZhukFVConjecture}, are heavily involved, and for this reason, we also
set the task of presenting simpler algorithms that solve the CSPs
of reflexive complete $2$-edge-coloured  graphs that do not pp-construct
$\bK_3$. We begin by presenting such algorithms in Section~\ref{sec:homogeneous}.

\section{Tractability and homogeneous sets}
\label{sec:homogeneous}

In this section we present a structural property $\Pi$,
such that if a reflexive complete $2$-edge-coloured graph $\mathbb H$ satisfies
this property $\Pi$, then $\CSP(\bH)$ can be solved in polynomial
time (Theorem~\ref{thm:homogeneous-Ptime}).

Given a reflexive $2$-edge-coloured graph $\bH$, we say that
a set $S\subseteq H$ is \emph{homogeneous} if every vertex
$h\in H\setminus S$ is connected to $s\in S$ by a blue
edge if and only if $ss\in B$, and by a red edge if and only if
$ss\in R$. We say that a vertex $u\in H$ is \emph{homogeneous
vertex} if $\{u\}$ is a homogeneous set, and an $\ast$-edge $uv\in B\cap R$
is \emph{homogeneous} $\ast$-edge if $\{u,v\}$ is a homogeneous set.

We will show that if $\bH$  contains a homogeneous
vertex or a homogeneous $\ast$-edge, then the list version of
$\CSP(\bH)$ reduces in polynomial-time to the list version of
$\CSP(\bH')$ where $\bH'$ is obtained from $\bH$ by removing
the homogeneous vertex or homogeneous $\ast$-edge, respectively.

\begin{observation}\label{obs:universal-vertex}
Let $\bH$ be a reflexive complete $2$-edge-coloured  graph.
If $\bH$ contains a homogeneous vertex $h$, then 
the list version of $\CSP(\bH)$ reduces in polynomial-time
to the list version of $\CSP(\bH\setminus \{h\})$.
\end{observation}
\begin{proof}
    The claim is straightforward, we present a reduction
    in Algorithm~\ref{alg:universal-vertex}.
\end{proof}

\RestyleAlgo{ruled}
\SetAlgoVlined{}
\begin{algorithm}\label{alg:universal-vertex}
\DontPrintSemicolon{}
\SetKwInOut{Input}{input}
\Input{a $2$-edge-coloured graph $\bG$ and a list $L(v)\subseteq H$ for each $v\in G$}
\caption{A reduction from $\CSP(\bH)$ to the list version of $\CSP(\bH\setminus \{h\})$
whenever 
$h$ is a homogeneous vertex.}
    Let $S =  G$\\
    If $hh\in B(\bH)$:
    \ForEach{$u \in G$}{
        remove $h$ from $L(u)$ \textbf{if} $u$ is incident to some red edge in $\bG$\\
        remove $u$ from $S$ \textbf{if} $h\in L(u)$ and $u$ is not incident to some red edge in $\bG$\\
        }
    If $hh\in R(\bH)$:
    \ForEach{$u \in G$}{
        remove $h$ from $L(u)$ \textbf{if} $u$ is incident to some blue edge in $\bG$\\
        remove $u$ from $S$ \textbf{if} $h\in L(u)$ and $u$ is not incident to some blue edge in $\bG$\\
        }
    \textbf{Return:} $\bG[S]$ with the updated lists $L(u)$. 
\end{algorithm}

Recall that a \emph{red component} of a $2$-edge-coloured graph $\bH$
is a connected component of $(H_B,R)$, and a \emph{blue component}
is a connected component of $(H_R,B)$.

\begin{lemma}\label{lem:hom-mono-edge}
Let $\bH$ be a reflexive complete $2$-edge-coloured  graph.
If $\bH$ contains a homogeneous $\ast$-edge $h_1h_2$, where
$h_1h_1,h_2h_2\in B(\bH)$ or $h_1h_1,h_2h_2\in R(\bH)$,
then the list version of $\CSP(\bH)$ reduces in polynomial-time
to the list version of $\CSP(\bH\setminus \{h_1,h_2\})$.
\end{lemma}
\begin{proof}
    We present a reduction in Algorithm~\ref{alg:universal-mono-edge}
    when $h_1h_1,h_2h_2\in B(\bH)$ --- the case $h_1h_1,h_2h_2\in B(\bH)$
    is symmetric.
    The consistency of the reduction follows by the
    following simple observation. For a $2$-edge-coloured graph $\bG$
    there is a homomorphism 
    $f\colon \bG\to \bH$ and $f(g) = h_i$ for some $g\in G$ and $i\in[2]$
    if and only if $f$ maps the red component of $g$ to $\{h_1,h_2\}$
    (because all vertices in $H\setminus\{h_1,h_2\}$ are connected
    to $h_1$ and $h_2$ by a blue edge). For the correctness
    suppose  that there is a homomorphism $f\colon \bG\to \bH$
    respecting the lists of each $v\in G$. Notice that if there is a
    partial homomorphism $f'$ from the red component $\bC$ of some $u\in G$
    to $\bH[h_1,h_2]$, then the mapping $f''\colon G\to H$
    defined by $f'$ for each $c\in C$, and by $f$ elsewhere
    is indeed a homomorphism $f''\colon \bG\to \bH$ (because
    each $c\in C$ is connected only by blue vertices to vertices 
    in $H\setminus C$, and $h_1$ and $h_2$ are connected by blue
    edges to all vertices in $\bH$).
\end{proof}

\RestyleAlgo{ruled}
\SetAlgoVlined{}
\begin{algorithm}\label{alg:universal-mono-edge}
\DontPrintSemicolon{}
\SetKwInOut{Input}{input}
\Input{a $2$-edge-coloured graph $\bG$ and a list $L(v)\subseteq H$ for each $v\in G$}
\caption{A reduction from the list version of $\CSP(\bH)$ to the list version of $\CSP(\bH\setminus \{h_1,h_2\})$ whenever 
$h_1h_2$ is a homogeneous $\ast$-edge and $h_1h_1,h_2h_2\in B(\bH)$.} 
    Let $S =  G$\\
    \ForEach{red component $\bC$ of $\bG$}{
        remove $C$ from $S$ \textbf{if} there is a homomorphism $\bC\to \bH[h_1,h_2]$
        respecting the lists $L(c)\cap \{h_1,h_2\}$ for each $c\in C$.\\
        \textbf{otherwise}, remove $h_1$ and $h_2$ from $L(c)$ for each $b\in C$.\\
    }
    \textbf{Return:} $\bG[S]$ with the updated lists $L(u)$. 
\end{algorithm}

Finally, we present a reduction when $\bH$ contains a homogeneous
$\ast$-edge with $h_1h_1\in B(\bH)$ and $h_2h_2\in R(\bH)$. 
An \emph{alternating path} in a $2$-edge-coloured graph $\bG$
is a path that has no monochromatic pair of consecutive edges. For
a vertex $g\in G$ we denote by $A_R(g)$ the set of vertices
$u$ for which there is an alternating $gu$-path starting
with a red edge. 
The \emph{alternating distance} from $g$ to $u$ in
$A_R(g)$ is the length of the shortest alternating
$gu$-path starting with a red edge. We define $A_B(g)$ 
and the alternating distance from $g$ to $u$ in $A_B(g)$ analogously.

\begin{lemma}\label{lem:hom-het-edge}
Let $\bH$ be a reflexive complete $2$-edge-coloured  graph.
If $\bH$ contains a homogeneous $\ast$-edge $h_1h_2$ where
$h_1h_1\in B(\bH)$ and $h_2h_2\in R(\bH)$,
then the list version of $\CSP(\bH)$ reduces in polynomial-time
to the list version of $\CSP(\bH\setminus \{h_1,h_2\})$.
\end{lemma}
\begin{proof}
    We present a reduction in Algorithm~\ref{alg:universal-split}.
    Similarly as before, the consistency of the reduction
    follows by a simple observation: for a $2$-edge-coloured graph $\bG$
    there is a homomorphism 
    $f\colon \bG\to \bH$ and $f(g) = h_1$ for some $g\in G$ 
    if and only if $f$ maps $A_R(g)$ to $\{h_1,h_2\}$. Since
    the unique red neighbour of $h_1$ in $\bH$ is $h_2$, and
    the unique blue neighbour of $h_2$ in $\bH$ is $h_1$, the previous
    observation for $u\in A_R(g)$ follows with  simple inductive argument
    over the length of the shortest alternating $gu$-path that starts with
    a red edge. The correctness follows again with similar arguments
    as in the proof of Lemma~\ref{lem:hom-mono-edge} using the
    definition of $A_R(g)$, and for every $h\in H$ there is blue
    edge $h_1h$ and a red edge $h_2h$ in $\bH$.
\end{proof}

\RestyleAlgo{ruled}
\SetAlgoVlined{}
\begin{algorithm}\label{alg:universal-split}
\DontPrintSemicolon{}
\SetKwInOut{Input}{input}
\Input{a $2$-edge-coloured graph $\bG$ and a list $L(v)\subseteq H$ for each $v\in G$}
\caption{A reduction from the list version of $\CSP(\bH)$ to the list version of $\CSP(\bH\setminus \{h_1,h_2\})$ whenever 
$h_1h_2$ is a homogeneous $\ast$-edge with $h_1h_1\in B(\bH)$ and $h_2h_2\in R(\bH)$.}
    Let $S =  G$\\
    \While{there is $g\in S$ with $h_1\in L(g)$}{
        remove $A_R(g)$ from $S$ \textbf{if} the mapping $u\mapsto h_1$ 
        when the alternating distance from $g$ to $u$
        is even, and $u\mapsto h_2$ otherwise, respects the lists $L(u)$
        for each $u\in A_R(g)$, and it defines a homomorphism
        $\bG[A_R(g)]\to H[h_1,h_2]$.\\
        \textbf{otherwise}, remove $h_1$ from $L(g)$.\\
    }
    \While{there is $g\in S$ with $h_2\in L(g)$}{
        remove $A_B(g)$ from $S$ \textbf{if} the mapping $u\mapsto h_2$ 
        when the alternating distance from $g$ to $u$
        is even, and $u\mapsto h_1$ otherwise, respects the lists $L(u)$
        for each $u\in A_B(g)$, and it defines a homomorphism
        $\bG[A_B(g)]\to \bH[h_1,h_2]$.\\
        \textbf{otherwise}, remove $h_2$ from $L(g)$.\\
    }
    \textbf{Return:} $\bG[S]$ with the updated lists $L(u)$. 
\end{algorithm}

\subsection{Homogeneous concatenation}
\label{sec:hom-con}

We introduce the following operation on $2$-edge-coloured graphs. 
Given a pair of $2$-edge-coloured graphs $\bH$ and $\mathbb A$
the \emph{homogeneous concatenation} $\bH \triangleleft_h \bA$
is the structure obtained from the disjoint union $\bH +\bA$
by adding edges between $H$ and $A$ so that $A$ a homogeneous set
in $\bH \triangleleft_h \bA$. That is, for $h\in H$ and $a\in A$
there is a blue edge $ha\in B(\bH \triangleleft_h \bA)$ if
$aa\in B(\bA)$, and a red edge if $aa\in R(\bA)$.

Since the list version of each $2$-edge-coloured
graph on at most $2$-vertices can be solved in polynomial time, 
the next claim follows by recursively calling 
Algorithms~\ref{alg:universal-vertex}, \ref{alg:universal-mono-edge},
and~\ref{alg:universal-split}. 

\begin{theorem}\label{thm:homogeneous-Ptime}
    The list version of $\CSP(\bH)$ can be solved in polynomial time,
    whenever $\bH$ is a reflexive complete $2$-edge-coloured  graph
    that admits a decomposition
    \[
        \bH:= \bA_1 \triangleleft_h \dots \triangleleft_h \bA_k,
    \]
    where $\bA_1$ is a $2$-element structure, and $\bA_i$
    is a single vertex (with a loop), or a reflexive $\ast$-edge
    for each $2\le i\le k$.
\end{theorem}

For a reflexive complete 2-edge-coloured graph $\bH$, one can efficiently 
check if the core of $\bH$ admits a decomposition as in 
Theorem~\ref{thm:homogeneous-Ptime}.

\begin{proposition}\label{prop:recognizing_tractables_in_ptime}
    For every reflexive complete 2-edge-coloured graph $\bH$,
    one can check in time polynomial in $|H|$ that the core of $\bH$ 
    admits a decomposition of the form
    $\bA_1 \triangleleft_h \dots \triangleleft_h \bA_k$,
    where $\bA_1$ is a structure on at most two elements, and each
    $\bA_i$ is a single vertex (with a loop), or a reflexive $\ast$-edge
    for each $2\le i\le k$.
\end{proposition}
\begin{proof}
    If $\bH$ contains an $\ast$-loop, we simply accept because its core
    a structure with exactly one element, and clearly admits such a decomposition.
    We now suppose that $\bH$ has no $\ast$-loops.
    We will now use the following facts: (1) one can check in polynomial time
    if the core of $\bH$ has size at most  two (because the CSP of each reflexive
    complete $2$-edge-coloured graph on at most two vertices is solvable in polynomial time);
    and similarly, (2) one can check in polynomial time if the core of $\bH$ is either 
        an $\ast$-edge or a single vertex.

The algorithm is now quite simple. On input $\bH$ we compute a minimal (by inclusion) non-empty
homogeneous set $S\subseteq H$. To do so, for each $v\in H$ compute the minimal homogeneous
set $S_v$ containing $v$: at first, add $v$ to $S_v$; then,  while there are $u\in S_v$
and $w\in H\setminus S_v$ such that there is an edge connecting $uw$ of different colour to
the loop $uu$, add $w$ to $S_v$.
Among all homogeneous sets $S_v$ choose a minimal one, and denote it by $S$.
If $S\neq H$, check if
the substructure $\bH[S]$ induced by $S$ is either a single vertex or an $\ast$-edge. If not, reject;
otherwise, remove $S$ from $H$ and repeat the procedure until the smallest homogeneous set $S = H$.
In this case, check if the core of $\bH$ has size at most $2$. In the positive case we accept, and we
reject otherwise. This whole procedure can clearly be done in polynomial time.
\end{proof}

\subsection{Datalog}
\label{sect:datalog}

Many readers may have noticed that the spirit of
Algorithms~\ref{alg:universal-vertex}--\ref{alg:universal-split} 
is some sort of consistency checking. Moreover, the reader familiar
with \emph{Datalog} may have noticed that these algorithms can 
be encoded as Datalog programs. Datalog is a framework commonly
used in constraint satisfaction theory and database theory, which
encompasses a concrete notion of local consistency checking. In
particular, Datalog programs run in polynomial time, and contrary
to polynomial-time algorithms, it is known that Datalog programs
cannot solve all CSPs (e.g., neither $\CSP(\bK_3)$ nor linear
equations over finite fields can be solved by Datalog~\cite[Theorem 31]{FederVardi}).

We now reprove (and improve) 
Theorem~\ref{thm:homogeneous-Ptime}, by showing that the CSP of every
structure $\bH$ in the scope of this theorem can be solved by a Datalog program. 
However, instead of formally introducing Datalog, and rewriting
Algorithms~\ref{alg:universal-vertex}--~\ref{alg:universal-split}
as Datalog programs, we build on the algebraic approach to CSPs
and use the following theorem as the definition of finite-domain
CSPs solved by Datalog --- we refer the reader to~\cite{BoundedWidthJournal,Pol}
for a detailed exposition on CSPs solvable by Datalog.

\begin{theorem}[\hspace{-0.8pt}{\cite{BoundedWidthJournal, BartoWidth, Maltsev-Cond}, see also Theorem 47 in~\cite{Pol}}]\label{thm:wnu34}
    For any finite structure $\bH$, the following are equivalent:
    \begin{enumerate}
        \item  $\CSP(\bH)$  is \emph{solved by Datalog};
        \item $\bH$ has a ternary WNU polymorphism $f_3$ and a
         4-ary WNU polymorphism $f_4$ such that
        \begin{equation}\label{eq:wnu34}
        f_3(y,x,x)=f_4(y,x,x,x)
        \end{equation}
    \end{enumerate}
\end{theorem}

\begin{remark}\label{rmk:list-wnu4}
    The list version of $\CSP(\bH)$ corresponds to the CSP
    of $(\bH, \{U_X\}_{\varnothing\neq X\subseteq H})$, i.e., $\bH$
    together with a unary predicate for each non-empty subset $X$ of $H$
    (see also~\cite{Conservative}). 
    Clearly, a polymorphism $f\colon \bH^m\to \bH$ is a polymorphism of 
    $(\bH, \{U_X\}_{\varnothing\neq X\subseteq H})$ if and only if 
    $f$ is \emph{conservative}, i.e., $f(x_1,\dots, x_m)\in\{x_1,\dots, x_m\}$
    for every $(x_1,\dots, x_m)\in H^m$. Hence, the list version of
    $\CSP(\bH)$ is solved by Datalog if and only if
    $\bH$ has a ternary WNU conservative polymorphism $f_3$ and a
    conservative 4-ary WNU polymorphism $f_4$ satisfying equation~(\ref{eq:wnu34})
\end{remark}

\begin{theorem}\label{thm:bounded_width}
    Let $\bH$ be a reflexive complete $2$-edge-coloured  graph. 
    If $\bH$ admits a decomposition
    \[
        \bH:= \bA_1 \triangleleft_h \dots \triangleleft_h \bA_k,
    \]
    where $\bA_1$ is a $2$-element structure, and each
    $\bA_i$ is a single vertex (with a loop), or a reflexive $\ast$-edge
    for each $2\le i\le k$,
    then the list version of $\CSP(\bH)$ is solved by Datalog.
\end{theorem}
\begin{proof}
    By Theorem~\ref{thm:wnu34}  and Remark~\ref{rmk:list-wnu4}, it suffices to find a
    pair of WNU conservative polymorphisms $f_3,f_4$ that satisfy~(\ref{eq:wnu34}). 
    Let $m\in\{3,4\}$ denote the arity of the polymorphisms.
    Call a tuple $\bar a\in H^m$ \emph{united}, if for some $i\in[k]$, we have that
    $\bar a\in A_i^m$.    Otherwise, a tuple is called \emph{scattered}.

    For every united tuple $\bar a\in  A_i^m$  for
    some $i\in[k]$, let $f_3(\bar a)$ act as 
    the unique majority  on $A_i$, and $f_4(\bar a)$ act as near unanimity 
    on tuples having at least three same entries
    (i.e., $f_4(y,x,x,x) = \dots = f(x,x,x,y) = f(x,x,x,x)$),
    and as the projection on the first coordinate otherwise.
    Clearly, this is already sufficient for $f_3,f_4$ to satisfy (\ref{eq:wnu34})
    restricted to united tuples.

    For every $\bar a\in H^m$, let $\max(\bar a)$ stand for the maximal possible
    $j\in[k]$ such that 
     $\bar a$ contains a coordinate from $A_j$, i.e., 
    $a_i\in A_j$ for some $i\in[m]$.
    For every scattered tuple $\bar a\in H^m$, let $f_m(\bar a)$ be equal to $a_i$,
    where $i\in[m]$ is  the first coordinate 
    such that $a_i\in A_{\max(\bar a)}$.  It follows from the definition
    of $f_3$ and $f_4$ on scattered tuples, that they satisfy~(\ref{eq:wnu34})
    on scattered tuples as well. Since $(y,x,x)$ is united (resp.\ scattered) if and only if
    $(y,x,x,x)$ is united (resp.\ scattered) , we conclude that $f_3$ and $f_4$
    satisfy~(\ref{eq:wnu34}).
    Notice that both $f_3,f_4$ are conservative.

    It remains to show that $f_3$ and $f_4$ are polymorphisms of $\bH$.
    Let $\bar a,\bar b\in H^m$ be two tuples,  we show that $f_m$ 
    preserves the edge types between $\bar a$ and $\bar b$ for $m\in\{3,4\}$, i.e., if
    $\bar a\bar b\in B(\bH^m)$ (resp.\ in $R(\bH^m)$), then
    $f_m(\bar a)f_m(\bar b)\in B(\bH)$ 
    (resp.\ in $R(\bH)$). Consider the following 
    case distinction that depends on whether $\bar a$ and $\bar b$ are
    scattered or united, and on the relation between $\max(\bar a)$ and
    $\max(\bar b)$.
    \begin{itemize}
        \item We begin with all cases where $i = \max(\bar a)  >  \max(\bar b)$.
        Let $j\in[m]$ be the first coordinate of $\bar a$ such that $a_j\in A_i$. 
        Since $A_i$ is homogeneous in $\bA_1\triangleleft_h\dots \triangleleft_h\bA_i$,
        and $b_j\in A_{i'}$ for some $i' < i$, it must be the case that $a_j$
        is connected to $b_j$ by at most one edge type in $\bH$. Moreover, 
        this edge type depends on whether $a_ja_j\in B(\bH)$ or $a_ja_j\in R(\bH)$.
        Without loss of generality assume that $a_ja_j\in B(\bH)$, and 
        so $\bar a\bar b\in B(\bH^m)$ and $\bar a\bar b\not\in R(\bH^m)$.
        Also notice that $a_j$ is connected by a blue edge to each
        entry $b_{j'}$ of $\bar b$ (by homogeneity of $\bA_i$).
        \begin{itemize}
            \item \textbf{Assume $\bar a$ is united.} In this case, 
            each coordinate $a_{j'}$ ($j'\in[m]$)  of $\bar a$  belongs to $A_i$, 
            and by similar arguments as before, $a_{j'}$ is connected
            to each entry $b_{j''}$  of $\bar b$ by a blue edge in $\bH$. 
            Since $f(\bar a)$ equals some entry of $\bar a$, and $f(\bar b)$
            is some entry of $\bar b$, we conclude that $f(\bar a)$ and $f(\bar b)$
            are connected by a blue edge in $\bH$.
            \item \textbf{Assume $\bar a$ is scattered.} Since $j\in[m]$ is the
            first coordinate such that $a_j\in A_i$, it follows that $f(\bar a) = a_j$. 
            We already argued that $a_j$ is connected to each $b_{j'}$ by a blue
            edge for $j'\in[m]$, and since $f(\bar b) = b_{j'}$ for some $j'\in [m]$, 
            it follows that $f(\bar a) f(\bar b)\in B(\bH)$. 
        \end{itemize}
        This shows that $f$ preserve the edge types connecting tuples
        $\bar a,\bar b\in \bH^m$ whenever $\max(\bar a)  >  \max(\bar b)$, 
        and the tuples where $\max(\bar a)  <  \max(\bar b)$ follow with
        symmetric arguments.

        \item We now consider the cases where $\max(\bar a) = \max(\bar b)$,
        and we distinguish whether both are united, both are scattered, 
        or one is united and the other is a scattered tuple.
        \begin{itemize}
            \item \textbf{Assume $\bar a$ and $\bar b$ are united.}
            For such tuples, $f_3$ is defined as the majority 
            polymorphism, and every component
            $\bA_i$ is preserved by this polymorphism, so $f_3$ preserves the
            edge types between such tuples. 
            For $f_4$, we show that there is always some coordinate
            $i\in[4]$ such that $f_4(\bar a) = a_i$ and $f_4(\bar b) = b_i$, and so, 
            it follows that $f_4$ also preserves the edge types connecting $\bar a$ and $\bar b$.
            This is clear when each $\bar a$ and $\bar b$ 
            have at least 3 equal 
            entries (e.g., $\bar a = (a,a,a,a')$ and $\bar b = (a',a',a,a')$).
            If $\bar a$ has at least 3 same entries and $\bar b$ has at most 2 same entries, then
            $f_4(\bar b)=b_1$. In this case, either $f_4(\bar a)=a_1$ and we are done, or $\bar a = (a',a,a,a)$.
            Since $\bar b$ is united and no three coordinates have the same value, it must
            be the case that the equality $b_i = b_1$ holds for exactly one $i\in\{2,3,4\}$.
            It thus follows that $(f_4(\bar a),f_4(\bar b)) = (a,b_1) = (a_i,b_i)$
            for some $i\in\{2,3,4\}$.
            If both $\bar a$ and $\bar b$ have at most 2 same entries, then $f_4$ works as the projection on
            the first coordinate so we are done.

        \item \textbf{Assume $\bar a$ and $\bar b$ are scattered.} 
            Let $i=\max(\bar a)=\max(\bar b)$. Notice that $i>1$ because $\bar a$ and $\bar b$
            are scattered.    If, for some $j\in[m]$, both $a_j$ and $b_j$ are the first elements
            to be in $\bA_i$, then we are done because $f_m(\bar a)=a_j$ and $f_m(\bar b)=b_j$.
            Otherwise, suppose that $a_j=f_m(\bar a)$ is the first element of $\bar a$ to be in
            $\bA_i$ while, for each $j'\leq j$, $b_{j'}$ is in $A_{i'}$, for  some $i'<i$.
            Assume without loss of generality that $a_ja_j\in B(\bH)$.
            Then, by homogeneity, $a_jb_j\in B(\bH)$ and $a_jb_j\not\in R(\bH)$,
            and so, $\bar a$ and $\bar b$ are connected by a blue edge and not by
            a red edge. Now, if $f_m(\bar b) = f_m(\bar a) = a_j$, then $f_m(\bar a) f_m(\bar b)\in B(\bH)$. 
            Alternatively, $f_m(\bar b)\in A_i\setminus \{a_j\}$ so $|A_i| = 2$, and  since $i > 1$,
            it must be the case the case that $\bA_i$ is a reflexive $\ast$-edge. In particular, 
            there is a red and a blue edge connecting the two different vertices in $A_i$, 
            so $f(\bar a) f(\bar b)\in B(\bH)$.
   
    \item \textbf{Assume $\bar a$ is scattered and $\bar b$ is united.} 
    Let $i=\max(\bar a)=\max(\bar b)$ and let $a_j=f_m(\bar a)$ be the first coordinate to be in $A_i$.
    As $\bar a$ is scattered, we have $i>1$, and so, $\bA_i$ consists of a single
    vertex, or a reflexive $\ast$-edge. Also,
    for some $j'\in[m]$, we have that $a_{j'}\in A_{i'}$ for $i'<i$, which implies that the edge $a_{j'}b_{j'}$ belongs either to $B(\bH)$ or to $R(\bH)$ but not to both at the same time.
    We assume without loss of generality that $a_{j'}b_{j'}\in B(\bH)$.
     Notice that if $a_ja_j\in B(\bH)$ or if $f(\bar a)\neq f(\bar b)$, then
     $f(\bar a)f(\bar b)\in B(\bH)$:  indeed, if $f(\bar b) = f(\bar a) = a_j$, 
     then $f(\bar a)f(\bar b)= a_ja_j\in B(\bH)$; if $f(\bar a)\neq f(\bar b)$, 
     then $\bA_i$ must contain two elements connected by an $\ast$-edge, and
     $f(\bar a)f(\bar b)\in B(\bH)$.
     Hence, it suffices to show that if $a_ja_j\in R(\bH)$, then $f(\bar a) \neq f(\bar b)$.
     Recall that $a_{j'}b_{j'} \in B(\bH)$, and since $a_{j'}\in A_{i'}$ for some
     $i' < i$, it follows by homogeneity of $A_i$ that $b_{j'}b_{j'}\in B(\bH)$,
     so $b_{j'}\neq a_j$.
     Also, since $a_jb_j\in B(\bH)$ and $a_ja_j\in R(\bH)$, it must be the case that
     $b_j\neq a_j$. From the definitions of $j$ and $j'$, we know that $j\neq j'$, which
     implies that $\bar b$ contains at least two entries different from $a_j$. Finally, 
     since each $a_\ell$ is connected to $b_\ell$ by a blue edge, and if $\ell < j$, then
     $a_\ell\in A_{i'}$ for some $ i'< i$, then $b_\ell b_\ell \in B(\bH)$. All these observations
     together show that $b_1\neq a_j$,  and that
     $\bar b$ contains at least two entries different from $a_j$. It now follows from
     the definitions of $f_3$ and $f_4$, that they either act as a near unanimity
     on $\bar b$, or a projection on the first coordinate, and in either case
     $f(\bar b) = b_1$, and $b_1\neq a_j$.
    \end{itemize}
    \end{itemize}
\end{proof}

\section{Full-homomorphisms}
\label{sec:full-hom}

In this brief section, we apply Theorem~\ref{thm:homogeneous-Ptime}
to obtain a P vs.\ NP-complete classification of the complexity
of sandwich problems for \emph{full-homomorphisms}.
A full-homomorphism from a graph $\bG$ to a graph $\bH$ is a vertex mapping
$f\colon G \to H$ such that  $uv\in E(\bG)$ if and only if $f(u)f(v)\in E(\bH)$
for every $u,v\in G$. Notice that if $M$ is the adjacency matrix of $\bH$, then $\bG$
admits an $M$-partition if and only if $\bG$ admits a full-homomorphism to $\bH$.

A graph $\bG$ is  \emph{point-determining} if for every pair of different vertices
$u$ and  $v$
there is vertex $w$ such that $uw\in E$ and $vw\not\in E$ or vice versa.
A pair of different vertices $i$ and $j$ of $\bG$ are \emph{twins} if for every $w\in V$
there is an edge $uw$ in $\bG$ if and only if there is an edge $vw$ in $\bG$;
so $\bG$ is point-determining if and only if it contains no twin vertices.
We say that $\bG'$ is a \emph{blow-up} of $\bG$ if $\bG'$ can be obtained from $\bG$ by
iteratively adding twins. Notice that $\bG$ admits a full-homomorphism to $\bH$ if and only
if $\bG$ is a blow-up of some induced subgraph $\bH'$ of $\bH$.

\begin{observation}\label{obs:K2+K1}
    A graph $\bG$ admits a full-homomorphism to $\bK_1 + \bK_2$ if and only
    if $\bG$ is a $\{\bK_3,2\bK_2,\bP_4\}$-free graph.
\end{observation}
\begin{proof}
    Clearly, if $\bG$ is an induced subgraph of $\bH$, and $\bH$ admits a 
    full-homomorphism to $\bK_1+\bK_2$, then $\bG$ admits a full-homomorphism
    to $\bK_1+\bK_2$ as well.
    Hence, the backward direction follows because  neither
    $\bK_3,2\bK_2,$ nor $\bP_4$ admits a full-homomorphism to $\bK_1+\bK_2$.
    Suppose now
    that $\bH$ is $\{\bK_3,2\bK_2,\bP_4\}$-free. In particular, $\bH$ is a cograph
    because it is $\bP_4$-free, and so, $\bH$ is either the join or the disjoint union
    of two smaller cographs. In the former case, since $\bH$ is triangle-free, 
    $\bH$ is the join of two independent sets, and thus it admits a full-homomorphism
    to $\bK_2$ (and thus to $\bK_1 +\bK_2$). In the latter case, when $\bH$ is the disjoint
    union of smaller cographs $\bH_1,\dots, \bH_k$, it follows from the $2\bK_2$-freeness
    of $\bH$ that at most one of these $\bH_i$ contains an edge. If no such $i$ exists, then
    $\bH$ is an empty graph, i.e., a blow-up of $\bK_1$; if such $i$ exists, the we
    proceed similarly as before to notice that $\bH_i$ is a blow-up of $\bK_2$, and
    thus $\bH$ is a blow-up of $\bK_1+\bK_2$. 
\end{proof}

Full-homomorphism of $2$-edge-coloured graphs are defined analogously: 
a \emph{full-homomorphism} $f\colon \bG\to \bH$ is a vertex mapping
$f\colon G\to H$ such that for every pair of vertices $u,v\in G$
it is the case that $uv\in B(\bG)$ if and only if $f(u)f(v)\in B(\bH)$, 
and $uv\in R(\bG)$ if and only if $f(u)f(v)\in R(\bH)$.

\begin{lemma}\label{lem:full-hom-LAC}
    Let $\bG$ and $\bH$  be  $2$-edge-coloured graphs. If 
    there is a full-homomorphism $f\colon \bG\to \bH$, then 
    the list version $\CSP(\bG)$ reduces in polynomial time to
    the list version of $\CSP(\bH)$.
\end{lemma}
\begin{proof}
    Associate an input $\mathbb I$ with lists $L(i)\subseteq G$
    for each $i$, with the input $\mathbb I$ with lists 
    $L'(i) = \{f(g)\colon g\in G\}$.
    This results in a polynomial-time reduction from the list version
    of $\CSP(\bG)$ to the list version of $\CSP(\bH)$. 
\end{proof}

\begin{figure}[ht!]
\centering
    \begin{tikzpicture}
    
        \begin{scope}[scale=0.7]
            \node (L1) at (0,-1.6) {$\bK_1 + \bK_2$};
            \node (0) [vertex] at (90:1.3){};
            \node (1) [vertex] at (210:1.3){};
            \node (2) [vertex] at (330:1.3){};
            
            \draw [edge] (1) to  (2);
            
        \end{scope}

        \begin{scope}[scale=0.7, xshift = 6.6cm]
           \node (L1) at (0,-1.6) {$(\bK_1 + \bK_2)^\ast$};
            \node (0) [vertex] at (90:1.3){};
            \node (1) [vertex] at (210:1.3){};
            \node (2) [vertex] at (330:1.3){};
            
            \draw [edge, red, dashed] (0) to [out=55,in=125, looseness=12] (0);
            \draw [edge, red, dashed] (1) to [out=55 + 90,in=125 + 90, looseness=12] (1);
            \draw [edge, red, dashed] (2) to [out=55 -90,in=125 + -90, looseness=12] (2);
               
            \foreach \from/\to in {0/1, 2/0} 
                \draw [edge, red, dashed] (\from) to [bend right = 20] (\to);
            \foreach \from/\to in {2/1} 
                \draw [edge, blue] (\from) to [bend left = 20] (\to);
        \end{scope}
        
        \begin{scope}[scale=0.7, xshift=13cm]
            \node (L1) at (0,-1.6) {$\overline{(\bK_1 + \bK_2)^\ast}$};
            \node (0) [vertex] at (90:1.3){};
            \node (1) [vertex] at (210:1.3){};
            \node (2) [vertex] at (330:1.3){};
            
            \draw [edge, blue] (0) to [out=55,in=125, looseness=12] (0);
            \draw [edge, blue] (1) to [out=55 + 90,in=125 + 90, looseness=12] (1);
            \draw [edge, blue] (2) to [out=55 -90,in=125 + -90, looseness=12] (2);
               
            \foreach \from/\to in {0/1, 2/0} 
                \draw [edge, blue] (\from) to [bend right = 20] (\to);
            \foreach \from/\to in {2/1} 
                \draw [edge, red, dashed] (\from) to [bend left = 20] (\to);
                
        \end{scope}

\end{tikzpicture}
\caption{To the left, the (uncoloured) graph $\bK_1 + \bK_2$. In the middle and right, two 
$2$-edge-coloured graphs obtained from $\bK_1 + \bK_2$: the former by colouring edges with blue
and non-edges with red,  and the latter by colouring edges with red and non-edges with blue.}
\label{fig:full-hom}
\end{figure}
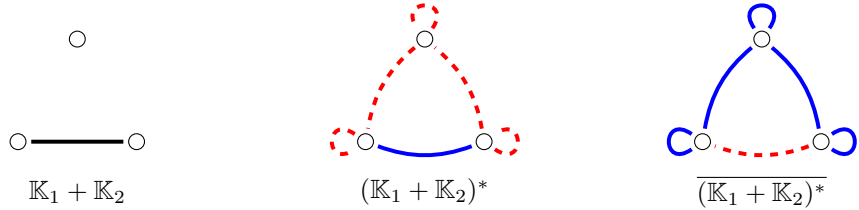

\begin{proposition}\label{prop:full-hom}
    Let $\bH$ be a reflexive complete $2$-edge-coloured graph with no $\ast$-loops and
    where all vertices are incident to a red loop (resp.\ to a  blue loop).
    If $\bH$ does not pp-construct $\bK_3$, then $\bH$ admits a full-homomorphism
    $(\bK_1 + \bK_2)^\ast$ (resp.\ to $\overline{(\bK_1+ \bK_2)^\ast}$); see
    Figure~\ref{fig:full-hom}. In this case, the list version of $\CSP(\bH)$ is
    polynomial-time solvable.
\end{proposition}
\begin{proof}
    We proceed by contrapositive, and assume that $\bH$ does
    not admit a full-homomorphism to $(\bK_1 + \bK_2)^\ast$. Using
    Observation~\ref{obs:K2+K1} it follows that one of $\bK_3^\ast$, 
    $(2\bK_2)^\ast$ or $\bP_4^\ast$ are induced substructures.
    Consider now the following pp-definition $\delta_E(x,y):=\exists z.\; R(x,z)
    \land B(y,z)$. Since $\bH$ has no $\ast$-edges, then the digraph $(H, E)$
    where $(x,y)\in E$ if and only if $\bH\models \delta_E(x,y)$ is a loopless digraph
   In each case, $(H,E)$ contains a smooth digraph that does not map homomorphically
    to a directed cycle. Hence, by Theorem~\ref{thm:HH-digraphs},  we conclude
    that $(H,E)$ pp-constructs $\bK_3$.
    Since $\bH$ pp-defines $(H,E)$ we conclude that $\bH$ pp-constructs $\bK_3$ as well.
    By Lemma~\ref{lem:full-hom-LAC}, if $\bH$ admits a full-homomorphism
    to $(\bK_1+ \bK_2)^\ast$, then the list version of $\CSP(\bH)$ reduces in 
    polynomial time to the list
    version of $\CSP((\bK_1 + \bK_2)^\ast)$. The latter can be solved in polynomial
    time by Theorem~\ref{thm:homogeneous-Ptime}, because $(\bK_1 + \bK_2)^\ast$
    can be constructed from a reflexive complete $2$-edge-coloured graph on two vertices
    ($\bK_2^\ast$) by adding a homogeneous vertex.
\end{proof}

\section{Hereditary pp-constructions}
\label{sec:HH}

The aim of this section is obtaining small $2$-edge-coloured graphs
$\bH$ such that whenever a reflexive complete $2$-edge-coloured  graph
$\bH'$ contains $\bH$ as a (not necessarily induced) substructure,
then $\bH$ pp-constructs $\bK_3$.

\emph{Hereditarily hard digraphs} were considered by Bang-Jensen, Hell, and MacGillivray 
in~\cite{bangjensenDM138}. A digraph $\mathbb D$ is hereditarily hard if 
$\CSP(\mathbb D)$ is NP-complete, and $\CSP(\bH)$ is $\NP$-complete whenever $\bH$ is a
loopless digraph such that $\mathbb D\to \bH$. In particular, 
undirected cycles of odd length are hereditarily hard. 

In order for hereditarily hardness to be a meaningful definition, one
must assume that $\cP\neq \NP$. In this paper, we avoid complexity theoretic 
assumptions, and thus consider the following variant. A digraph
$\bD$ \emph{hereditarily pp-constructs} $\bK_3$ if $\bD$ pp-constructs
$\bK_3$, and whenever $\bH$ is a loopless digraph such that $\bD\to \bH$,
then $\bH$ pp-constructs $\bK_3$. The following theorem is an immediate
consequence of the dichotomy for smooth digraphs~\cite{BartoKozikNiven}.

\begin{theorem}\label{thm:HH-digraphs}
    For every smooth digraph $\mathbb D$ one of the following statements holds:
    \begin{itemize}
        \item Either $\bD\to \vec{\bC}_n$ (the directed cycle of length $n$) for some positive integer $n$, 
        and in this case $\bD$ does not hereditarily pp-construct $\bK_3$;
        \item otherwise, $\bD$ hereditarily pp-constructs $\bK_3$. 
    \end{itemize}
\end{theorem}

Here we consider a similar notion for $2$-edge-coloured graphs. 
We say that a $2$-edge-coloured graph $\bH$ \emph{hereditarily pp-constructs} $\bK_3$
if $\bH$ pp-constructs $\bK_3$, and $\bH'$ pp-constructs $\bK_3$
whenever $\bH\to \bH'$ and $\bH'$ does not contain an $\ast$-loop.
For instance,  suppose that $\mathbb H$ contains an \emph{$\ast$-odd-cycle}, i.e., 
a sequence of vertices $h_1,\dots, h_{2n+1}$ such that
$h_1h_{2n+1}\in R\cap B$ and $h_ih_{i+1}\in R\cap B$ for every $i\in[2n]$.
In any such case, if $\mathbb H$ is loopless, then $\mathbb H$ hereditarily pp-constructs $\bK_3$:
indeed, if $\mathbb H\to \mathbb H'$ and $\mathbb H'$ has no $\ast$-loops, then the formula
$\delta_E(x,y):= R(x,y)\land B(x,y)$ pp-defines a loopless non-bipartite graph, and
any loopless non-bipartite graph pp-constructs $\bK_3$.

\begin{observation}\label{obs:red-odd-cycle}
    Let $\bH$ be a $2$-edge-coloured graph without $\ast$-loops. If there is a
    red-odd-cycle (resp.\ blue-odd-cycle) $v_1,\dots, v_{2n+1}$ in $\bG$ such that $v_iv_i\in B$
    (resp.\ $v_iv_i\in R$) for every $i\in[2n+1]$, then $\bH$ hereditarily pp-constructs $\bK_3$. 
\end{observation}
\begin{proof}
     Suppose that $\bH\to \bH'$ and that $\bH'$ has no $\ast$-loops.
     It follows from the assumption on $\bH$ that the undirected graph $(H', E)$
     where $E$ is the set of pairs $(x,y)$ that satisfy the pp-formula
     $ B(x,x)\wedge B(y,y)\wedge R(x,y)$ is a non-bipartite graph. Hence, 
    $\bH'$ pp-constructs a non-bipartite loopless graph, and thus it pp-constructs
    $\bK_3$.
\end{proof}

We now consider the equivalence relation $\sim_S$ defined on $4$-tuples of vertices
defined by the Siggers identity, i.e., the smallest equivalence
relation that contains $\sim'$ where
\[
    (x,y,z,w)\sim' (x',y',z',w'):= (x,y,z,w) = (a,r,e,a) \text{ and } (x',y',z',w') = (r,a,r,e).
\]
Given a $2$-edge-coloured graph $\mathbb H$ we define the \emph{Siggers power}
$\Sig(\mathbb H)$ of $\mathbb H$ to be the quotient graph $\mathbb H^4/{\sim_S}$, i.e., the fourth power of $\mathbb H$ 
factor by the equivalence relation defined by the Siggers identity.
So a Siggers polymorphism is a homomorphism $f\colon\Sig(\bH)\to \bH$.

\begin{lemma}\label{lem:siggers-power}
    If $\bH$ has no $\ast$-loops, and $\Sig(\bH)$ contains an $\ast$-loop or an $\ast$-odd-cycle, 
    then $\bH$ hereditarily pp-constructs $\bK_3$. 
\end{lemma}
\begin{proof}
    Suppose that $g\colon \bH\to \bH'$, and anticipating a contradiction suppose that
    $\bH'$ does not contain an $\ast$-loop, and that $\bH'$ does not pp-construct $\bK_3$.
    It follows by Theorem~\ref{thm:CSP-dichotomy} that $\bH'$ has a 
    Siggers polymorphism, i.e., a homomorphism $f\colon \Sig(\bH')\to \bH'$. 
    Since $\bH\to \bH'$, it is also the case that $\Sig(\bH)\to \Sig(\bH')$
    (the mapping $g'\colon\Sig(\bH)\to \Sig(\bH')$ is obtained by applying $g$
    component wise is a homomorphism). 
    This implies that either $\bH'$ contains an $\ast$-loop or an $\ast$-odd-cycle, 
    and since $\bH'$ is has no $\ast$-loops, we conclude that $\bH'$ pp-defines
    a loopless non-bipartite graph via $E(x,y):= R(x,y)\land B(x,y)$.
    This implies that $\bH'$ pp-constructs $\bK_3$  contradicting our assumption.
\end{proof}

Similarly,  for a positive integer $p\ge 2$, we define the equivalence
relation $\sim_p$ to be the smallest equivalence relation that contains
$\sim_p'$ where
\[
    (x_1,\dots, x_p)\sim_p'(y_1,\dots, y_p)\Leftrightarrow x_2 = y_1, \dots x_p = y_{p-1} \text{ and } x_1 = y_p.
\]
Given a $2$-edge-coloured graph $\bH$ we define the \emph{$p$-cyclic power}
$\Cyc_p(\mathbb H)$ of $\bH$ to be the quotient graph $\bH^p/{\sim_p}$.
With analogous arguments as above, one can prove the following lemma. 

\begin{lemma}\label{lem:cyclic-power}
    If $\bH$ has no $\ast$-loops, and $\Cyc_p(\bH)$ contains an
    $\ast$-loop or an $\ast$-odd-cycle for some prime $p> |H|$,
    then $\bH$ is a hereditarily hard $2$-edge-coloured graph. 
\end{lemma}

We apply Lemma~\ref{lem:siggers-power} to show that the $2$-edge-coloured
reflexive graphs from Figure~\ref{fig:3-element} (and their duals)
hereditarily pp-construct $\bK_3$.

\begin{figure}[ht!]
\centering
    \begin{tikzpicture}
    
        \begin{scope}[scale=0.7]
            \node (L1) at (0,-1.6) {(3A)};
            \node (0) [vertex, label = left:{\scriptsize $0$}] at (90:1.3){};
            \node (1) [vertex, label = below:{\scriptsize $1$}] at (210:1.3){};
            \node (2) [vertex, label = below:{\scriptsize $2$}] at (330:1.3){};
            
            \draw [edge, red, dashed] (0) to [out=55,in=125, looseness=12] (0);
             \draw [edge, red, dashed] (1) to [out=55 + 90,in=125 + 90, looseness=12] (1);
            \draw [edge,  red, dashed] (2) to [out=55 -90,in=125 + -90, looseness=12] (2);
               
            \foreach \from/\to in {0/1, 1/2, 2/0} 
                \draw [edge, blue] (\from) to [bend right = 20] (\to);
            
        \end{scope}

        \begin{scope}[scale=0.7,xshift=7.5cm]
           \node (L1) at (0,-1.6) {(3B)};
            \node (0) [vertex, label = left:{\scriptsize $0$}] at (90:1.3){};
            \node (1) [vertex, label = below:{\scriptsize $1$}] at (210:1.3){};
            \node (2) [vertex, label = below:{\scriptsize $2$}] at (330:1.3){};
            
            \draw [edge, red, dashed] (0) to [out=55,in=125, looseness=12] (0);
            \draw [edge, red, dashed] (1) to [out=55 + 90,in=125 + 90, looseness=12] (1);
            \draw [edge, blue] (2) to [out=55 -90,in=125 + -90, looseness=12] (2);
               
            \foreach \from/\to in {0/1, 2/0, 1/2} 
                \draw [edge, blue] (\from) to [bend right = 20] (\to);
            \foreach \from/\to in {1/2} 
                \draw [edge, red, dashed] (\from) to [bend left = 20] (\to);
        \end{scope}
    
        \begin{scope}[scale=0.7, xshift = 15cm]
            \node (L1) at (0,-1.6) {(3C)};
            \node (0) [vertex, label = left:{\scriptsize $0$}] at (90:1.3){};
            \node (1) [vertex, label = below:{\scriptsize $1$}] at (210:1.3){};
            \node (2) [vertex, label = below:{\scriptsize $2$}] at (330:1.3){};
            
            \draw [edge, red, dashed] (0) to [out=55,in=125, looseness=12] (0);
            \draw [edge, red, dashed] (1) to [out=55 + 90,in=125 + 90, looseness=12] (1);
            \draw [edge, blue] (2) to [out=55 -90,in=125 + -90, looseness=12] (2);
               
            \foreach \from/\to in {0/1, 1/2} 
                \draw [edge, blue] (\from) to [bend right = 20] (\to);
            \foreach \from/\to in {1/2, 0/2} 
                \draw [edge, red, dashed] (\from) to [bend left = 20] (\to);
        \end{scope}

        \begin{scope}[scale=0.7, yshift = -5.1cm, xshift = 3.25cm]
            \node (L1) at (0,-1.4) {(4A)};
            \node (0) [vertex, label = right:{\scriptsize $0$}] at (1,2){};
            \node (1) [vertex, label = left:{\scriptsize $1$}] at (-1,2){};
            \node (2) [vertex, label = left:{\scriptsize $2$}] at (-1,0){};
            \node (3) [vertex, label = right:{\scriptsize $3$}] at (1,0){};
            
            \draw [edge, blue] (0) to [out=55,in=125, looseness=12] (0);
             \draw [edge, blue] (1) to [out=55,in=125, looseness=12] (1);
            \draw [edge,  red, dashed] (2) to [out=55 - 180,in=125 + -180, looseness=12] (2);
            \draw [edge,  red, dashed] (3) to [out=55 - 180,in=125 + -180, looseness=12] (3);

            \draw [edge, red, dashed] (2) to [bend left = 20] (3);
            \draw [edge, blue] (2) to [bend right = 20] (3);
            
            \foreach \from/\to in {0/1, 1/2, 1/3} 
                \draw [edge, red, dashed] (\from) to  (\to);
            \foreach \from/\to in {0/2, 0/3} 
                \draw [edge, blue] (\from) to  (\to);
            
        \end{scope}

        \begin{scope}[scale=0.7, yshift = -4cm, xshift = 11cm]
           \node (L1) at (0,-2.5) {(5A)};
            \node (0) [vertex, label = left:{\scriptsize $0$}] at (90:1.5){};
            \node (1) [vertex, label = left:{\scriptsize $1$}] at (162:1.5){};
            \node (2) [vertex, label = left:{\scriptsize $2$}] at (234:1.5){};
            \node (3) [vertex, label = right:{\scriptsize $3$}] at (306:1.5){};
            \node (4) [vertex, label = right:{\scriptsize $4$}] at (18:1.5){};

            \draw [edge, blue] (0) to [out=55,in=125, looseness=12] (0);
            \draw [edge, red, dashed] (1) to [out=55, in=125, looseness=12] (1);
            \draw [edge, red, dashed] (2) to [out=55 + 180,in=125 + 180, looseness=12] (2);
            \draw [edge, red, dashed] (3) to [out=55 + 180,in=125 + 180, looseness=12] (3);
            \draw [edge, blue] (4) to [out=55, in=125, looseness=12] (4);
               
            \foreach \from/\to in {4/0, 0/3, 1/2, 1/4, 2/4} 
                \draw [edge, blue] (\from) to  (\to);
            \foreach \from/\to in {0/1, 1/3, 2/3, 3/4, 0/2} 
                \draw [edge, red, dashed] (\from) to  (\to);
                
        \end{scope}

\end{tikzpicture}
\caption{Some minimal hereditarily-hard reflexive complete $2$-edge-coloured graphs.}
\label{fig:3-element}
\end{figure}
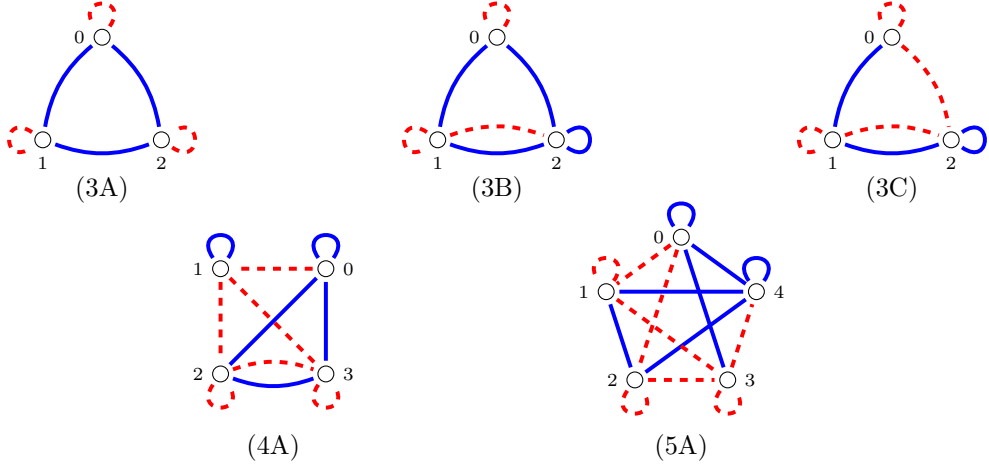

\begin{lemma}\label{lem:3-element-sig}
    If $\bH$ is a $2$-edge-coloured graph from Figure~\ref{fig:3-element}, then $\bH$
    and its dual $\overline\bH$ hereditarily pp-construct $\bK_3$.
\end{lemma}
\begin{proof}
    We prove the claim for $\bH$ in Figure~\ref{fig:3-element}, their duals
    follows with symmetric arguments.
    The case when $\bH$ is $\mathrm{3A}$ is covered by Observation~\ref{obs:red-odd-cycle}.
    For the remaining cases we show that $\Sig(\bH)$ contains an $\ast$-triangle, and
    the claim thus follows by Lemma~\ref{lem:siggers-power}.
    We now depict an $\ast$-triangle in $\Sig(\bH)$ for each $\bH\in \{\mathrm{3B,3C, 4A}\}$:
    a sequence $abcd$ represents the tuple $(a,b,c,d)$, edges represent edges in $\bH^4$,
    and $\sim_S$-equivalence classes are represented by vertex colour classes (black, white, and grey).
    
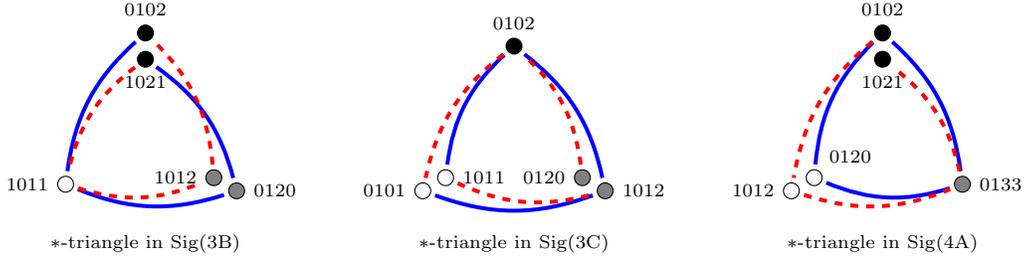
\begin{center}
    \begin{tikzpicture}

        \begin{scope}[scale=0.7]
            \node (L1) at (0,-2) {\scriptsize $\ast$-triangle in Sig(3B)};
            \node (0) [vertex, fill = black, label = below:{\scriptsize $1021$}] at (90:1.5){};
            \node (00) [vertex, fill = black,  label = above:{\scriptsize $0102$}] at (90:2){};
            \node (1) [vertex, label = left:{\scriptsize $1011$}] at (210:1.75){};
            \node (2) [vertex,  fill = gray, label = left:{\scriptsize $1012$}] at (330:1.5){};
            \node (22) [vertex,  fill = gray, label = right:{\scriptsize $0120$}] at (330:2){};
               
            \foreach \from/\to in {00/1,  22/0, 1/22} 
                \draw [edge, blue] (\from) to [bend right = 20] (\to);
            \foreach \from/\to in {0/1, 1/2, 2/00} 
                \draw [edge, red, dashed] (\from) to [bend right = 20] (\to);
        \end{scope}
    
        \begin{scope}[scale=0.7,xshift=7cm]
           \node (L1) at (0,-2) {\scriptsize $\ast$-triangle in Sig(3C)};
            \node (00) [vertex, fill = black,  label = above:{\scriptsize $0102$}] at (90:1.75){};
            \node (1) [vertex, label = right:{\scriptsize $1011$}] at (210:1.5){};
            \node (11) [vertex, label = left:{\scriptsize $0101$}] at (210:2){};
            \node (2) [vertex,  fill = gray, label = right:{\scriptsize $1012$}] at (330:2){};
            \node (22) [vertex,  fill = gray, label = left:{\scriptsize $0120$}] at (330:1.5){};
               
            \foreach \from/\to in {00/1, 11/2, 2/00} 
                \draw [edge, blue] (\from) to [bend right = 20] (\to);
            \foreach \from/\to in {00/11, 1/2, 22/00} 
                \draw [edge, red, dashed] (\from) to [bend right = 20] (\to);
        \end{scope}
    
        \begin{scope}[scale=0.7, xshift = 14cm]
            \node (L1) at (0,-2) {\scriptsize $\ast$-triangle in $\Sig(\mathrm{4A})$};
            \node (0) [vertex, fill = black, label = below:{\scriptsize $1021$}] at (90:1.5){};
            \node (00) [vertex, fill = black,  label = above:{\scriptsize $0102$}] at (90:2){};
            \node (1) [vertex, label = 45:{\scriptsize $0120$}] at (210:1.5){};
            \node (11) [vertex, label = left:{\scriptsize $1012$}] at (210:2){};
            \node (2) [vertex, fill = gray, label = right:{\scriptsize $0133$}] at (330:1.75){};
               
            \foreach \from/\to in {00/1, 1/2, 2/00} 
                \draw [edge, blue] (\from) to [bend right = 20] (\to);
            \foreach \from/\to in {00/11, 11/2, 2/0} 
                \draw [edge, red, dashed] (\from) to [bend right = 20] (\to);
        \end{scope}
    
\end{tikzpicture}
\end{center}
Following the same conventions as in the previous illustration, we depict
$\ast$-cycles of length five in  $\Sig(\mathrm{5A})$.
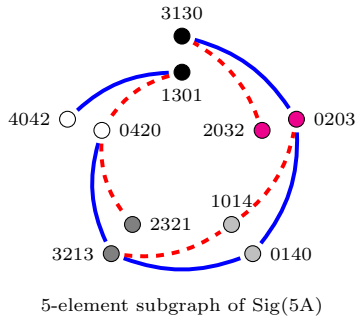
\begin{center}
    \begin{tikzpicture}
        \begin{scope}[scale=0.8]
             \node (L1) at (0,-2.5) {\scriptsize 5-element subgraph of $\Sig(\mathrm{5A})$};
            \node (0) [vertex, fill = black, label = below:{\scriptsize $1301$}] at (90:1.4){};
            \node (00) [vertex, fill = black,  label = above:{\scriptsize $3130$}] at (90:2){};
            \node (1) [vertex, label = right:{\scriptsize $0420$}] at (162:1.4){};
            \node (11) [vertex, label = left:{\scriptsize $4042$}] at (162:2){};
            \node (2) [vertex, fill = gray, label = right:{\scriptsize $2321$}] at (234:1.4){};
            \node (22) [vertex, fill = gray, label = left:{\scriptsize $3213$}] at (234:2){};
            \node (33) [vertex, fill = lightgray, label = right:{\scriptsize $0140$}] at (306:2){};
            \node (3) [vertex, fill = lightgray, label = above:{\scriptsize $1014$}] at (306:1.4){};
            \node (4) [vertex, fill = magenta, label = left:{\scriptsize $2032$}] at (18:1.4){};
            \node (44) [vertex, fill = magenta,  label = right:{\scriptsize $0203$}] at (18:2){};

            \foreach \from/\to in {0/11, 1/22, 22/33, 33/44, 44/00} 
                \draw [edge, blue] (\from) to [bend right = 20] (\to);
            \foreach \from/\to in {0/1, 1/2, 3/44, 22/3, 4/00} 
                \draw [edge, red, dashed] (\from) to [bend right = 20] (\to); 
        \end{scope}
\end{tikzpicture}
\end{center}
\end{proof}

A simple structural implication of this lemma is that whenever a $2$-edge-coloured
graph $\bH$ contains either of the graphs from Figure~\ref{fig:ast-path}
it pp-constructs $\bK_3$. Indeed, regardless of the edge type connecting
the end vertices in these paths, we can find either of the graphs 3B, 3C, or
one of their duals as subgraphs of $\bH$.

\begin{corollary}\label{cor:ast-edges}
    Let $\bH$ be a reflexive complete $2$-edge-coloured  graph without
    $\ast$-loops. If $\bH$ does not pp-construct $\bK_3$, then 
    $\bH$ does not contain either of the graphs in Figure~\ref{fig:ast-path}
    as subgraphs.
\end{corollary}

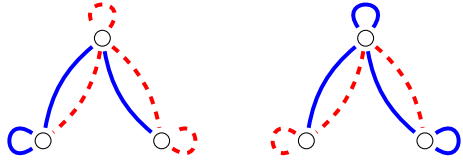
\begin{figure}[ht!]
\centering
    \begin{tikzpicture}
    
        \begin{scope}[scale=0.7]
            \node (0) [vertex] at (90:1.3){};
            \node (1) [vertex] at (210:1.3){};
            \node (2) [vertex] at (330:1.3){};
            
            \draw [edge, red, dashed] (0) to [out=55,in=125, looseness=12] (0);
             \draw [edge, blue] (1) to [out=55 + 90,in=125 + 90, looseness=12] (1);
            \draw [edge,  red, dashed] (2) to [out=55 -90,in=125 + -90, looseness=12] (2);
               
            \foreach \from/\to in {0/1, 0/2} 
                \draw [edge, blue] (\from) to [bend right = 20] (\to);
                \foreach \from/\to in {0/1, 0/2} 
                \draw [edge, red, dashed] (\from) to [bend left = 20] (\to);
            
        \end{scope}
    
        \begin{scope}[scale=0.7, xshift=5cm]
            \node (0) [vertex] at (90:1.3){};
            \node (1) [vertex] at (210:1.3){};
            \node (2) [vertex] at (330:1.3){};
            
            \draw [edge, blue] (0) to [out=55,in=125, looseness=12] (0);
            \draw [edge, red, dashed] (1) to [out=55 + 90,in=125 + 90, looseness=12] (1);
            \draw [edge, blue] (2) to [out=55 -90,in=125 + -90, looseness=12] (2);
               
            \foreach \from/\to in {0/1, 0/2} 
                \draw [edge, blue] (\from) to [bend right = 20] (\to);
                \foreach \from/\to in {0/1, 0/2} 
                \draw [edge, red, dashed] (\from) to [bend left = 20] (\to);
        \end{scope}

\end{tikzpicture}
\caption{Two paths on three vertices such that if a reflexive complete $2$-edge-coloured 
graph $\bH$ contains either of them as subgraphs, then $\bH$ pp-constructs $\bK_3$.}
\label{fig:ast-path}
\end{figure}

\subsection*{Restricted hereditarily pp-constructions}

Using our results from above, we show that the three element
$2$-edge-coloured graph from Figure~\ref{fig:H3} hereditarily
pp-constructs $\bK_3$ in the class of reflexive $2$-edge-coloured graphs.
That is, if $\bG$ is a reflexive $2$-edge-coloured 
graph without $\ast$-loops and $\bH_3\to \bG$, then $\bG$ pp-constructs
$\bK_3$.

The technique we use is similar as above. We proceed by contradiction
and notice that if $\bH_3\to \bG$ and $\bG$ has a Siggers polymorphism, 
then $\Sig(\bH_3)\to \bG$. We then use the fact that $\bG$ is a reflexive
graph to observe that this implies that there is graph $\bH$
(from Figure~\ref{fig:3-element}, or one of their duals)
that hereditarily pp-constructs $\bK_3$, and such that $\bH\to \bG$.
This will settle our claim.

\begin{figure}[ht!]
    \centering
    \begin{tikzpicture}
    \begin{scope}[scale=0.7, xshift=5cm]
            \node (L1) at (0,-1.6) {$\bH_3$};
            \node (0) [vertex, label = left:{\scriptsize $0$}] at (90:1.3){};
            \node (1) [vertex, label = below:{\scriptsize $1$}] at (210:1.3){};
            \node (2) [vertex, label = below:{\scriptsize $2$}] at (330:1.3){};
            
            \draw [edge, red, dashed] (0) to [out=55,in=125, looseness=12] (0);
            \draw [edge, red, dashed] (1) to [out=55 + 90,in=125 + 90, looseness=12] (1);
            \draw [edge, blue] (2) to [out=55 -90,in=125 + -90, looseness=12] (2);
               
            \foreach \from/\to in {0/1, 1/2} 
                \draw [edge, blue] (\from) to [bend right = 20] (\to);
            \foreach \from/\to in {0/2} 
                \draw [edge, red, dashed] (\from) to [bend left = 20] (\to);
                
        \end{scope}

\end{tikzpicture}
\caption{A $2$-edge-coloured graph $\bH_3$ on three vertices that is hereditarily
hard for the class of reflexive $2$-edge-coloured graphs.}
\label{fig:H3}
\end{figure}
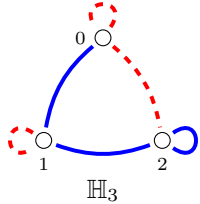

\begin{lemma}\label{lem:H3}
    Let $\bG$ be a reflexive $2$-edge-coloured  graph with no $\ast$-loops.
    If $\bH_3\to \bG$, then $\bG$  pp-constructs $\bK_3$.
\end{lemma}
\begin{proof}
    We show that if $\bG$ admits a Siggers polymorphism, then there
    is a structure $\bH$ from Figure~\ref{fig:3-element} 
    such that $\bH\to \bG$, and so  the claim follows from Lemma~\ref{lem:3-element-sig}.
    Since $\bH_3\to \bG$, if $\Sig(\bG)\to \bG$, then $\Sig(\bH_3)\to \bG$.
    We consider the following 4-element structure induced in $\Sig(\bH_3)$
    where edges represent edges in $(\bH_3)^4$,  colour classes represent
    $\sim_S$-equivalence classes, and to simplify notation we write $abcd$
    instead of $(a,b,c,d)$.
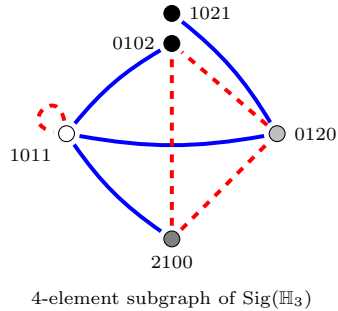
    \begin{center}
    \begin{tikzpicture}
        \begin{scope}[scale=0.8]
            \node (L1) at (0,-2.75) {\scriptsize 4-element subgraph of $\Sig(\bH_3)$};
            \node (0) [vertex, fill = black, label = right:{\scriptsize $1021$}] at (90:2){};
            \node (00) [vertex, fill = black,  label = left:{\scriptsize $0102$}] at (90:1.5){};
            \node (11) [vertex, label = below left:{\scriptsize $1011$}] at (180:1.75){};
            \node (22) [vertex, fill = gray, label = below:{\scriptsize $2100$}] at (270:1.75){};
            \node (33) [vertex, fill = lightgray, label = right:{\scriptsize $0120$}] at (0:1.75){};

            \draw [edge, red, dashed] (11) to [out=55 +45,in=125 +45, looseness=12] (11);
              
            \foreach \from/\to in {00/11, 33/0, 11/22, 11/33} 
                \draw [edge, blue] (\from) to [bend right = 10] (\to);
            \foreach \from/\to in {33/00, 22/00, 22/33} 
                \draw [edge, red, dashed] (\from) to (\to);       
        \end{scope}
\end{tikzpicture}
\end{center}
Let $f\colon\Sig(\bH_3)\to \bG$ be a homomorphism, and to simplify notation we write
$f(abcd)$ to denote the image of the equivalence class of $(a,b,c,d)$, in symbols
$f(abcd):=f([(a,b,c,d)]_{\sim_S})$. Since $\bG$ is reflexive, each of $f(1021)$, 
$f(2100)$ and $f(0120)$ are incident to some loop; we now proceed by a case distinction
depending on the colour of these loops. 
\begin{itemize}
    \item If $f(0102)$ and $f(0120)$ are incident to  red loops, then 3A$\to \bG$;
    \item If $f(0102)$ and $f(0120)$ are incident to loops of different colours, 
    then 3C$\to\bG$; 
    \item If $f(0102)$ and $f(0120)$ are incident to blue loops, and so is
    $f(2100)$ then $\overline{\mathrm{3A}}\to \bG$;
    \item If $f(0102)$ and $f(0120)$ are incident to blue loops, and 
    $f(2100)$ is incident to a red loop, then $\overline{\mathrm{4A}}\to \bG$.
\end{itemize}
It is clear that these cases cover all possible cases. We thus conclude
via Lemma~\ref{lem:3-element-sig}  that if $\bG$ does not contain an
$\ast$-loop, then $\bG$ pp-constructs $\bK_3$.
\end{proof}

\begin{corollary}\label{cor:unique-edge-type}
    Let $\bH$ be a reflexive complete $2$-edge-coloured  graph, and
    $\bC$ be a connected component of $(H_R, B)$ (resp.\ of $(H_B, R)$).
    If $\bH$ does not pp-construct $\bK_3$, then every blue vertex $v\in H_B$
    (resp.\ red vertex $b\in H_R$) is connected to  every $c\in C$ by the same
    edge type (either red or blue, but not both).
\end{corollary}
\begin{proof}
    It suffices to show that for each blue edge $uv$ of $\bH$ connecting red vertices,
    each blue vertex $b$ is connected to $u$ and to $v$ either by a blue or by a red edge,
    but not both.
    Firstly, if $b$ is connected to $u$ or to $v$ 
    by an $\ast$-edge,
    then there is a homomorphism from 3C or 3D 
    to $\bH$ showing that $\bH$ pp-constructs
    $\bK_3$ (Lemma~\ref{lem:3-element-sig}).
    If $b$ is connected to $u$ and to $v$
    by different type of edges, then we can find a homomorphism from $\bH_3$
    (Figure~\ref{fig:H3}) to $\bH$, and conclude that $\bH$ pp-constructs $\bK_3$
    (Lemma~\ref{lem:H3}).
\end{proof}

\section{Alternating components}
\label{sec:alt-components}

We begin this section by introducing the \emph{alternating digraph}
$\Alt(\bH)$ of a reflexive complete $2$-edge-coloured graph $\bH$, and we use
this digraph to introduce the alternating components of $\bH$. We then
use the results from Section~\ref{sec:HH}
to present a series of structural properties $\Pi$ such that
if an alternating component $\bA$ of $\bH$ does not satisfy $\Pi$,
then $\bH$ pp-constructs $\bK_3$. 

The vertex set of $\Alt(\bH)$ is $H$, and there is an arc $(x,y)$ if
$xx\in B(\bH)$ and $xy\in R(\bH)$ or $xx\in R(\bH)$ and $xy\in B(\bH)$.
Equivalently, the arc set of $\Alt(\bH)$ is defined by the
quantifier-free positive formula:
\begin{equation}
    \alpha(x,y):= (B(x,x)\land R(x,y))\lor (R(x,x)\land B(x,y)) .
\end{equation}
We further consider $\Alt(\bH)$ as a vertex coloured digraph where
a vertex $h$ is coloured with blue if $h\in H_B$, and with red
if $h\in H_R$, i.e., $h$ is coloured with blue (resp.\ red) if
it is incident to a blue (resp.\ red) loop in $\bH$.
In particular, we say that an arc $(u,v)$ in $\Alt(\bH)$ is
monochromatic if $u$ and $v$ have the same vertex colour.

\begin{remark}\label{rmk:trivial}
    If $\bH$ is a reflexive complete $2$-edge-coloured  graph, then
    for a blue vertex $b$ and a red vertex $r$ at least
    one of $(b,r)$ or $(r,b)$ is an arc in $\Alt(\bH)$, and
    both are arcs if and only if $rb$ is an $\ast$-edge in
    $\bH$. Also, for every substructure $\bH'$ of $\bH$,
    the digraph $\Alt(\bH')$ is isomorphic to the subdigraph
    of $\Alt(\bH)$ induced by the vertex set $H'$.
\end{remark}

The \emph{alternating component} of a vertex $h$ in $\bH$ is
the strongly connected component of $h$ in $\Alt(\bH)$, i.e.,
the maximal subset $H'\subseteq H$ containing $h$ and such
that for every pair of vertices $u,v\in \bH$ there is
a directed path from $u$ to $v$ and from $v$ to $u$
in $\Alt(\bH)$. In particular, every $\ast$-edge is contained
in some alternating component of $\bH$.

\begin{observation}\label{obs:endomorphism}
    Let $\bH$ be a reflexive complete $2$-edge-coloured  graph, 
    and let $\bA$ be an alternating component of $\bH$. 
    If $f\colon \bA\to \bA$ is an endomorphism of $\bA$,
    then the following mapping is an endomorphism of $\bH$
    \[
        g(h) =  \begin{cases}
                    f(h) & \text{if } h\in A,\\
                    h & \text{otherwise.}
                \end{cases}
    \]
\end{observation}

\begin{observation}\label{obs:automorphism}
    Let $\bH$ be a reflexive complete $2$-edge-coloured  graph, 
    and let $\bA$ be an alternating component of $\bH$. 
    If $f\colon \bH\to \bH$ is an automorphism of $\bH$,
    then $f[\bA]$ is an alternating component of $\bH$.
\end{observation}

Recall that a \emph{topological ordering} 
$\bD_1 \le \dots \le \bD_k$ of the strongly
connected components of a digraph $\bD$
is a linear ordering such that whenever 
there is an arc from $u\in \bD_i$ to $v\in \bD_j$
it is the case that $i\le  j$.  When we talk
about a topological ordering of the alternating components
of a $2$-edge-coloured graph $\bH$, we refer to a
topological ordering of the strongly connected components
of $\Alt(\bH)$. Using the definition of the arc set of
$\Alt(\bH)$, it readily follows that if 
$\bA_1\leq\dots\leq \bA_k$ is a topological ordering of the
alternating components of $\bH$, then $A_k$ is a homogeneous
set in $\bH$. By applying this argument inductively, we
obtain the following  observation that connects the notions of
homogeneous concatenation from Section~\ref{sec:hom-con}
and alternating components.

\begin{observation}\label{obs:topological-ordering}
    Let $\bH$ be a reflexive complete $2$-edge-coloured  graph.
    If $\bA_1 \le \cdots \le \bA_k$
    is a topological ordering of the alternating components of $\bH$, then
    $\bH = \bA_1 \triangleleft_h \dots \triangleleft_h \bA_k$.
\end{observation}

In the rest of this section we study the structure of the
alternating components of $2$-edge-coloured graphs that 
do not pp-construct $\bK_3$. We begin with the following
claim, where $A_B$ and $A_R$ denote the set of blue and of red
vertices in some vertex set $A$, respectively.

\begin{lemma}\label{lem:mono-vs-bipartite}
    Let $\bH$ be a reflexive complete $2$-edge-coloured  graph,
    and let $\bA$ be an alternating component of $\bH$. If $\bH$
    does not pp-construct $\bK_3$, then one of the following
    holds
    \begin{itemize}
        \item $A = A_R$ and $(A, B)$ is a bipartite connected component of $(H_R, B)$,
        \item $A = A_B$ and $(A, R)$ is a bipartite connected component of $(H_B, R)$, or
        \item $A_R$ induces a reflexive red clique with no blue edges,
        and $A_B$ induces a reflexive blue clique with no red edges.
    \end{itemize}
\end{lemma}
\begin{proof}
    Assume that $A = A_R$, and notice that, in this case,
    there is an arc $(a,b)$ for $a,b\in A$ if and only if there is a
    blue edge connecting $a$ and $b$ in $\bH$.  Hence $(A, B)$
    is a connected component of $(H_R, B)$. 
    It follows from
    Observation~\ref{obs:red-odd-cycle} that it is bipartite
    because $\bH$ does not pp-construct $\bK_3$. The case $A = A_B$
    follows with symmetric arguments. 

    To show that exactly one of the three itemized cases must hold, 
    we argue that if the third one does not, then one of the
    first two does.
    Up to colour symmetry, if suffices to show 
    that if $\bA$ contains a pair of vertices $u$ and $v$ such that $uv\in B(\bA)$,
    and $uu,vv\in R(\bA)$, then $A = A_R$. Let $\bC$ be the connected
    component of $u$ and $v$ in $(H_R, B)$.
    Clearly, $C\subseteq A$, and using Corollary~\ref{cor:unique-edge-type}
    and the definition of $\Alt(\bH)$ we see that 
    for every blue vertex $b$\blue{,} either $(b,c)$ is an arc for all $c\in C$,
    or $(c,b)$ is an arc for all $c\in C$, but not both. It also follows from
    the definition of $C$ that if $r$ is a red vertex not in $C$, then
    there is no arc in $\Alt(\bH)$ connecting $r$ to some $c\in C$.
    In particular, if $u$ is reached by some red vertex $r\in H_R\setminus C$,
    then it   must be a directed path of length at least $2$. 
    We now show that if $r$ reaches $u$, then $r$ reaches $u$ by a directed path
    of length (exactly) two.
    We also show that
    every out-neighbour 
    of $u$ in $H_B$ is an out-neighbour of $r$,
    and that every in-neighbour 
    of $r$ in $H_B$ is an in-neighbour of $u$.
    
    To do so, we first show that every out-neighbour $b'\in H_B$ of $u$ 
    is an out-neighbour of $r$ when there is a directed path of length
    $2$ from $r$ to $u$. This yields the following structure 
    in $\Alt(\bH)$ (to the left), and in $\bH$ (to the right),
    where $u$, $v$ and $r$ have already been introduced, 
    and $b$ is a vertex witnessing that there is a directed path of
    length two from $r$ to $u$ (and to $v$) in $\Alt(\bH)$.
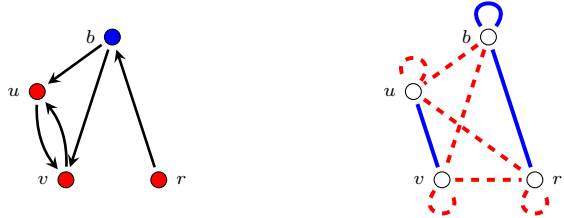
       \begin{center}
    \begin{tikzpicture}
    \begin{scope}[scale=0.7]
            \node (0) [vertex, fill = blue, label = left:{\scriptsize $b$}] at (90:1.5){};
            \node (1) [vertex, fill = red,  label = left:{\scriptsize $u$}] at (162:1.5){};
            \node (2) [vertex, fill = red, label = left:{\scriptsize $v$}] at (234:1.5){};
            \node (3) [vertex, fill = red, label = right:{\scriptsize $r$}] at (306:1.5){};
            
            \foreach \from/\to in {3/0, 0/1, 0/2} 
                \draw [arc] (\from) to  (\to);
            \foreach \from/\to in {1/2, 2/1} 
                \draw [arc] (\from) to [bend right = 20]  (\to);
        \end{scope}
        
    \begin{scope}[xshift = 5cm, scale=0.7]
            \node (0) [vertex, label = left:{\scriptsize $b$}] at (90:1.5){};
            \node (1) [vertex, label = left:{\scriptsize $u$}] at (162:1.5){};
            \node (2) [vertex, label = left:{\scriptsize $v$}] at (234:1.5){};
            \node (3) [vertex, label = right:{\scriptsize $r$}] at (306:1.5){};
            
            \draw [edge, blue] (0) to [out=55,in=125, looseness=12] (0);
            \draw [edge, red, dashed] (1) to [out=55, in=125, looseness=12] (1);
            \draw [edge, red, dashed] (2) to [out=55 + 180,in=125 + 180, looseness=12] (2);
            \draw [edge, red, dashed] (3) to [out=55 + 180,in=125 + 180, looseness=12] (3);
            
            \foreach \from/\to in { 0/3, 1/2} 
                \draw [edge, blue] (\from) to  (\to);
            \foreach \from/\to in {0/1, 1/3, 2/3, 0/2} 
                \draw [edge, red, dashed] (\from) to  (\to);
        \end{scope}
\end{tikzpicture}
\end{center}

Assume that there is a vertex $b'\in H_B$ such that
$b'$ is connected to 
$u$ by a blue edge (and so, also to $v$), and $b'$ is connected to
$r$ by a red edge, and consider the following to possible cases:
either $bb'$ is a red edge, or $bb'$ is a blue edge. In the former case,
there is a homomorphism from $\overline \bH_3$
(see Figure~\ref{fig:H3}) to the substructure induced by $b',b$ and $r$. In the
latter one, there is a homomorphism
from 5A (see Figure~\ref{fig:3-element}) to the substructure
induced by $\{u,v,r,b,b'\}$. By Lemmas~\ref{lem:3-element-sig} and~\ref{lem:H3},
any of these cases implies that $\bH$ pp-constructs $\bK_3$, contradicting our
assumption.  Therefore, every blue vertex $b'$ connected to $u$ by a blue edge
is connected to $r$ by a blue edge as well, and if $b'$ is connected to
$r$ by a red edge, then it is connected to $u$ by a red edge as well.
Translating back to $\Alt(\bH)$, this means that 
blue out-neighbour of $u$ in $\Alt(\bH)$ is an out-neighbour of $r$,
and every in-neighbour $b'\in H_B$ of $r$ is an in-neighbour of $u$.
By the latter fact, it follows that if there is a directed path from some
red vertex $r'$ to $u$ (equivalently, to some $c\in C$), then there is
a directed path of length two. Indeed, let $r'= r_1,b_1,\dots, r_{k-1},b_{k-1},r_k = u$
be a shortest path from $r'$ to $u$ --- we assume that this path alternates
between blue and red vertices, because otherwise, we can 
find a homomorphism from $\bH_3$ or from $\overline{\bH}_3$ to $\bH$
contradicting Lemma~\ref{lem:H3}
and the assumption that $\bH_3$ does not
pp-construct $\bK_3$. By the previous arguments, 
it follows that every in-neighbour of $r_{k-1}$ is an in-neighbour of $u$,
hence, if $k > 2$, we find a shorter path $r',b_1,\dots, r_{k-2},b_{k-2},u$,
contradicting the choice of path from $r'$ to $u$.

Finally, to show that the alternating component $\bA$ of $u$
equals $\bC$ we distinguish between the following cases, 
and each of these will lead to a contradiction.
\begin{itemize}
    \item Assume there is another red vertex $r\in A\setminus C$.
    We can choose $r\in A\setminus C$ so that there is
    a directed $ub$-path  $ubr$ of length $2$, where $b\in H_B$.
    By the arguments above, since $r$ also reaches $u$, every blue
    in-neighbour of $r$ is an in-neighbour of $u$.  Hence,
    $(u,b)$ and $(b,u)$ are arcs of $\Alt(\bH)$, contradicting
    Remark~\ref{rmk:trivial} because $(u,v)$ is a monochromatic arc of $\Alt(\bH)$.
    \item $A_R = C$, in this case, by Remark~\ref{rmk:trivial}
    all blue vertices are connected to each $c\in C$ by an out-going
    arc or an in-going arc but not both. Hence, there must be an
    arc in $\Alt(\bH)$ connecting the blue out-neighbour of $C$ with the
    blue in-neighbours of $C$. This yields a monochromatic arc connecting blue
    vertices that contradicts Remark~\ref{rmk:trivial}.
\end{itemize}
\end{proof}

\begin{proposition}\label{prop:comp-with-ast-edges}
    Let $\bH$ be a reflexive complete $2$-edge-coloured  graph that is a core,
    and $\bA$ an alternating component of $\bH$. If $\bH$ does not pp-construct
    $\bK_3$, and $\bA$ contains an $\ast$-edge, then $|A| = 2$. 
\end{proposition}
\begin{proof}
    Use Lemma~\ref{lem:mono-vs-bipartite} and Observation~\ref{obs:endomorphism}.
\end{proof}

\subsection{Monochromatic components}

We say that an alternating component $\bA$ is monochromatic,
if all vertices in $\bA$ are incident to loops of the same colour.
In this subsection we show that if $\bA$ is a monochromatic
alternating component of $\bH$, and $\bH$ is a core that
does not pp-construct $\bK_3$, then $\bA$ has at most 
two vertices.

\begin{lemma}\label{lem:mono-components}
    Let $\bH$ be a reflexive complete $2$-edge-coloured  graph
    that is a core, and let $\bA$ be a monochromatic alternating component
    of $\bH$. If $\bH$ does not pp-construct $\bK_3$,
    then $(\bH,h_1,\dots, h_n)$ pp-defines the set $A$.
\end{lemma}
\begin{proof}
    Proving the claim for blue components and for red components can be done with
    symmetric arguments. We prove the claim for a blue component $\bA$, 
    and notice that if $A = \{h_i\}$ for some $i\in[n]$, then $x = h_i$ pp-defines $A$.
    Also, if $\bA$ is the unique blue alternating component, then $B(x,x)$ pp-defines $A$. 
    So we assume that $|A|\ge 2$, and that $\bA$ is not the unique blue 
    alternating component of $\bH$.
    
    Without loss of generality we assume that the vertices $h_1,\dots, h_k$ are
    blue vertices, and the vertices $h_{k+1},\dots, h_n$ are red vertices. By 
    Corollary~\ref{cor:unique-edge-type}, we assume that for each red vertex
    $h_j$ there is a unique edge type connecting $h_j$ to each vertex in $\bA$.  
    We denote by $E_j\in\{R,B\}$ the unique edge type connecting $h_j$
    to the vertices
    in (the already fixed component) $\bA$.
    Since $|A|>2$ every vertex in $A$ is incident to a red edge. We claim that $A$ is
    pp-definable via the formula
    \[
        \delta(x):= B(x,x) \land (\exists y.\; R(x,y)\land B(y,y)) \land \bigwedge_{i\le k} B(x,b_i) \land \bigwedge_{k\le j\le n} E_j(x,h_j).
    \]
    It follows from all the current assumptions that $\bH\models \delta(a)$
    for every $c\in A$. We now see that if $\bH\models \delta(h)$, then
    $h\in A$. It is immediate to see that $h$ belongs to a blue component
    with at least two vertices (the first two conjuncts guarantee this fact).
    Also notice that if $\bC$ is a blue component with at least two vertices
    and no $\ast$-edge, then for every $u\in C$ there is a some $i\le k$
    such that $h_i\in C$ such that $uh_i\in R$ and $uh_i\not\in B$
    (simply choose $h_i$ to be  a vertex in $\bC$ connected to $u$ by a red edge).
    Hence, $h$ belongs to $\bA$ or to a blue component $\bC$ on at least two vertices
    that contains an $\ast$-edge. To conclude the proof we show that
    there is some red vertex $h_j$ such that the colour of the edges
    connecting $h_j$ to $\bC$ is different from the colour of the
    edges connecting $h_j$ to $\bA$. Firstly, it follows from
    Corollary~\ref{cor:unique-edge-type} (second item) that $\bC$ is an $\ast$-edge whose
    vertices are incident to blue loops. Since $(A,R)$ is bipartite, there is a
    homomorphism $f\colon \bA\to \bC$, and since $\bH$ is a core, the extension $\bar f$
    of $f$ that acts as the identity in $H\setminus A$ is not an endomorphism.
    Since each vertex $c\in A\cup C$ is connected to all blue vertices in
    $H_B\setminus (A\cup C)$ only by blue edges, there must be some red
    vertex $r_j$ with $k\le j \le n$
    such that  $r_j$ is connected to vertices in $A$  with a different edge type
    than the edge type connecting $v$ to the vertices in $C$
    (otherwise, $\bar f\colon \bH\to \bH$ is a non-surjective homomorphism). 
    This shows that if $c\in C$, then $\bH\not\models E_j(c,h_j)$,
    and so $\bH\not\models \delta(c)$.
    Putting all together we conclude that if $\bH\models \delta(h)$, then
    $h\in A$, and thus $(\bH, h_1,\dots, h_n)$ pp-defines the alternating
    set $A$.
\end{proof}

\begin{corollary}\label{cor:mono-size-2}
    Let $\bH$ be a reflexive complete $2$-edge-coloured  graph that is a core,
    and $\bA$ a monochromatic alternating component of $\bH$. 
    If $\bH$ does not pp-construct $\bK_3$, then $|A|\le 2$.
\end{corollary}
\begin{proof}
    The claim follows from Corollary~\ref{cor:unique-edge-type}
    when $\bA$ contains an $\ast$-edge, so we assume that $\bA$
    does not contain an $\ast$-edge.
    Since $\bH$ is a core, it pp-constructs $(\bH, h_1,\dots, h_n)$.
    Hence, by Lemma~\ref{lem:mono-components}, $\bH$ pp-constructs
    $\bA$. Assume without loss of generality that all vertices in
    $\bA$ are incident to a red loop. Since $\bH$ does not pp-construct
    $\bK_3$, $\bA$ cannot pp-construct $\bK_3$ either, and it thus
    follows from Proposition~\ref{prop:full-hom} that $\bA$ admits 
    a full-homomorphism to $(\bK_1 + \bK_2)^\ast$. Since $\bA$
    is an monochromatic alternating component, $(A, R)$ is connected
    and so, $\bA$ must admit a full-homomorphism to $\bK_2^\ast$.
    Finally, since $\bH$ is a core, and every endomorphism of $\bA$
    extends to an endomorphism of $\bH$ (Corollary~\ref{cor:unique-edge-type}),
    we conclude that $\bA$ is isomorphic to either $\bK_1^\ast$ or
    to $\bK_2^\ast$. 
\end{proof}

\subsection{Bichromatic components}

The goal of this subsection is to show that if $\bH$ is a core
that does not pp-construct $\bK_3$, and $\bA$ is an alternating
bichromatic component of $\bH$, then $\bA$ is either an $\ast$-edge connecting
a blue and a red vertex,  or the structure depicted in
Figure~\ref{fig:bichromatic}.

\begin{figure}[ht!]
    \centering
    \begin{tikzpicture}

        \begin{scope}[scale=0.7, xshift = 8cm]
            \node (L1) at (0,-1.4) {(4Alt)};
            \node (0) [vertex, label = left:{\scriptsize $0$}] at (-1,2){};
            \node (1) [vertex, label = right:{\scriptsize $1$}] at (1,2){};
            \node (2) [vertex, label = left:{\scriptsize $2$}] at (-1,0){};
            \node (3) [vertex, label = right:{\scriptsize $3$}] at (1,0){};
            
            \draw [edge, blue] (0) to [out=55,in=125, looseness=12] (0);
             \draw [edge, blue] (1) to [out=55,in=125, looseness=12] (1);
            \draw [edge,  red, dashed] (2) to [out=55 - 180,in=125 + -180, looseness=12] (2);
            \draw [edge,  red, dashed] (3) to [out=55 - 180,in=125 + -180, looseness=12] (3);

            \foreach \from/\to in {3/2, 0/3, 2/1} 
                \draw [edge, red, dashed] (\from) to  (\to);
            \foreach \from/\to in {0/1, 0/2, 1/3} 
                \draw [edge, blue] (\from) to  (\to);
        \end{scope}

\end{tikzpicture}
\caption{A type of bichromatic alternating component.}
\label{fig:bichromatic}
\end{figure}
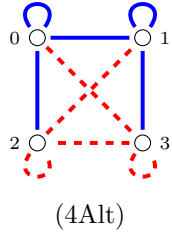

\begin{observation}\label{obs:4Alt}
    Let $\bH$ be a reflexive complete $2$-edge-coloured  graph that is
    a core, and $\bA$ an alternating bichromatic component of $\bH$
    with no $\ast$-edges.
    If 
    $|A_B| = 2$ or $|A_R| = 2$, then 
    $\bA$ is isomorphic to $\text{4Alt}$ (Figure~\ref{fig:bichromatic})
    or $\bH$ pp-constructs $\bK_3$.
\end{observation}
\begin{proof}
    Suppose that $\bH$ does not pp-construct $\bK_3$. Since
    $\bA$ is a bichromatic component, there is no red edge
    connecting vertices in $A_B$, and no blue edge
    connecting vertices in $A_R$ (Lemma~\ref{lem:mono-vs-bipartite}).
    Hence, if $A_R = \{u,v\}$, then each vertex $a\in A_B$ is connected to
    one vertex in $A_R$ by a blue edge, and to one vertex in $A_R$ by a 
    red edge (otherwise, $a$ would be a sink or source in $\Alt(\bA)$
    contradicting the assumption that $\Alt(\bA)$ is strongly connected. 
    Since $\bA$ has no $\ast$-edges, we can partition $A_B$ into 
    $A_B^1:=\{a\in A_B\colon au\in B(\bH), av\in R(\bH)$, and
    $A_B^2:=\{a\in A_B\colon au\in R(\bH), av\in B(\bH)$. Hence, 
    contracting $A_B^1$ and $A_B^2$ to a single vertex each, yields an endomorphism
    of $\bA$. Since every endomorphism of $\bA$ extends to an endomorphism
    of $\bH$ (Observation~\ref{obs:endomorphism}), and $\bH$ is a core,
    it must be the case that $|A_B^1| = 1$, that $|A_R^1| = 1$, and so,
    $\bA$ is isomorphic to 4Alt from Figure~\ref{fig:bichromatic}.
\end{proof}

\begin{lemma}\label{lem:alt-subsets}
    Let $\bH$ be a reflexive complete $2$-edge-coloured  graph that is
    a core, and $\bA$ a bichromatic alternating component of $\bH$. 
    If $\bH$ does not pp-construct $\bK_3$, 
    then $(\bH,h_1,\dots, h_n)$ pp-defines a pair of sets $S_B$ and
    $S_R$ such that,
    \begin{itemize}
        \item $A_B\subseteq S_B$ and $A_R\subseteq S_B$, and
        \item no vertex $s\in S_B$ (resp.\ $s\in S_R$)
        is incident to an $\ast$-edge
        $sv$ where $v\in H_R$ (resp.\ $v\in H_B$).
    \end{itemize}
\end{lemma}
\begin{proof}
    We show that such a set $S_B$ exists, and the existence of $S_R$
    can be done with symmetric arguments. Let $r_1,\dots, r_k$
    be an enumeration of the vertices in $H_R\setminus A_R$. 
    We first notice that for each $i\in [k]$ there is a unique
    edge type $E_i$ (either blue or red, but not both) connecting $r_i$
    to any vertex in $A_B$. Assume for a contradiction that there are
    vertices $a,b\in A_B$ such that $r_ia\in B(\bH)$ and $r_ib\in R(\bH)$. 
    Hence, there are arcs $(b,r_i)$ and $(r_i,a)$ in $\Alt(\bH)$, and
    since $a$ and $b$ belongs to the same strongly connected component
    in $\Alt(\bH)$, $r_i$ also belongs to this component, contradicting
    the choice of $r_i\in H_R\setminus A_R$. We claim that 
    the subset $S_B$ pp-defined by the following formula satisfies
    the claims of this lemma
    \[
        \sigma(x):=B(x,x)\land \bigwedge_{i\in[k]} E_i(r_i,x).
    \]
    By the arguments above, it follows that every vertex $a\in A_B$
    satisfies $\sigma(x)$ in $\bH$. We now show that if $\bH\models\sigma(h)$
    for some $h\in H_B$, then there is no $h'\in H_R$ such that $hh'$ is
    an $\ast$-edge. Consider the mapping $f\colon A\to \{h,h'\}$
    defined by $f(a) = h$ if $a\in A_B$ and $f(a) = h'$ if $a\in A_R$.
    Since $\bH$ does not pp-construct $\bK_3$, it follows from
    Lemma~\ref{lem:mono-vs-bipartite} that $A_R$ and $A_B$ induce monochromatic
    reflexive cliques, and so $f\colon\bA\to \bH[\{h,h'\}]$ is a homomorphism. 
    Since $\bH$ is a core, Corollary~\ref{cor:ast-edges} guarantees
    that $\bH[\{h,h'\}]$ is an alternating component  of $\bH$. Now,
    using similar arguments as before, e.g., using the definitions of 
    alternating components and of $\sigma$,
    one can notice that $f$ extends to an endomorphism of $\bH$ by defining 
    $f'$ as the identity on $H\setminus A$ and $f' = f$ on $A$. This contradicts
    the assumption that $\bH$ is a core, and so the claim follows.
\end{proof}

Recall that a digraph $\bD$ is \emph{smooth} if every vertex
has positive in-degree and positive out-degree. Also, $\bD$ is 
a \emph{semicomplete digraph} if for every pair of vertices
$u,v\in D$ at least one arc $(u,v)$ or $(v,u)$ is present in $\bD$.

\begin{lemma}\label{lem:AB-3}
    Let $\bH$ be a reflexive complete $2$-edge-coloured  graph
    that is a core, and $\bA$ a bichromatic alternating component
    of a $\bH$. 
    If $\bH$ does not pp-construct $\bK_3$, and $|A_B|\ge 3$, then the digraph 
    \[
        \bD:= (A_B,\{(x,y)\colon \exists z\in H.\; R(z,z)\land B(x,z)\land R(y,z)\})
    \]
    is semicomplete smooth digraph that contains at least two directed
    cycles.
\end{lemma}
\begin{proof}
    It follows from Lemma~\ref{lem:mono-vs-bipartite} that $A_B$ induces
    a reflexive blue clique. From the definition of alternating component,
    every vertex $a\in A_B$ is connected to each vertex $b\in H_B\setminus A_B$
    only by a blue edge.  Since $\bH$ is core, mapping a vertex
    $a\in A_B$ to a vertex $b\in A_B\setminus\{a\}$ cannot be extended to an
    endomorphism of $\bH$. This implies that for every pair of different
    vertices $a,b\in A_B$ there is some vertex $c\in H_R$ such that
    $ac\in R(\bH)$ and $bc\in B(\bH)$. Hence, $\bD$ is a semicomplete
    digraph. The fact that $\bD$ is smooth follows from the assumption
    that $\bA$ is an alternating component (if $\bD$ has a source or
    a sink $v$, then $v$ is also a source or a sink in $\Alt(\bA)$, which
    implies that $\Alt(\bA)$ is not strongly connected). It is straightforward
    to observe that all smooth semicomplete digraphs on at least four
    vertices have at least two directed cycles. Consider the case where
    $|A_B| = 3$ and let $v_1,v_2,v_3$ be an enumeration of its vertex set.
    Since $\bD$ is smooth, it must contain a directed $3$-cycle, so 
    assume without loss of generality that $(v_1,v_2),(v_2,v_3),(v_3,v_1)\in E(\bD)$. 
    By the definition of $\bD$ we can find three vertices $u_1,u_2,u_3\in H_R$
    such that $v_1u_1,v_2u_2,v_3u_3\in B(\bH)$ and $u_1v_2,u_2v_3,v_3u_1\in E(\bH)$.
    Now, if $v_1u_2\in B(\bH)$, then $(v_1,v_3)\in E(\bD)$, and if 
    $v_1u_2\in R(\bH)$, then $(v_2,v_1)\in \bD$. In both cases we find
    a directed $2$-cycle, and a directed $3$-cycle in $\bD$, and so
    $\bD$ satisfies the claim of this lemma.
\end{proof}

\begin{proposition}\label{prop:alt-components}
    Let $\bH$ be a reflexive complete $2$-edge-coloured  graph that is
    a core, and let $\bA$ be an alternating bichromatic component of $\bH$.
    If $\bH$ does not pp-construct $\bK_3$, then $\bA$ is an $\ast$-edge
    connecting a vertex in $H_B$ with a vertex in $H_R$, or $\bA$
    is isomorphic to 4Alt (Figure~\ref{fig:bichromatic}). 
\end{proposition}
\begin{proof}
    By Observation~\ref{obs:4Alt} it suffices to show that if 
    $\bA$ is not an $\ast$-edge, then $|A_B| =2$ or $|A_R| = 2$.
    We show that if $|A_B|\ge 3$, then $\bH$ pp-constructs
    $\bK_3$, and since $\bH$ is a core, it is equivalent to show that
    $(\bH,h_1,\dots, h_n)$ pp-constructs $\bK_3$. Let $S_B$
    be a set pp-defined in $(\bH,h_1,\dots, h_n)$ that contains $A_B$
    and no vertex $v\in S_B$ belongs to an $\ast$-edge $vw$ where $w\in H_R$
    (Lemma~\ref{lem:alt-subsets}).
    Consider now the binary relation $E$ pp-defined by
    \[
        E(x,y):= \exists z.\; R(z,z)\land B(x,z)\land R(y,z).
    \]
    We claim that $(S_B, E)$ is a loopless digraph containing 
    a smooth digraph that does not map homomorphically to a directed
    cycle (and so, by Theorem~\ref{thm:HH-digraphs} we conclude that
    $(\bH,h_1,\dots, h_n)$ pp-constructs $\bK_3$). Firstly, the
    fact that $(S_B, E)$ is a loopless digraph follows from the 
    assumption that no vertex $v\in S_B$ belongs to an $\ast$-edge $vw$
    where $w\in H_R$. To find the smooth subdigraph that does not
    map homomorphically to a directed cycle we first use the fact that
    $A_B\subseteq S_B$, and then conclude via Lemma~\ref{lem:AB-3}.
\end{proof}

\subsection{Homogeneous concatenations}

Theorem~\ref{thm:homogeneous-Ptime} guarantees that if $\bH$
admits a specific kind of decomposition in terms of certain homogeneous
concatenations, then the list version of $\CSP(\bH)$ can be solved
in polynomial time.  Building on previous results from this section,
we obtain the following decomposition result for $2$-edge-coloured graphs
that do not pp-construct $\bK_3$.

\begin{proposition}\label{prop:intermediate-class}
    Let $\bH$ be a reflexive complete $2$-edge-coloured  graph. 
    If $\bH$ does not pp-construct $\bK_3$, then $\bH$ admits
    a decomposition
    \[
        \bH:= \bA_1 \triangleleft_h\dots \triangleleft_h \bA_k
    \]
    where each $\bA_i$ is an alternating component of $\bH$, and
    $\bA_i$ or its dual $\overline{\bA_i}$ is a structure
    from Figure~\ref{fig:hom-blocks}, for every $i\in[k]$.
\end{proposition}
\begin{proof}
    Combine Observation~\ref{obs:topological-ordering},  together
    with Corollary~\ref{cor:mono-size-2} and Proposition~\ref{prop:alt-components}.
\end{proof}

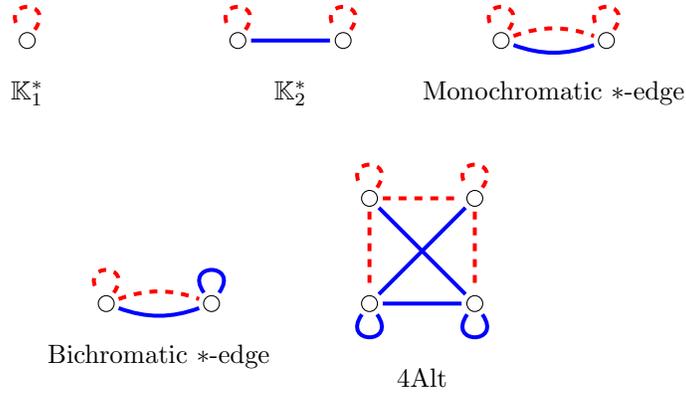
\begin{figure}[ht!]
    \centering
    \begin{tikzpicture}
    
        \begin{scope}[scale=0.7]
            \node (L1) at (0,-1) {$\bK_1^\ast$};
            \node (0) [vertex] at (0,0){};

            \draw [edge,  red, dashed] (0) to [out=55,in=125, looseness=12] (0);
        \end{scope}

        \begin{scope}[scale=0.7, xshift = 5cm]
            \node (L1) at (0,-1) {$\bK_2^\ast$};
            \node (0) [vertex] at (-1,0){};
            \node (1) [vertex] at (1,0){};
            
            \draw [edge, red, dashed] (0) to [out=55,in=125, looseness=12] (0);
             \draw [edge, red, dashed] (1) to [out=55,in=125, looseness=12] (1);
           
            \draw [edge, blue] (0) to  (1);
        \end{scope}

        \begin{scope}[scale=0.7, xshift = 10cm]
            \node (L1) at (0,-1) {Monochromatic $\ast$-edge};
            \node (0) [vertex] at (-1,0){};
            \node (1) [vertex] at (1,0){};
            
            \draw [edge, red, dashed] (0) to [out=55,in=125, looseness=12] (0);
             \draw [edge, red, dashed] (1) to [out=55,in=125, looseness=12] (1);
           
            \draw [edge, blue] (0) to [bend right = 20] (1);
            \draw [edge, red, dashed] (0) to [bend left = 20] (1);
        \end{scope}

        \begin{scope}[scale=0.7, yshift = -5cm, xshift = 2.5cm]
            \node (L1) at (0,-1) {Bichromatic $\ast$-edge};
            \node (0) [vertex] at (-1,0){};
            \node (1) [vertex] at (1,0){};
            
            \draw [edge, red, dashed] (0) to [out=55,in=125, looseness=12] (0);
             \draw [edge, blue] (1) to [out=55,in=125, looseness=12] (1);
           
            \draw [edge, blue] (0) to [bend right = 20] (1);
            \draw [edge, red, dashed] (0) to [bend left = 20] (1);
        \end{scope}
        
        \begin{scope}[scale=0.7, yshift = -5cm, xshift = 7.5cm]
            \node (L1) at (0,-1.4) {4Alt};
            \node (0) [vertex] at (-1,2){};
            \node (1) [vertex] at (1,2){};
            \node (2) [vertex] at (-1,0){};
            \node (3) [vertex] at (1,0){};
            
            \draw [edge, red, dashed] (0) to [out=55,in=125, looseness=12] (0);
             \draw [edge, red, dashed] (1) to [out=55,in=125, looseness=12] (1);
            \draw [edge,  blue] (2) to [out=55 - 180,in=125 + -180, looseness=12] (2);
            \draw [edge,  blue] (3) to [out=55 - 180,in=125 + -180, looseness=12] (3);
           
            \foreach \from/\to in {3/2, 0/3, 2/1} 
                \draw [edge, blue] (\from) to  (\to);
            \foreach \from/\to in {0/1, 0/2, 1/3} 
                \draw [edge, red, dashed] (\from) to  (\to);
        \end{scope}

\end{tikzpicture}
\caption{Admissible building components from Proposition~\ref{prop:intermediate-class}}
\label{fig:hom-blocks}
\end{figure}

\section{A structural classification}
\label{sec:structura-clas}

The aim of this section is to close the gap between
Theorem~\ref{thm:homogeneous-Ptime} and
Proposition~\ref{prop:intermediate-class}.
To do so, we will use the following well-known fact
(see, e.g.,~\cite[Proposition~1.2.11]{Book}).

\begin{proposition}\label{prop:core-orbits}
    Let $\bH$ be a $2$-edge-coloured graph that is a core. 
    If $O\subseteq H$ is the orbit of some $h\in H$, 
    then $O$ is pp-definable in $\bH$.
\end{proposition}

\subsection{The 4Alt case}

\begin{lemma}\label{lem:4Alt-ontop}
    Let $\bH$ be a reflexive complete $2$-edge-coloured  graph
    that is a core. Suppose that $\bH$ does not pp-construct $\bK_3$, 
    and let  $\bH:=\bA_1\triangleleft_h\dots \triangleleft_h \bA_k$
    as in Proposition~\ref{prop:intermediate-class}.
    If $\bA_j$ is isomorphic to \text{4Alt} for some $j\in[k]$,
    then $\bH$ pp-defines the set $A_1\cup \dots \cup A_j$.
\end{lemma}
\begin{proof}
    First notice that both blue vertices $A_j\cap H_B$
    from $\bA_j$ belong to the same orbit $O_B$ in $\bH$,
    and both red vertices $A_j\cap H_R$ belong to the same
    orbit $O_R$. We first show that if $O_R = A_j\cap H_R$
    and  $O_B = A_j\cap H_B$, then $\bH$ pp-defines
    the set $A_1\cup \dots \cup A_j$. Using the fact that $O_B$
    and $O_R$ are pp-definable (Proposition~\ref{prop:core-orbits}),
    consider the following pp-definition,
    \[
        \delta(x):= \exists y,z.\; (y\in O_B \land B(x,y))\land (z\in O_R\land R(x,z)).
    \]
    Clearly, every vertex $a\in A_j$ satisfies $\delta(x)$ in $\bH$. 
    Since $A_j$ is a homogeneous set in $\bA_1\triangleleft_h\dots \triangleleft_h \bA_k$,
    every vertex $a\in A_1\cup \dots A_{j-1}$ satisfies $\delta(x)$ in $\bH$ as well.
    Now, suppose that $a\in A_\ell$ for some $\ell> j$. It follows
    again by homogeneity of $A_\ell$ in $\bA_1\triangleleft_h\dots \triangleleft_h \bA_\ell$
    that, if $aa\in R(\bH)$, then $ya\in R(\bH)$ for every $y\in O_B$, and
    if $aa\in B(\bH)$, then $za\in B(\bH)$ for every $z\in O_R$. Hence, 
    if $O_R = A_j\cap H_R$ and  $O_B = A_j\cap H_B$, then $\bH$ pp-defines
    the set $A_1\cup \dots \cup A_j$.

    To conclude the proof we show that $O_R = H_R\cap A_j$ (a symmetric
    proof shows that $O_B = H_B\cap A_j$). Proceeding by contradiction, 
    suppose that there is some $r\in O_R\setminus (H_R\cap A_j)$. 
    Using Observation~\ref{obs:automorphism}, we see that there
    is some $i\neq j$ such that $\bA_j\cong \bA_i\cong \text{4Alt}$
    and $A_i\cup A_j\subseteq O_R$. 
    Consider the digraph $\bD:=\bigl(O_R, \{(x,y)\colon xz\in B(\bH), yz\in R(\bH)$
    for some $z\in H\}\bigr)$. Since $O_R$ is pp-definable in $\bH$, 
    $\bD$ is pp-constructible in $\bH$. We arrive to a contradiction by
    showing that $\bD$ pp-constructs $\bK_3$. To do so, first notice
    that $\bD$ is a loopless digraph: this follows from the definition
    of alternating component, so no $a\in A_j$ is incident to 
    an $\ast$-edge, and thus, no $h\in O_R$ is incident to an $\ast$-edge.
    We now observe that $(A_i\cup A_j)\cap O_R$ induces a complete (symmetric) graph
    in $\bD$. Let $A_i\cap O_R = \{a_i,b_i\}$ and  $A_j\cap O_R = \{a_j,b_j\}$.
    Firstly,  one can easily notice that $(a_i,b_i)$ and $(b_i,a_i)$ are
    arcs of $\bD$: simply use the vertices in
    $O_B\cap A_i$ as witnesses for $z$. Analogously, we see that  
    $(a_j,b_j)$ and $(b_j,a_j)$ are arcs of $\bD$. Secondly, 
    since $a_ib_i\in B(\bH)$, and $a_jb_i\in R(\bH)$ (because $a_j$
    and $b_i$ belong to different alternating components), we
    see that $(a_i,b_j)$ is an arc of $\bD$. With similar arguments,
    we see that $(a,b)$ and $(b,a)$ are arcs of $\bD$ for every
    $a\in A_i\cap O_R$ and $b\in A_j\cap O_R$. Hence, $\bD$
    is a loopless digraph that contains $\bK_4$, and so,
    $\bD$ pp-constructs $\bK_3$ (see, e.g., Theorem~\ref{thm:HH-digraphs}),
    and $\bH$ pp-constructs $\bK_3$ as well.
\end{proof}

In the rest of this section, we will work with the
cyclic power $\Cyc_p(\bH)$ of a $2$-edge-coloured graph $\bH$ (introduced
in the paragraph before Lemma~\ref{lem:cyclic-power}). The vertices
of $\Cyc_p(\bH)$ are equivalence classes of the equivalence
relation $\sim_p$ in $H^p$. To simplify notation, we will refer
to an equivalence class $X$ of $\sim_p$ by a word $x_1\dots x_p$
where $(x_1,\dots, x_p)\in H^p$ is any representative of $X$. 
Also, for a positive integer $n$ we write $x^n$ to denote
the word consisting of $n$ repetitions of $x$.

\begin{proposition}\label{prop:2-element-constructions}
    Let $\bH$ be a reflexive complete $2$-edge-coloured  graph. 
    If $\bH$ does not pp-construct $\bK_3$, then $\bH$ admits
    a decomposition
    \[
        \bH:= \bA_1 \triangleleft_h\dots \triangleleft_h \bA_k
    \]
    where each $\bA_i$ an alternating component of $\bH$,
    and $|A_i|\leq 2$ 
    for every $i\in[k]$.
\end{proposition}
\begin{proof}
    By Proposition~\ref{prop:intermediate-class}, $\bH$
    admits a decomposition $\bH:= \bA_1 \triangleleft_h\dots \triangleleft_h \bA_k$
    where each $\bA_i$ is an alternating component of $\bH$, and
    $\bA_i$ or its dual $\overline{\bA_i}$ is a structure
    from Figure~\ref{fig:hom-blocks}, for every $i\in[k]$.
    Thus, it suffices to prove that if some $\bA_i$ is isomorphic
    to 4Alt, then $\bH$ pp-constructs $\bK_3$, and by Lemma~\ref{lem:4Alt-ontop}
    we assume that $k$ is the first index  $i$ such that $\bA_i\cong$ 4Alt.

    Observe that if $\mathbb H$ contains no $\ast$-edges, then $\mathbb H$ pp-constructs
    $\mathbb K_3$:  consider the pp-definable binary relation
    $E(x,y):=\exists z.\; B(x,z)\land R(z,y)$; 
    with similar arguments as before,
    the reader can notice that $(H, E)$ is a loopless reflexive digraph, and
    the subdigraph induced by $A_k$ is a smooth digraph that does not map
    homomorphically to a directed cycle (conclude by Theorem~\ref{thm:HH-digraphs}).

    We now claim that there is some $\ell\in[k-1]$ such that
   $\bA_\ell$ is a monochromatic component of size 2, and that $\bA_j$
   is a monochromatic component of size 1 for each $\ell < j < k$. By the previous
   paragraph, we know that $\bA_i$ is an $\ast$-edge for some $i < k$, and so, if
   no $\bA_i$ is a bichromatic $\ast$-edge for some $i < k$, then our claim trivially 
   holds. So
    suppose  that there is some $i\in[k-1]$ such that 
    $\mathbb A_i$ is a bichromatic $\ast$-edge.
    Then, for some $\ell\in\{i+1,\ldots,k-1\}$, the component $\mathbb A_\ell$ is 
    monochromatic and has size 2, i.e., either $\mathbb K_2^*$,  $\overline{\bK_2^\ast}$
    or a monochromatic $\ast$-edge; otherwise all components $\mathbb A_{i+1},\ldots,\mathbb A_k$
    would collapse to $\mathbb A_i$ contradicting  that $\mathbb H$ is a core.
    So from now on we fix $\ell\in[k-1]$ such that $\bA_\ell$ is a monochromatic
    component of size $2$, and each $\bA_j$ consists of a single vertex 
    for $j\in\{\ell +1, \dots, k-1\}$.
 
    Denote the vertices of the copy of 4Alt with labels $0,1,2,3$ such that
    $00,11,01,02,13\in  B(\bH)$ and $22,33,23,12,03\in  R(\bH)$
    (this labelling corresponds to Figure~\ref{fig:bichromatic}).
    Without loss of generality, assume that $\mathbb A_\ell$  is
    a red component of size 2, i.e., its domain is $r_\ell, r_\ell'$  where 
    $r_\ell r_\ell, r_\ell' r_\ell'\in   R(\bH)$ and $r_\ell r_\ell'\in 
    B(\bH)$  (and possibly $r_\ell r'_\ell\in R(\bH)$).
    For each $j\in \{\ell+1,\ldots,k-1\}$, let $b_j$ (resp.\ $r_j$)
    denote the vertex of $\mathbb A_j$ if $b_jb_j\in B(\bH)$
    (resp.\ $r_jr_j\in R(\bH)$).

    Proceeding by contradiction, suppose that $\bH$ does not pp-construct
    $\bK_3$. Hence, by Lemma~\ref{lem:cyclic-power} there is a
    homomorphism $f\colon \Cyc_p(\bH)\to \bH$ for every prime 
    $p> |H|$. Let $p>|H|$ be a prime such that $p = 4m+1$
    for some positive integer $m$.
    Since $\mathbb H$ is a core, we assume that
    $f(x,\ldots,x)=x$ for each $x\in H$.

    Consider the elements
    $a:=a_0 = 0^{2m+1}r_\ell^{2m}\in \Cyc_p(\bH)$, $a_1 = 1^{2m+1}r_\ell'^{2m}$, $a_2 = 2^{2m+1}(r_\ell')^{2m}$, 
    and $a_3 = 3^{2m+1}r_\ell^{2m}$. Using the definition of $\sim_p$, and the structure
    decomposition of $\bH$, we see that $a_0a_0,a_1a_1\in B(\Cyc_p(\bH))$
    and $a_2a_2,a_3a_3\in R(\Cyc_p(\bH))$ (so $f(a_0), f(a_1)\in H_B$ and 
    $f(a_2),f(a_3)\in H_R$). Moreover, $a_0a_3$ is a red edge, 
    $a_3a_1$ is a blue edge, $a_1a_2$ is a red edge, and $a_2a_0$ is a blue edge. 
    Hence, $f$ maps $\{a_0,a_1,a_2,a_3\}$ to a bichromatic alternating component $\bA_i$
    of $\bH$.  Now notice that $a_0$ is connected to $(r_\ell')^p$ by
    a blue edge in $\Cyc_p(\bH)$, and since $f((r_\ell')^p) = r_\ell$, it must
    be the case that $\bA_i$ is an alternating component above $\bA_\ell$, i.e.,
    $i >\ell$. By the choice of $\ell$, every alternating component $\bA_i$
    with $\ell < i < k$ is a monochromatic component, and thus
    $i = k$. This implies that $f(a),f(a_1)\in\{0,1\}$ and $f(a_2),f(a_3)\in \{2,3\}$.

    Let $b = 0^{m+1} r_\ell^{m} 1^{2m}$ and $c = r_\ell^{m+1}1^{3m}$.
    Similarly as we showed that $f(a)\in\{0,1\}$, one can
    show that $f(b),f(c)\in \{0,1\}$, e.g., show that $f$ maps
    the elements $b_0 = b$, 
    $b_1 = 1^{m+1} r_\ell'^{m} 0^{2m}$,  $b_2  = 2^{m+1} r_\ell'^{m} 3^{2m}$,
     and $b_3 = 3^{m+1} r_\ell^{m} 2^{2m}$ to $\bA_k$.

    With similar arguments as we have used before, 
    one can notice that $\{0,1\}$ and that $\{2,3\}$ are pp-definable
    in $(\bH,h_1,\dots, h_n)$, and since $f$ is idempotent, it must be the
    case that $f(x_1\dots x_p )\in\{2,3\}$ 
    whenever $x_i\in\{2,3\}$, for all $i\in [p]$.
    Finally, consider the elements
    $u_0=3^p$, $u_1=3^{2m+1}2^{2m}$, $u_2=3^{m+1}2^{3m}$, and $u_3=2^p$. 
    Notice that $u_0a, au_1, u_1b,  bu_2, u_2c, cu_3\in R(\Cyc_p(\bH))$,
     and since $f(u_0) = f(3^p) = 3$, it must be the case that 
    $f(a)=f(b)=f(c)=1$ and $f(u_1)=f(u_2)=f(u_3)=f(u_0)=3$. 
    In particular, $f(2^p) = f(u_3) =3$ which contradicts the idempotency
    of $f$. Therefore, there is no idempotent cyclic polymorphism 
    $f\colon \bH^p\to \bH$, 
    and thus $\bH$ pp-constructs $\bK_3$.
\end{proof}

\subsection{The $\bK_2^\ast$ and $\overline{\bK_2^\ast}$ case}

We now show that if $\bH$ admits a decomposition 
$\bH:= \bA_1\triangleleft_h\dots \triangleleft_h\bA_k$,
and $\bA_i$ is isomorphic to $\bK_2^\ast$ or to
$\overline{\bK_2^\ast}$ for some $2\le i \le k$, then $\bH$
pp-constructs $\bK_3$.

\begin{lemma}\label{lem:consecutive-red}
    Let $\bH:= \bA_1 \triangleleft_h \dots \triangleleft_h \bA_k$
    be a decomposition of a reflexive complete $2$-edge-coloured 
    graph that is a core.
    If there is a pair $i,j\in[k]$ with $i < j$, such that
    $\bA_i$ and $\bA_j$ are isomorphic to $\bK_2^\ast$, 
    and each $\bA_\ell$ is a red alternating component for
    $i\le \ell \le j$, then $\bH$ pp-constructs $\bK_3$. 
\end{lemma}
\begin{proof}
    Let $O$ be the orbit of some $a\in A_i$, and clearly
    $A_i\subseteq O$. It is straightforward that there is an
   automorphism that maps $\bA_i$ 
   to $\bA_j$, $\bA_j$ to 
   $\bA_i$, and fixes the rest of $\bH$, and so, 
   $A_i\cup A_j\subseteq O$. Also,  every $h\in O$
   is incident to a red loop, and there is no $\ast$-edge
   induced by $O$ because $a$ is not incident to any $\ast$-edge
   (this follows from the choice of $\bA_i$, and from the definition
   of alternating component). Hence, $O$ induces a structure
   in the scope of Proposition~\ref{prop:full-hom}, and  since
   $A_i\cup A_j$ induces a copy of $(\bK_2+\bK_2)^\ast$, 
   we conclude that $\bH[O]$ does not admit a
   full-homomorphism to $(\bK_1 + \bK_2)^\ast$ (Observation~\ref{obs:K2+K1}).
   We conclude by Proposition~\ref{prop:full-hom} that 
   the substructure induced by $O$ pp-constructs $\bK_3$, and
   since $O$ is pp-definable in $\bH$, then $\bH$ pp-constructs
   $\bK_3$.  
\end{proof}

\begin{lemma}\label{lem:K2-on-top}
    Let $\bH$ be a reflexive complete $2$-edge-coloured  graph
    that is a core. If $\bH$ does not pp-construct
    $\bK_3$, $\bH$ admits a decomposition as in Theorem~\ref{thm:homogeneous-Ptime},
    or a substructure $\bH'$ of $\bH$ admits a decomposition
    $\bH':= \bA_1 \triangleleft_h \dots \triangleleft_h \bA_\ell$ such that
    \begin{itemize}
        \item $(\bH,h_1,\dots, h_n)$ pp-defines the union $A_1\cup\dots \cup A_\ell$,
        \item $A_\ell$ is the first index $i\ge 2$ such that $\bA_i$
        is isomorphic to $\bK_2^\ast$ or to $\overline{\bK_2^\ast}$, 
        \item and $A_{\ell-1}$ contains a vertex $a$  incident
        to loop of opposite colour to the loops in $\bA_\ell$.
    \end{itemize}
\end{lemma}
\begin{proof}
    By Proposition~\ref{prop:2-element-constructions}, we assume that
    $\bH:=\bA_1 \triangleleft_h \dots \triangleleft_h \bA_k$  where
    $|A_i|\le 2$ for each $i\in[k]$.
    Notice that if there is some $i\in[k]$ such that
    $\bA_i$ and $\bA_{i+1}$ are monochromatic components of the same
    colour, then $\bH$ can be equivalently decomposed
    as 
    \[
        \bH:= \bA_1 \triangleleft_h \dots \bA_{i+1} \triangleleft_h \bA_i\triangleleft_h \dots \triangleleft_h \bA_k.
    \]
    We choose a decomposition $\bH:=\bA_1 \triangleleft_h \dots \triangleleft_h \bA_k$
    such that if $\bA_i$ and $\bA_{i+1}$ are monochromatic components of the 
    same colour, and one of them is a loop, then $\bA_{i+1}$ is a loop.
    Notice if each $\bA_i$ is either an $\ast$-edge or a loop
    for $i> 2$, then $\bH$ admits a decomposition as in Theorem~\ref{thm:homogeneous-Ptime}. 
    If this is not the case, let $i>2$
    be the first integer such that $\bA_i$ is isomorphic to 
    $\bK_2^\ast$ or to $\overline{\bK_2^\ast}$, and assume without loss
    of generality that $\bA_i\cong \bK_2^\ast$.

    We first assume that Claim 1 (below) holds, i.e., 
    $(\bH, h_1,\dots, h_n)$ pp-defines  $A_1\cup \dots \cup A_i$. 
    We show that $A_{i-1}$ contains  a vertex $a$ such that $aa\in B(\bH)$. If
    not, $\bA_{i-1}$ is a red alternating component. Since $\bH$ is a core,
    $\bA_{i-1}$ is not an $\ast$-edge (otherwise, collapsing $\bA_i$ into $\bA_{i-1}$
    yields a non-surjective endomorphism of $\bH$). Also, 
    $\bA_{i-1} \not\cong \bK_2^\ast$ by Lemma~\ref{lem:consecutive-red}, and
    so $\bA_{i-1}$ is a red loop $aa$. But this contradicts the choice
    of decomposition $\bA_1 \triangleleft_h \dots \triangleleft_h \bA_k$.
    Hence, the statement of this lemma follows by proving the following claim.

    \vspace{0.2cm}
    \noindent
    \textbf{Claim 1. } $(\bH,h_1,\dots, h_n)$ pp-defines the union
    $A_1\cup \dots \cup A_i$.

    \vspace{0.1cm}
    \noindent
    \textit{Proof of Claim 1.}
    Since $A_k$ is a homogeneous set in $\bH$,
    every homomorphism  $f\colon \bA_1\triangleleft_h\dots \triangleleft_h \bA_{k-1}\to\bA_1\triangleleft_h\dots \triangleleft_h \bA_{k-1}$ extends to a homomorphism 
    $f'\colon \bH\to \bH$ by acting as the identity on $A_k$. 
    In particular, if $\bH$ is a core, then 
    $\bA_1\triangleleft_h\dots \triangleleft_h \bA_{k-1}$ is a core.
    Hence, it suffices to prove that $(\bH, h_1,\dots, h_n)$ 
    pp-defines the a set $A_1\cup \dots \cup A_{k-1}$,
    and the claim follows by finite induction for $i < k-1$.
    Consider now the following (boring) case distinction.
    \begin{itemize}
        \item $|A_k| = 1$, and without loss of generality, assume
        that $\bA_k$ induces a blue loop $aa\in B(\bH)$. In this case,
        $a$ is not incident to any red edges, and every vertex
        different not in $A_k$ must a have a red neighbour because
        $\bH$ is a core. Hence, $\exists z.\; R(x,z)$ pp-defines
        $A_1\cup \dots \cup A_{k-1}$, and $i\le k-1$.
        \item    
        $\bA_k$ induces a monochromatic $\ast$-edge, and
        we assume that $\bA_k = \{a,b\}$ where $aa,bb,ab\in B(\bH)$
        and $ab\in R(\bH)$.
        As $\bH$ is a core, $\bA_{k-1}$ can be either $\bK_1^*$, $\bK_2^\ast$,
        a red monochromatic $*$-edge,  or a bichromatic $*$-edge.
        In the first 3 cases,  $A_{k-1}$ is pp-definable (Lemma~\ref{lem:mono-components}),
        and so the formula $\exists z.\; z\in A_j\wedge R(x,z)$ pp-defines $A_1\cup \dots \cup A_{k-1}$.
        In the last case, let $c$ be the red vertex of
        of $\bA_{k-1}$, then the formula $R(x,c)$ pp-defines $A_1\cup \dots \cup A_{k-1}$.
        
       \item $\bA_k\cong \overline{\bK_2^\ast}$. Since $\bH$ does not
       pp-construct $\bK_3$, it follows from Lemma~\ref{lem:consecutive-red}
       that $\bA_{k-1}\not\cong \overline{\bK_2^\ast}$. 
       Since $\bH$ is a core, 
       $\bA_{k-1}$ is not an $\ast$-edge connecting vertices in $H_R$. 
       By the choice of decomposition, $\bA_{k-1}$ is not a red loop. 
       Hence, $\bA_{k-1}$ can be either $\bK_1^*$,  $\bK_2^\ast$,
       a red monochromatic $\ast$-edge, or a bichromatic $\ast$-edge
       and so, with the same arguments as in the previous case we see that
        $A_1\cup \dots A_{k-1}$ is pp-definable.
       
       \item $\bA_k\cong \bK_2^\ast$. Follows with symmetric
       arguments as the previous case.
       \item $\bA_k$ is a bichromatic $\ast$-edge $br\in B(\bH)\cap R(\bH)$
       with $bb\in B(\bH)$ and $rr\in R(\bH)$. Using the fact that
       $\bH$ is a core, we see that $\bA_k$ is not a bichromatic
       $\ast$-edge, and also not a loop. Hence $\bA_{k-1}$ is
       either a monochromatic $\ast$-edge, $\bK_2^\ast$ or 
       $\overline{\bK_2^\ast}$. By Lemma~\ref{lem:mono-components}, 
       the set $A_{k-1}$ is pp-definable in $(\bH,h_1,\dots, h_n)$. 
       Assume without loss of generality that $A_{k-1}\subseteq H_R$.
       The pp-formula $\exists z.\; z\in A_{k-1} \land R(x,z)$ pp-defines
       the set $A_1\cup \dots A_{k-1}\cup\{r\}$, and the substructure
       induces by this set admits a decomposition
       $\bA_1\triangleleft_h \dots \triangleleft_h \bA_{k-1}\triangleleft \bA'_k$
       where $A'_k =\{r\}$. Thus, with similar arguments as in the
       first case, $\bA_1\triangleleft_h \dots \triangleleft_h \bA_{k-1}\triangleleft \bA'_k$
       pp-defines $A_1\cup\dots \cup A_{k-1}$, and by composing pp-definitions
       we conclude that $(\bH,h_1,\dots, h_n)$ pp-defines
       $A_1\cup\dots \cup A_{k-1}$.
    \end{itemize}
    These are all possible cases because each $\bA_i$ contains at most
    at most two vertices, and so Claim 1 and this lemma follow.
\end{proof}

\begin{lemma}\label{lem:K2-case}
    Let $\bH$ be a reflexive complete $2$-edge-coloured  graph
    that is a core. If $\bH$ does not pp-construct
    $\bK_3$, then $\bH$ admits a decomposition as in Theorem~\ref{thm:homogeneous-Ptime}.
\end{lemma}
\begin{proof}
    By Lemma~\ref{lem:K2-on-top}, it suffices to show that
    $\bH$ pp-constructs $\bK_3$ whenever
    it admits a decomposition
    \[
    \bH:=\bA_1 \triangleleft_h \dots \triangleleft_h \bA_k,
    \]
    where $k$ is the first index $i\ge 2$ such that $\bA_i$
    is isomorphic to $\bK_2^\ast$ or to $\overline{\bK_2^\ast}$, 
    and $A_{k-1}$ contains a vertex $a$  incident
    to loop of opposite colour to the loops in $\bA_k$.
    Up to colour symmetry, we may further assume that
    $\bA_k\cong \bK_2^\ast$ and so $A_{k-1}$ contains
    a blue loop. First consider the case when 
    $\bA_{k-1}$ is a blue monochromatic component with
    two vertices $a$ and $b$. In particular, $ab\in R(\bH)$
    and $cb\in R(\bH)$ for each $c\in A_k$. Also, 
    $cb\not\in R(\bH)$ for each $c\in H\setminus (A_k\cup\{a\})$,
    and so $R(x,b)$ pp-defines the subset $A_k\cup \{a\}$.
    It is straightforward to observe that the binary
    relation $E(x,y):=\exists z\in A_k\cup\{a\}.\; B(x,z)\land R(z,y)$
    is the inequality relation in $A_k\cup\{a\}$. In other
    words, $(A_k\cup\{a\}, E)\cong \bK_3$, and hence 
    $(\bH,h_1,\dots, h_n)$ pp-defines $\bK_3$, and
    so $\bH$ pp-constructs $\bK_3$.

    From now on, we assume that either $\bA_{k-1}$ is a blue
    loop, or a bichromatic $\ast$-edge. Let us denote
    by $0$ and $1$ the vertices in $A_k$, and by $2$
    the unique vertex in $A_{k-1}$ incident to a blue loop. 
    In particular, $01, 22\in B(\bH)$ and $00,11,02,12\in R(\bH)$
    (and no other edges are
    induced by $\{0,1,2\}$).  Proceeding by contradiction,
    suppose that $\bH$ does not pp-construct $\bK_3$. Hence,
    by Lemma~\ref{lem:cyclic-power} there is a homomorphism
    $f\colon \Cyc_p(\bH)\to \bH$ where $p$ is a prime larger
    than $|H|$. We further choose $p$ to be congruent
    to $1$ modulo $4$, and let $m$ be such that $p = 4m+1$.
    We follow the notation introduced in the paragraph before 
    Proposition~\ref{prop:2-element-constructions}.
    Through the rest of the proof we fix the following elements
    in $\Cyc_p(\bH)$: $a:=(0212)^m0$,  $b= (1202)^m1$, and
    $c=(0110)^m0$. It follows from the definition $\sim_p$ that 
    $aa,bb,ac,bc\in R(\Cyc_p(\bH))$ and $ab\in B(\Cyc_p(\bH))$.
    Since $A_k$ is pp-definable in $(\bH,h_1,\dots, h_n)$,
    it must be the case that $f(c)\in \{0,1\}$.
    Now, consider the following case distinction.
    \begin{itemize}
        \item Every monochromatic red component
        $\bA_j$ with $j < k$ consists of a single vertex.
         In this case, since $aa,bb\in R(\Cyc_p(\bH))$,
         $ab\in B(\Cyc_p(\bH))$, and
        $\bA_k$ is the only red alternating component
        with a blue edge, it must be the case that
        $\{f(a),f(b)\} =\{0,1\}$. Using the
        fact that $ac,bc\in R(\bH)$ and that $f(c)\in\{0,1\}$,
        we conclude that $0$ and $1$ are connected by a red
        edge in $\bH$. This contradicts the assumption
        that $\bA_k\cong\bK_2^\ast$.
        \item There is some red monochromatic component
        $\bA_j$ with $j <k$ and $|A_j| = 2$. Let $\ell$ be
        the maximum integer $j < k$
        such that $\bA_j$ is a red alternating component. 
        Let $A_\ell = \{3,4\}$. In particular, 
        $33,44\in R(\bH)$ and $34\in B(\bH)$. Let $S\subseteq H$
        be the subset defined by $B(x,2)\land \exists z.\; z\in A_\ell\land B(x,z)$.
        Since $(\bH,h_1,\dots, h_n)$ pp-defines the subset $A_\ell$, 
        the set $S$ is pp-definable in $(\bH,h_1,\dots, h_n)$. 
        With similar arguments as before (e.g., using homogeneity
        of $\bA_j$ in $\bA_1\triangleleft_h \dots \triangleleft_h\bA_j$),
        it follows that $S = \{2,3,4\}\cup \{u\in H_B\cap A_j\colon \ell < j <k-1\}$.
        Let $d := (2323)^m2 \in \Cyc_p(\bH)$, and clearly $dd\in B(\Cyc_p(\bH))$. Since 
        $33,44\in R(\bH)$ and $S$ is pp-definable in $(\bH,h_1,\dots, h_n)$ and $2,3\in S$,
        the image $f(d)$ belongs to $S\setminus\{3,4\}$. Similarly, 
        the image of $d':=(2424)^m 2$ belongs 
        to $S\setminus\{3,4\}$.
        Finally, consider the elements $e:= (0303)^m0$ and $e':= (1414)^m 1$.
        It readily follows that $de,d'e'\in R(\Cyc_p(\bH))$ and $ee'\in B(\Cyc_p(\bH))$. 
        Since $f(d)\in A_j$ and $f(d')\in A_{j'}$ for some $j,j'< k$
        and $f(d)$ and $f(d')$ are incident to blue loops, it must be the
        case that $f(e),f(e') \in \{0,1\}$ because $de,d'e'\in R(\Cyc_p(\bH))$.
        Now, using the fact that $ee'\in B(\Cyc_p(\bH))$, we see
        that $\{f(e),f(e')\} =\{0,1\}$. We arrive to a contradiction
        by recalling that $f(c)\in\{0,1\}$ (see discussion before the case
        distinction), and noticing that $ce,ce'\in R(\bH)$. Indeed,
        this implies that $01\in B(\bH)\cap R(\bH)$, but $\bA_k\cong \bK_2^\ast$.
    \end{itemize}
\end{proof}

\subsection{A structural dichotomy}

Our main theorem now follows from Theorem~\ref{thm:homogeneous-Ptime}
and Lemma~\ref{lem:K2-case}.

\begin{theorem}\label{thm:main}
    For every reflexive complete $2$-edge-coloured  graph $\bH$ one of the following
    holds:
    \begin{itemize}
        \item  
        either the core of $\bH$
        admits a decomposition 
            \[
                \bA_1 \triangleleft_h \dots \triangleleft_h \bA_k,
            \]
    where $\bA_1$ is a structure on at most two elements, and $\bA_i$
    is a single vertex (with a loop), or a reflexive $\ast$-edge
    for each $2\le i\le k$ --- and in this case
    $\CSP(\bH)$ can be solved in polynomial time, or
    \item
     $\bH$ pp-constructs $\bK_3$, and $\CSP(\bH)$ is  $\NP$-complete.
    \end{itemize}
\end{theorem}

The following statement is implied by this theorem and Proposition~\ref{prop:recognizing_tractables_in_ptime}.

\begin{corollary}\label{cor:dichotomy-P}
    For a reflexive complete $2$-edge-coloured graph $\bH$, one can check in
    polynomial time whether $\CSP(\bH)$ is in $\PO$ or $\NP$-complete.
\end{corollary}

\section{A CSP approach to Graph Sandwich Problems}
\label{sec:CSP-SP}

Very recently, Bodirsky and Guzm\'an-Pro introduced a CSP approach to 
Graph Sandwich Problems~\cite{bodirskySODA2026}.
They noticed that 
for several graph classes $\mathcal C$, the SP for $\mathcal C$
can be naturally modelled as an infinite-domain CSP.
We now observe that SP for matrix partitions fall in this scope.

In Section~\ref{sec:full-hom} we introduced the notions of
twins, and blows-ups. A pair of different vertices $u$ and $v$
of $\bG$ are \emph{co-twins} if $uv\in E(\bG)$ and for every $w\in G$
there is an edge $uw$ in $\bG$ if and only if there is an edge $vw$ in $\bG$.
We say that $\bG'$ is a \emph{co-blow-up} of $\bG$ if $\bG'$ can
be obtained from $\bG$ by iteratively adding co-twins.
Now, consider a partition $(I,C)$ of $G$ (with possibly one empty part).
A \emph{split blow-up} of $\bG$ \emph{with respect to} $(I,C)$ is a graph $\bG'$ obtained 
from $\bG$ by iteratively replacing vertices $u\in I$ with a set of twins, and vertices 
$v\in C$ with a set of co-twins (so $I$ stands for ``independent'' because a vertex
$v\in I$ is blown-up to an independent set, and $C$ stands for ``complete''). 
We say that a class $\mathcal C$ is preserved by \emph{split blow-ups} if for every
graph $G$ in $\mathcal C$, there is a partition $(I,C)$ (with possibly one empty part) 
of the vertices of $G$ such that any split blow-up of $G$ with respect to $(I,C)$
belongs to $\mathcal C$ as well. 

It is straightforward to observe that for every matrix $M$ the class of
graphs that admit an $M$-partition is preserved under split-blow ups. 
Moreover, it is closed under induced subgraphs, and has the 
\emph{joint-embedding property}, i.e., if
$\bG$ and $\bH$ admit an $M$-partition, then there is a graph
$\mathbb F$ that contains $\bG$ and $\bH$ as induced subgraphs
and $\mathbb F$ admits an $M$-partition. Hence, it follows from
Proposition~3 in~\cite{bodirskySODA2026} that the SP
for $M$-partitions is the CSP of a $2$-edge-coloured complete graph $\mathbb M$.
Here we argue that these sandwich problems are  polynomial-time
equivalent to $\CSP(\mathbb M)$.

Consider the following equivalent representation 
of the SP for a class $\mathcal C$. The input is a triple
$(V,E,N)$ where $N$ is a set of non-edges and $E \cap N = \varnothing$. 
The task is to find an edge set $E'$ such that $E\subseteq E'$,
that $E'\cap N = \varnothing$, and $(V,E')$ belongs to $\mathcal C$. 
It is straightforward to observe that this definition of sandwich
problems is equivalent to the definition considered in the Introduction
(see also~\cite{golumbicJA19}).
In this setting, we think of $(V,E,N)$ as a $2$-edge-coloured graph
where $E$ is the set of blue edges, and $N$ the set of red edges.

\begin{lemma}\label{lem:SP->CSP}
    For every symmetric matrix $M$ over $\{0,1,\ast\}$, the following statements hold.
    \begin{enumerate}
        \item The Sandwich Problem for $M$-partitions is polynomial-time equivalent
        to $\CSP(\mathbb M)$.
        \item The Sandwich Problem for list $M$-partitions is polynomial-time equivalent to the list
        version of $\CSP(\mathbb M)$.
    \end{enumerate}
\end{lemma}
\begin{proof}
    It is straightforward to observe that if $(V,E,N)$ is an input
    to the SP for the $M$-partition, then $(V,E,N)$ is a yes-instance
    if and only if the $2$-edge-coloured graph $\bH := (V, R(\bH) = N, B(\bH)= E)$
    maps homomorphically to $\mathbb M$.  Moreover, notice that this is a reduction
    is a one-to-one reduction, where its image is the class of loopless
    $2$-edge-coloured graphs with no $\ast$-edge.
    Hence, if $(V,R,B)$ is a loopless $2$-edge-coloured graph
    with no $\ast$-edges, then $(V, E= B, N =R)$ is a yes-instance
    of the SP for the $M$-partition if and only if  $(V,B,R)\to \mathbb M$.
    Finally, it follows from the Sparse Incomparability Lemma (Theorem~\ref{thm:large-girth})
    that $\CSP(\mathbb M)$ is polynomial-time equivalent to $\CSP(\mathbb M)$ restricted to
    instances of girth at least $3$, i.e., loopless $2$-edge-coloured graphs with no
    $\ast$-edge. This proves the first itemized claim, and with similar
    arguments one can prove the second one. 
\end{proof}

It now follows from Theorem~\ref{thm:CSP-dichotomy}
that SP problem for matrix partitions exhibit a P vs.\ NP-complete dichotomy. 
Moreover,  Theorem~\ref{thm:main} presents a structural dichotomy 
for this family of problems. 

\begin{corollary}
    Sandwich Problems for matrix partitions exhibit a $\PO$
    vs.\ $\NP$-complete dichotomy (see also Theorem~\ref{thm:main}).
\end{corollary}

\subsection{Full-homomorphisms}

As mentioned in Section~\ref{sec:full-hom}, a natural subclass
of matrix partition problems corresponds to full-homomorphism problems. 
The following statement is an immediate consequence of Observation~\ref{obs:K2+K1},
Proposition~\ref{prop:full-hom}, and Lemma~\ref{lem:SP->CSP}.

\begin{corollary}\label{cor:full-hom}
    For every loopless graph $\bH$ 
    one of the following statements holds:
    \begin{itemize}
        \item $\bH$
        admits a full-homomorphism to $\bK_1 + \bK_2$,
        and in this case the SP for the full-homomorphism to $\bH$, and its list
        variant are polynomial-time solvable;
        \item otherwise, $\bH$
        contains one of $\bK_3$, $2\bK_2$, $\bP_4$ as an induced subgraph,
        and in this case the SP for the full-homomorphism problem to $\bH$ is $\NP$-complete.
    \end{itemize}
\end{corollary}

It turns out that we can extend the previous structural classification
of SPs of full-homomorphisms of loopless graphs to arbitrary graphs.
This structural classification is closely related to 
\emph{threshold graphs}. Threshold graphs have several equivalent
definitions~\cite{mahaved1995}. In particular, a loopless graph $\bH$
is a threshold if and only if $\bH$ can be constructed from a single vertex by
repeatedly adding an independent vertex or a dominant vertex.

For graph $\bH$  with possible loops, we say  that a vertex $v$ is 
a \emph{dominant loop} if $vv\in E(\bH)$, and for each $u\in H$
there is an edge $uv\in E(\bH)$. Recall that every graph $\bH$
     contains
an induced subgraph $\bH'$ (unique up-to isomorphism) such that $\bH'$
is point-determining, and there is a full-homomorphism from $\bH$
     to $\bH'$
(and vice-versa).
We call $\bH'$ the \emph{point-determining core} of $\bH$. 
We write $\mathbb L$ to denote the graph consisting of one vertex incident to a loop.

\begin{corollary}\label{cor:full-hom-non-simple}
    For every graph $\bH$  (with possible loops) whose point-determining core
    is $\bH'$, one of the
    following statements holds:
    \begin{itemize}
        \item either $\bH'$ 
        can be constructed by
        repeatedly adding a loopless isolated vertex or a dominant loop, starting from 
        a single vertex, from $2\mathbb L$, or from $\bK_2$ --- and in this
        case, the (list version of the) SP for the full-homomorphism to $\bH$ 
        can be solved in polynomial time; or
        \item otherwise, the (list version of the) SP for the full-homomorphism to $\bH$ 
        is $\NP$-complete.
    \end{itemize}
\end{corollary}
\begin{proof}
    Let $\mathbb M$ and $\mathbb M'$ be the $2$-edge-coloured graphs associated to the 
    adjacency matrices of $\bH$ and of  the point-determining core $\bH'$ 
    of $\bH$. 
    Since $\mathbb M$ has no $\ast$-edges, 
    the dichotomy for the non-list version follows from Theorem~\ref{thm:main} via
    Lemma~\ref{lem:SP->CSP}. The hardness of the list version claimed in the second item,
    follows from the hardness of the non-list version. Finally, the tractability of the
    list version claimed in the first item follows from Theorem~\ref{thm:homogeneous-Ptime}
    together with Lemma~\ref{lem:full-hom-LAC} because $\mathbb M$ admits a
    full-homomorphism to $\mathbb M'$.
\end{proof}

\subsection{List version}

Notice that Theorem~\ref{thm:homogeneous-Ptime}
together with Theorem~\ref{thm:main}, also present a complexity
classification for the list version of SP of matrix partitions, 
when the associated structure $\mathbb M$ is a core.
In particular,  we obtain the following consequence.

\begin{corollary}
    For every matrix $M$ such that $\mathbb M$ is a core, the 
    SP for the M-partition is polynomial-time equivalent to 
    its list variant.
\end{corollary}

The following corollary follows from Lemma~\ref{lem:SP->CSP} and
the finite-domain dichotomy (Theorem~\ref{thm:CSP-dichotomy}).

\begin{corollary}
    For every matrix $M$,
    the list version for the $M$-partition
    problem is either in $\PO$ or $\NP$-complete.
\end{corollary}

However, Theorem~\ref{thm:main} does not yield a complexity classification
of the list version of CSPs of reflexive complete $2$-edge-coloured graphs (eq.\ of
the list version of  SPs of matrix partitions problems). Indeed,  all
we know from this result is the following: if $\bM$ is a core that admits
a decomposition as in Theorem~\ref{thm:main}, then the list version of the SP for
$M$-partitions if solvable in polynomial time; also, if the core of
$\mathbb M$ does not admit such a decomposition, then the list version
of the SP for $M$-partitions is NP-complete. However, all cases when
$\mathbb M$ is not a core, but its core admits such a decomposition do
not follow from Theorem~\ref{thm:main}. 

\begin{problem}\label{prob:list-version}
    Present a structural classification of the complexity of the list
    version of CSPs of reflexive complete $2$-edge-coloured  graphs. 
    Equivalently, classify the complexity of the list version of
    sandwich problems for matrix partitions.\footnote{Notice that the case
    when $\mathbb M$ has no $\ast$-edges nor $\ast$-loops, is settled
    by Corollary~\ref{cor:full-hom-non-simple}.}
\end{problem}

\section{Conclusions and open problems}
\label{sec:conclusion}

In this paper we presented a structural understanding of
the algorithmic power of the natural reduction of matrix partitions
to CSPs of reflexive complete $2$-edge-coloured graphs.
To do so, we presented a structural classification of the complexity of the latter class.
As a byproduct of this result, we presented a complexity classification
for the sandwich problem for matrix partitions problems. 
A natural quest is to pursue a structural complexity classification of
CSPs of all $2$-edge-coloured graphs. However, for every digraph CSP there
is a polynomial-time equivalent CSP of a $2$-edge-coloured
graph~\cite[Theorem 3.1]{brewsterDM340}.
So presenting a structural dichotomy for CSPs of $2$-edge-coloured graphs
might be as hard as presenting a structural classification of digraph CSPs
(which seems currently out of reach).
On the other hand, we believe that our structural classification can be extended
to (not necessarily reflexive) complete $2$-edge-coloured graphs,
and perhaps to reflexive (not necessarily complete) $2$-edge-coloured graphs.

\begin{problem}
    Present a $\PO$ vs.\ $\NP$-complete structural classification
    of CSPs of complete $2$-edge-coloured graphs $\bH$.
\end{problem}

\begin{problem}
    Present a $\PO$ vs.\ $\NP$-complete structural classification
    of CSPs of reflexive  $2$-edge-coloured graphs $\bH$.
\end{problem}

We also leave open a complexity classification of the list version 
of SP for $M$-partition problems (Problem~\ref{prob:list-version}).
Finally, analogous to the program
of characterising hereditarily hard digraphs (initiated in~\cite{bangjensenDM138},
concluded as an application from~\cite{BartoKozikNiven}) we leave
a classification of all $2$-edge-coloured graphs that hereditarily pp-construct
$\bK_3$ for future research (see Appendix~\ref{ap:HH} for a few more small examples).

\bibliographystyle{abbrv}
\bibliography{global.bib}

\def\cprime{$'$} \def\cprime{$'$} \def\cprime{$'$}
\begin{thebibliography}{10}

\bibitem{alonJAC8}
N.~Alon and T.~H. Marshall.
\newblock Homomorphisms of edge-colored graphs and coxeter groups.
\newblock {\em J. Algebraic Comb.}, 8(1):5–13, July 1998.

\bibitem{alvaradoAOR280}
J.~Alvarado, S.~Dantas, and D.~Rautenbach.
\newblock Sandwiches missing two ingredients of order four.
\newblock {\em Annals of Operations Research}, 280:47--63, 2019.

\bibitem{bangjensenSIDMA1}
J.~Bang-Jensen, P.~Hell, and G.~MacGillivray.
\newblock The complexity of colouring by semicomplete digraphs.
\newblock {\em SIAM Journal on Discrete Mathematics}, 1(3):281--298, 1988.

\bibitem{bangjensenDM138}
J.~Bang-Jensen, P.~Hell, and G.~MacGillivray.
\newblock Hereditarily hard h-colouring problems.
\newblock {\em Discrete Mathematics}, 138(1):75--92, 1995.
\newblock 14th British Combinatorial Conference.

\bibitem{BartoWidth}
L.~Barto.
\newblock The collapse of the bounded width hierarchy.
\newblock {\em Journal of Logic and Computation}, 26(3):923--943, 2016.

\bibitem{BoundedWidthJournal}
L.~Barto and M.~Kozik.
\newblock Constraint satisfaction problems solvable by local consistency methods.
\newblock {\em Journal of the {ACM}}, 61(1):3:1--3:19, 2014.

\bibitem{BartoKozikNiven}
L.~Barto, M.~Kozik, and T.~Niven.
\newblock The {CSP} dichotomy holds for digraphs with no sources and no sinks (a positive answer to a conjecture of {B}ang-{J}ensen and {H}ell).
\newblock {\em SIAM Journal on Computing}, 38(5), 2009.

\bibitem{Pol}
L.~Barto, A.~A. Krokhin, and R.~Willard.
\newblock Polymorphisms, and how to use them.
\newblock In {\em The Constraint Satisfaction Problem: Complexity and Approximability}, pages 1--44. Schloss Dagstuhl - Leibniz-Zentrum fuer Informatik, 2017.

\bibitem{wonderland}
L.~Barto, J.~Opr\v{s}al, and M.~Pinsker.
\newblock The wonderland of reflections.
\newblock {\em Israel Journal of Mathematics}, 223(1):363--398, 2018.

\bibitem{Book}
M.~Bodirsky.
\newblock {\em Complexity of Infinite-Domain Constraint Satisfaction}.
\newblock Lecture Notes in Logic (52). Cambridge University Press, Cambridge, United Kingdom; New York, NY, 2021.

\bibitem{bodirskySODA2026}
M.~Bodirsky and S.~Guzmán-Pro.
\newblock A {CSP approach to Graph Sandwich Problems}.
\newblock In {\em Proceedings of the 2026 Annual ACM-SIAM Symposium on Discrete Algorithms (SODA)}, pages 2403--2418.

\bibitem{bokDM346}
J.~Bok, R.~Brewster, T.~Feder, P.~Hell, and N.~Jedli{\v{c}}kov{\'a}.
\newblock List homomorphism problems for signed trees.
\newblock {\em Discrete Mathematics}, 346(3):113257, 2023.

\bibitem{bokMFCS2020}
J.~Bok, R.~Brewster, T.~Feder, P.~Hell, and N.~Jedli\v{c}kov\'{a}.
\newblock {List Homomorphism Problems for Signed Graphs}.
\newblock In J.~Esparza and D.~Kr\'{a}l', editors, {\em 45th International Symposium on Mathematical Foundations of Computer Science (MFCS 2020)}, volume 170 of {\em Leibniz International Proceedings in Informatics (LIPIcs)}, pages 20:1--20:14, Dagstuhl, Germany, 2020. Schloss Dagstuhl -- Leibniz-Zentrum f{\"u}r Informatik.

\bibitem{bokTCS1001}
J.~Bok, R.~Brewster, T.~Feder, P.~Hell, and N.~Jedličková.
\newblock List homomorphisms to separable signed graphs.
\newblock {\em Theoretical Computer Science}, 1001:114580, 2024.

\bibitem{bondy2008}
J.~A. Bondy and U.~S.~R. Murty.
\newblock {\em Graph Theory}.
\newblock Springer, 2008.

\bibitem{brewsterDAM49}
R.~Brewster.
\newblock The complexity of colouring symmetric relational systems.
\newblock {\em Discrete Applied Mathematics}, 49(1):95--105, 1994.
\newblock Special Volume Viewpoints on Optimization.

\bibitem{brewster1993vertex}
R.~Brewster et~al.
\newblock Vertex colourings of edge-coloured graphs.
\newblock 1993.

\bibitem{brewsterJGT110}
R.~Brewster, A.~Kidner, and G.~MacGillivray.
\newblock A dichotomy theorem for $\gamma$-switchable h-colouring on m-edge-coloured graphs.
\newblock {\em Journal of Graph Theory}, 110(2):200--208, 2025.

\bibitem{brewsterDM340}
R.~C. Brewster, F.~Foucaud, P.~Hell, and R.~Naserasr.
\newblock The complexity of signed graph and edge-coloured graph homomorphisms.
\newblock {\em Discrete Mathematics}, 340(2):223--235, 2017.

\bibitem{brewsterENDM5}
R.~C. Brewster and P.~Hell.
\newblock On homomorphisms to edge-coloured cycles.
\newblock {\em Electronic Notes in Discrete Mathematics}, 5:46--49, 2000.

\bibitem{BrewsterSiggersSigned}
R.~C. Brewster and M.~Siggers.
\newblock A complexity dichotomy for signed {$\bold{H}$}-colouring.
\newblock {\em Discrete Math.}, 341(10):2768--2773, 2018.

\bibitem{Conservative}
A.~A. Bulatov.
\newblock Tractable conservative constraint satisfaction problems.
\newblock In {\em Proceedings of the Symposium on Logic in Computer Science {(LICS)}}, pages 321--330, Ottawa, Canada, 2003.

\bibitem{BulatovFVConjecture}
A.~A. Bulatov.
\newblock A dichotomy theorem for nonuniform {CSP}s.
\newblock In {\em 58th {IEEE} Annual Symposium on Foundations of Computer Science, {FOCS} 2017, {B}erkeley, {CA}, {USA}, {O}ctober 15-17}, pages 319--330, 2017.

\bibitem{JBK}
A.~A. Bulatov, A.~A. Krokhin, and P.~G. Jeavons.
\newblock Classifying the complexity of constraints using finite algebras.
\newblock {\em SIAM Journal on Computing}, 34:720--742, 2005.

\bibitem{cameronSODA2004}
K.~Cameron, E.~M. Eschen, C.~T. Ho\`{a}ng, and R.~Sritharan.
\newblock The list partition problem for graphs.
\newblock In {\em Proceedings of the Fifteenth Annual ACM-SIAM Symposium on Discrete Algorithms}, SODA '04, page 391–399, USA, 2004. Society for Industrial and Applied Mathematics.

\bibitem{cameronSIDMA21}
K.~Cameron, E.~M. Eschen, C.~T. Hoàng, and R.~Sritharan.
\newblock The complexity of the list partition problem for graphs.
\newblock {\em SIAM Journal on Discrete Mathematics}, 21(4):900 – 929, 2007.
\newblock Cited by: 29.

\bibitem{MetaChenLarose}
H.~Chen and B.~Larose.
\newblock Asking the metaquestions in constraint tractability.
\newblock {\em {TOCT}}, 9(3):11:1--11:27, 2017.

\bibitem{cookDM310}
K.~Cook, S.~Dantas, E.~M. Eschen, L.~Faria, C.~M. de~Figueiredo, and S.~Klein.
\newblock $2k_2$ vertex-set partition into nonempty parts.
\newblock {\em Discrete Mathematics}, 310:1259--1264, 2010.

\bibitem{cyganSODA2011}
M.~Cygan, M.~Pilipczuk, M.~Pilipczuk, and J.~O. Wojtaszczyk.
\newblock The stubborn problem is stubborn no more (a polynomial algorithm for 3–compatible colouring and the stubborn list partition problem).
\newblock In {\em Proceedings of the 2011 Annual {ACM}-{SIAM} Symposium on Discrete Algorithms (SODA)}, pages 1666--1674, 2011.

\bibitem{dantasDAM143}
S.~Dantas, C.~M. de~Figueiredo, and L.~Faria.
\newblock Characterizations, probe and sandwich problems on (k,l)-cographs.
\newblock {\em Discrete Applied Mathematics}, 143:155--165, 2004.

\bibitem{dantasAOR188}
S.~Dantas, C.~M. de~Figueiredo, M.~C. Golumbic, S.~Klein, and F.~Maffray.
\newblock The chain graph sandwich problem.
\newblock {\em Annals of Operation Research}, 188:133--139, 2011.

\bibitem{dantasDAM182}
S.~Dantas, C.~M. de~Figueiredo, F.~Maffray, and R.~B. Teixeira.
\newblock The complexity of forbidden subgraph sandwich problems and the skew partition sandwich problem.
\newblock {\em Discrete Applied Mathematics}, 182:15--24, 2015.

\bibitem{dantasENTCS346}
S.~Dantas, C.~M. de~Figueiredo, P.~Petito, and R.~B. Teixeira.
\newblock A general method for forbidden induced subgraph sandwich problem np-completeness.
\newblock {\em Electronic Notes in Theoretical Computer Science}, 346:393--400, 2019.

\bibitem{dantasICCGI07}
S.~Dantas and L.~Faria.
\newblock On stubborn graph sandwich problems.
\newblock {\em 2007 International Multi-Conference on Computing in the Global Information Technology (ICCGI'07)}, pages 46--46, 2007.

\bibitem{dantasCATS2010}
S.~Dantas, L.~Faria, C.~M.~H. de~Figueiredo, S.~Klein, L.~T. Nogueira, and F.~Protti.
\newblock Advances on the list stubborn problem.
\newblock In T.~Viglas and A.~Potanin, editors, {\em Theory of Computing 2010, {CATS} 2010, Brisbane, Australia, January 2010}, volume 109 of {\em {CRPIT}}, pages 65--70. Australian Computer Society, 2010.

\bibitem{figueiredoJA37}
C.~M. de~Figueiredo, S.~Klein, Y.~Kohayakawa, and B.~A. Reed.
\newblock Finding skew partitions efficiently.
\newblock {\em Journal of Algorithms}, 37:505--521, 2000.

\bibitem{figueiredoDAM121}
C.~M. de~Figueiredo, S.~Klein, and K.~Vu\v{s}kovi\'c.
\newblock The graph sandwich problem for 1-join composition is {NP}-complete.
\newblock {\em Discrete Applied Mathematics}, 121:73--82, 2002.

\bibitem{Erd:Gtp}
P.~Erd{\H o}s.
\newblock Graph theory and probability.
\newblock {\em Canadian Journal of Mathematics}, 11:34--38, 1959.

\bibitem{federDM306}
T.~Feder and P.~Hell.
\newblock Matrix partitions of perfect graphs.
\newblock {\em Discrete Mathematics}, 306:2450 -- 2460, 2006.

\bibitem{feder_hell_full_hom}
T.~Feder and P.~Hell.
\newblock On realizations of point determining graphs, and obstructions to full homomorphisms.
\newblock {\em Discret. Math.}, 308(9):1639--1652, 2008.

\bibitem{ListStructuralAll}
T.~Feder, P.~Hell, and J.~Huang.
\newblock Bi-arc graphs and the complexity of list homomorphisms.
\newblock {\em J. Graph Theory}, 42(1):61--80, 2003.

\bibitem{federSIDMA16}
T.~Feder., P.~Hell, S.~Klein, and R.~Motwani.
\newblock List partitions.
\newblock {\em SIAM Journal on Discrete Mathematics}, 16(3):449--478, 2003.

\bibitem{federTCS349}
T.~Feder, P.~Hell, S.~Klein, L.~T. Nogueira, and F.~Protti.
\newblock List matrix partitions of chordal graphs.
\newblock {\em Theoretical Computer Science}, 349:52--66, 2005.

\bibitem{federSODA2005}
T.~Feder, P.~Hell, D.~Kr\'{a}l, and J.~Sgall.
\newblock Two algorithms for general list matrix partitions.
\newblock In {\em Proceedings of the Sixteenth Annual ACM-SIAM Symposium on Discrete Algorithms}, SODA '05, page 870–876, USA, 2005. Society for Industrial and Applied Mathematics.

\bibitem{federDAM159}
T.~Feder, P.~Hell, D.~G. Schell, and J.~Stacho.
\newblock Dichotomy for tree-structured trigraph list homomorphism problems.
\newblock {\em Discrete Applied Mathematics}, 159(12):1217--1224, 2011.

\bibitem{federDAM154}
T.~Feder, P.~Hell, and K.~Tucker-Nally.
\newblock Digraph matrix partitions and trigraph homomorphisms.
\newblock {\em Discrete Appl. Math.}, 154(17):2458–2469, Nov. 2006.

\bibitem{FederVardi}
T.~Feder and M.~Y. Vardi.
\newblock The computational structure of monotone monadic {SNP} and constraint satisfaction: {a} study through {D}atalog and group theory.
\newblock {\em {SIAM} Journal on Computing}, 28:57--104, 1999.

\bibitem{golumbicJA19}
M.~Golumbic, H.~Kaplan, and R.~Shamir.
\newblock Graph {S}andwich {P}roblems.
\newblock {\em Journal of Algorithms}, 19:449--473, 1995.

\bibitem{habibCSR4}
M.~Habib and C.~Paul.
\newblock A survey of the algorithmic aspects of modular decomposition.
\newblock {\em Computer Science Review}, 4:41--59, 2010.

\bibitem{hellEJC35}
P.~Hell.
\newblock Graph partitions with prescribed patterns.
\newblock {\em European Journal of Combinatorics}, 35:335--353, 2014.

\bibitem{hellDM234}
P.~Hell, A.~Kostochka, A.~Raspaud, and E.~Sopena.
\newblock On nice graphs.
\newblock {\em Discrete Mathematics}, 234(1):39--51, 2001.

\bibitem{HellNesetril}
P.~Hell and J.~Ne\v{s}et\v{r}il.
\newblock On the complexity of {H}-coloring.
\newblock {\em Journal of Combinatorial Theory, Series B}, 48:92--110, 1990.

\bibitem{HNBook}
P.~Hell and J.~Ne\v{s}et\v{r}il.
\newblock {\em Graphs and Homomorphisms}.
\newblock Oxford University Press, Oxford, 2004.

\bibitem{Hodges}
W.~Hodges.
\newblock {\em A shorter model theory}.
\newblock Cambridge University Press, Cambridge, 1997.

\bibitem{kimKMJ63}
H.~Kim and M.~Siggers.
\newblock Towards a dichotomy for the list switch homomorphism problem for signed graphs.
\newblock {\em Kyungpook Mathematical Journal}, 63(3), 2023.

\bibitem{Maltsev-Cond}
M.~Kozik, A.~Krokhin, M.~Valeriote, and R.~Willard.
\newblock Characterizations of several {M}altsev conditions.
\newblock {\em Algebra universalis}, 73(3):205--224, 2015.

\bibitem{Kun}
G.~Kun.
\newblock Constraints, {MMSNP}, and expander relational structures.
\newblock {\em Combinatorica}, 33(3):335--347, 2013.

\bibitem{mahaved1995}
N.~V.~R. Mahadev and U.~N. Peled.
\newblock {\em Threshold graphs and related topics}, volume~56 of {\em Ann. Discrete Math.}
\newblock Amsterdam: Elsevier, 1995.

\bibitem{montejanoDAM158}
A.~Montejano, P.~Ochem, A.~Pinlou, A.~Raspaud, and Éric Sopena.
\newblock Homomorphisms of 2-edge-colored graphs.
\newblock {\em Discrete Applied Mathematics}, 158(12):1365--1379, 2010.
\newblock Traces from LAGOS’07 IV Latin American Algorithms, Graphs, and Optimization Symposium Puerto Varas - 2007.

\bibitem{naserasrJGT79}
R.~Naserasr, E.~Rollov\'a, and E.~Sopena.
\newblock Homomorphisms of signed graphs.
\newblock {\em Journal of Graph Theory}, 79(3):178--212, 2015.

\bibitem{ochemJGT85}
P.~Ochem, A.~Pinlou, and S.~Sen.
\newblock Homomorphisms of 2-edge-colored triangle-free planar graphs.
\newblock {\em Journal of Graph Theory}, 85(1):258--277, 2017.

\bibitem{Zhuk20}
D.~Zhuk.
\newblock A proof of the {CSP} dichotomy conjecture.
\newblock {\em Journal of the {ACM}}, 67(5):30:1--30:78, 2020.

\bibitem{ZhukFVConjecture}
D.~N. Zhuk.
\newblock A proof of {CSP} dichotomy conjecture.
\newblock In {\em 58th {IEEE} Annual Symposium on Foundations of Computer Science, {FOCS} 2017, {B}erkeley, {CA}, {USA}, {O}ctober 15-17}, pages 331--342, 2017.
\newblock https://arxiv.org/abs/1704.01914.

\end{thebibliography}

\appendix

\section{Four more examples of hereditary pp-constructions of $\bK_3$}
\label{ap:HH}

We apply Lemma~\ref{lem:siggers-power} to show that the $2$-edge-coloured
graphs from Figure~\ref{fig:4-element} (and their duals)
hereditarily pp-construct $\bK_3$.

\begin{figure}[ht!]
    \centering
    \begin{tikzpicture}
    
        \begin{scope}[scale=0.7]
            \node (L1) at (0,-1.6) {(3D)};
            \node (0) [vertex, label = left:{\scriptsize $0$}] at (90:1.3){};
            \node (1) [vertex, label = below:{\scriptsize $1$}] at (210:1.3){};
            \node (2) [vertex, label = below:{\scriptsize $2$}] at (330:1.3){};
            
            \draw [edge, red, dashed] (0) to [out=55,in=125, looseness=12] (0);
            \draw [edge, red, dashed] (1) to [out=55 + 90,in=125 + 90, looseness=12] (1);
            \draw [edge, blue] (2) to [out=55 -90,in=125 + -90, looseness=12] (2);
               
            \foreach \from/\to in {0/1, 1/2} 
                \draw [edge, blue] (\from) to [bend right = 20] (\to);
            \foreach \from/\to in {0/1,  0/2} 
                \draw [edge, red, dashed] (\from) to [bend left = 20] (\to);

        \end{scope}
        \begin{scope}[scale=0.7, xshift=5cm]
            \node (L1) at (0,-1.4) {(4B)};
            \node (0) [vertex, label = left:{\scriptsize $0$}] at (-1,2){};
            \node (1) [vertex, label = right:{\scriptsize $1$}] at (1,2){};
            \node (2) [vertex, label = left:{\scriptsize $2$}] at (-1,0){};
            \node (3) [vertex, label = right:{\scriptsize $3$}] at (1,0){};
            
            \draw [edge, red, dashed] (0) to [out=55,in=125, looseness=12] (0);
             \draw [edge, blue] (1) to [out=55,in=125, looseness=12] (1);
            \draw [edge,  red, dashed] (2) to [out=55 - 180,in=125 + -180, looseness=12] (2);
            \draw [edge,  red, dashed] (3) to [out=55 - 180,in=125 + -180, looseness=12] (3);
               
            \draw [edge, red, dashed] (2) to [bend left = 20] (3);
            \draw [edge, blue] (2) to [bend right = 20] (3);
            
            \foreach \from/\to in {0/1, 0/2} 
                \draw [edge, red, dashed] (\from) to  (\to);
            \foreach \from/\to in {1/2, 1/3, 0/3} 
                \draw [edge, blue] (\from) to  (\to);

        \end{scope}
    
        \begin{scope}[scale=0.7,xshift=10cm]
           \node (L1) at (0,-1.4) {(4C)};
            \node (0) [vertex, label = left:{\scriptsize $0$}] at (-1,2){};
            \node (1) [vertex, label = right:{\scriptsize $1$}] at (1,2){};
            \node (2) [vertex, label = left:{\scriptsize $2$}] at (-1,0){};
            \node (3) [vertex, label = right:{\scriptsize $3$}] at (1,0){};
            
            \draw [edge, red, dashed] (0) to [out=55,in=125, looseness=12] (0);
             \draw [edge, blue] (1) to [out=55,in=125, looseness=12] (1);
            \draw [edge,  red, dashed] (2) to [out=55 - 180,in=125 + -180, looseness=12] (2);
            \draw [edge,  red, dashed] (3) to [out=55 - 180,in=125 + -180, looseness=12] (3);

            \draw [edge, red, dashed] (2) to [bend left = 20] (3);
            \draw [edge, blue] (2) to [bend right = 20] (3);
            
            \foreach \from/\to in {1/2, 1/3, 0/3} 
                \draw [edge, red, dashed] (\from) to  (\to);
            \foreach \from/\to in {0/1, 0/2} 
                \draw [edge, blue] (\from) to  (\to);
        \end{scope}
    
        \begin{scope}[scale=0.7, xshift = 15cm]
            \node (L1) at (0,-1.4) {(4D)};
            \node (0) [vertex, label = left:{\scriptsize $0$}] at (-1,2){};
            \node (1) [vertex, label = right:{\scriptsize $1$}] at (1,2){};
            \node (2) [vertex, label = left:{\scriptsize $2$}] at (-1,0){};
            \node (3) [vertex, label = right:{\scriptsize $3$}] at (1,0){};
            
            \draw [edge, red, dashed] (0) to [out=55,in=125, looseness=12] (0);
             \draw [edge, red, dashed] (1) to [out=55,in=125, looseness=12] (1);
            \draw [edge,  red, dashed] (2) to [out=55 - 180,in=125 + -180, looseness=12] (2);
            \draw [edge,  blue] (3) to [out=55 - 180,in=125 + -180, looseness=12] (3);

             \draw [edge, red, dashed] (2) to [bend left = 20] (3);
            \draw [edge, blue] (2) to [bend right = 20] (3);
            
            \foreach \from/\to in {0/2, 0/3, 2/1} 
                \draw [edge, red, dashed] (\from) to  (\to);
            \foreach \from/\to in {0/1, 1/3} 
                \draw [edge, blue] (\from) to  (\to);
        \end{scope}

\end{tikzpicture}
\caption{Some minimal hereditarily hard reflexive complete $2$-edge-coloured on at most four vertices.}
\label{fig:4-element}
\end{figure}
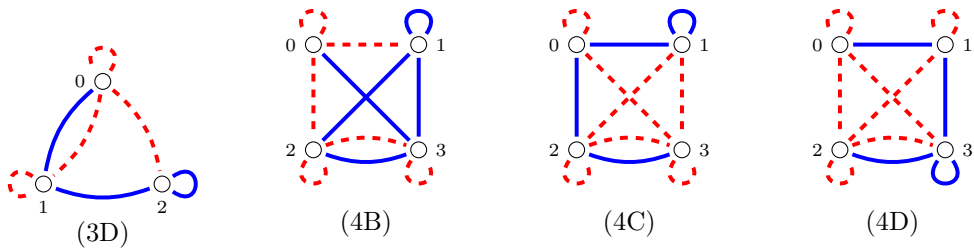

\begin{lemma}\label{lem:4-element-sig}
    If $\bH$ is a $2$-edge-coloured graph from  Figure~\ref{fig:4-element}, then 
    $\bH$ and $\overline \bH$ hereditarily pp-construct $\bK_3$.
\end{lemma}
\begin{proof}
    We proceed similarly to the proof of Lemma~\ref{lem:3-element-sig}: we show that
    for each $\bH\in\{\mathrm{3D,4B,4C,4D}\}$ the Siggers power $\Sig(\bH)$ has an $\ast$-odd-cycle,
    and then conclude by Lemma~\ref{lem:siggers-power}. The following is a depiction
    in $\Sig(\mathrm{3D})$ and in $\Sig(\mathrm{4B})$ where edges represent edges in $\mathbb H^4$, 
    colour classes represent $\sim_S$-equivalence classes, and to simplify notation we
    write $abcd$ instead of $(a,b,c,d)$.
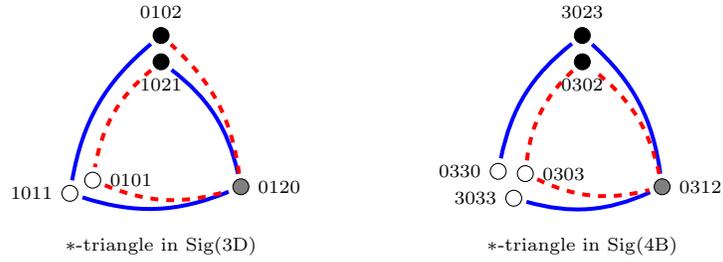
    \begin{center}
    \begin{tikzpicture} 

    \begin{scope}[scale=0.7]
            \node (L1) at (0,-2) {\scriptsize $\ast$-triangle in Sig(3D)};
            \node (0) [vertex, fill = black, label = below:{\scriptsize $1021$}] at (90:1.5){};
            \node (00) [vertex, fill = black,  label = above:{\scriptsize $0102$}] at (90:2){};
            \node (1) [vertex, label = left:{\scriptsize $1011$}] at (210:2){};
            \node (11) [vertex, label = right:{\scriptsize $0101$}] at (210:1.5){};
            \node (2) [vertex, fill = gray, label = right:{\scriptsize $0120$}] at (330:1.75){};
               
            \foreach \from/\to in {00/1, 1/2, 2/0} 
                \draw [edge, blue] (\from) to [bend right = 20] (\to);
            \foreach \from/\to in {0/11, 11/2, 2/00} 
                \draw [edge, red, dashed] (\from) to [bend right = 20] (\to);

        \end{scope}
        
        \begin{scope}[scale=0.7,xshift=8cm]
           \node (L1) at (0,-2) {\scriptsize $\ast$-triangle in $\Sig(\mathrm{4B})$};
            \node (0) [vertex, fill = black, label = below:{\scriptsize $0302$}] at (90:1.5){};
            \node (00) [vertex, fill = black,  label = above:{\scriptsize $3023$}] at (90:2){};
            \node (1) [vertex, label = right:{\scriptsize $0303$}] at (210:1.25){};
            \node (11) [vertex, label = left:{\scriptsize $3033$}] at (220:1.7){};
            \node (111) [vertex, label = left:{\scriptsize $0330$}] at (200:1.7){};
            \node (2) [vertex, fill = gray, label = right:{\scriptsize $0312$}] at (330:1.75){};
               
            \foreach \from/\to in {2/00, 11/2, 00/111} 
                \draw [edge, blue] (\from) to [bend right = 20] (\to);
            \foreach \from/\to in {0/1, 1/2, 2/0} 
                \draw [edge, red, dashed] (\from) to [bend right = 20] (\to);
        \end{scope}
\end{tikzpicture}
\end{center}

Finally, following the same conventions as in the previous illustration, we depict
$\ast$-cycles of length five in 
$\Sig(\mathrm{4C})$ and $\Sig(\mathrm{4D})$.
\begin{center}
    \begin{tikzpicture}
        \begin{scope}[scale=0.7]
            \node (L1) at (0,-2.5) {\scriptsize $\ast$-5-cycle in $\Sig(\mathrm{4C})$};
            \node (0) [vertex, fill = black, label = below:{\scriptsize $3123$}] at (90:1.5){};
            \node (00) [vertex, fill = black,  label = above:{\scriptsize $1312$}] at (90:2){};
            \node (1) [vertex, label = right:{\scriptsize $2021$}] at (162:1.5){};
            \node (11) [vertex, label = left:{\scriptsize $0210$}] at (162:2){};
            \node (22) [vertex, fill = gray, label = left:{\scriptsize $0130$}] at (234:2){};
            \node (33) [vertex, fill = lightgray, label = left:{\scriptsize $2022$}] at (306:1.5){};
            \node (333) [vertex, fill = lightgray, label = right:{\scriptsize $0220$}] at (306:2){};
            \node (44) [vertex, fill = magenta,  label = right:{\scriptsize $2032$}] at (18:2){};
               
            \foreach \from/\to in {00/11, 1/22,  22/33, 333/44, 44/0} 
                \draw [edge, blue] (\from) to [bend right = 20] (\to);
            \foreach \from/\to in {00/1, 11/22, 22/333,  33/44, 44/00} 
                \draw [edge, red, dashed] (\from) to [bend right = 20] (\to);       
        \end{scope}
    
        \begin{scope}[scale=0.7,xshift=7cm]
           \node (L1) at (0,-2.5) {\scriptsize $\ast$-5-cycle in $\Sig(\mathrm{4D})$};
            \node (00) [vertex, fill = black,  label = above:{\scriptsize $3130$}] at (90:2){};
            \node (1) [vertex, label = right:{\scriptsize $0103$}] at (162:1.4){};
            \node (11) [vertex, label = left:{\scriptsize $1031$}] at (162:2){};
            \node (22) [vertex, fill = gray, label = left:{\scriptsize $0120$}] at (234:2){};
            \node (33) [vertex, fill = lightgray, label = left:{\scriptsize $3133$}] at (306:1.5){};
            \node (333) [vertex, fill = lightgray, label = right:{\scriptsize $1331$}] at (306:2){};
            \node (4) [vertex, fill = magenta, label = left:{\scriptsize $0102$}] at (18:1.5){};
            \node (44) [vertex, fill = magenta,  label = right:{\scriptsize $1021$}] at (18:2){};
               
            \foreach \from/\to in {00/11, 11/22, 22/333, 33/44, 44/00} 
                \draw [edge, blue] (\from) to [bend right = 20] (\to);
            \foreach \from/\to in {00/1, 1/22, 22/33, 33/4, 4/00} 
                \draw [edge, red, dashed] (\from) to [bend right = 20] (\to);
        \end{scope}
\end{tikzpicture}
\end{center}
    
\end{proof}

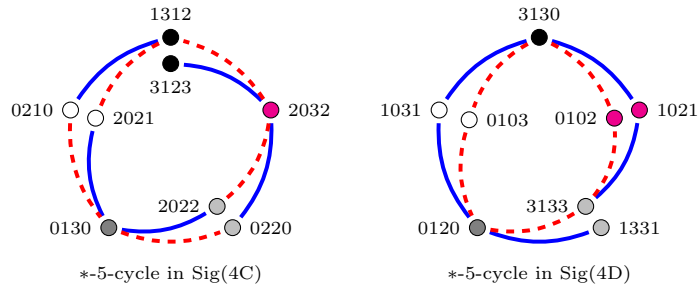
\end{document}